\PassOptionsToPackage{usenames,dvipsnames}{xcolor}
\documentclass[submission, Phys]{SciPost}
\hypersetup{
    colorlinks,
    linkcolor={red!50!black},
    citecolor={blue!50!black},
    urlcolor={blue!80!black}
}

\usepackage{physics}
\usepackage{tikz}
\usetikzlibrary{shapes.geometric, arrows.meta, positioning, calc, decorations.markings, decorations.pathreplacing,calligraphy}

\usepackage{authblk}
\usepackage{geometry}
\usepackage[english]{babel}
\usepackage[utf8]{inputenc}
\usepackage[T1]{fontenc}
\usepackage{indentfirst}
\usepackage{amsmath}
\usepackage{amssymb}
\usepackage{amsthm}
\usepackage{proof}
\usepackage{mathtools}
\usepackage{eufrak}
\usepackage{graphicx}
\usepackage{psfrag}
\usepackage{bbold}
\usepackage{mathrsfs}
\usepackage[normalem]{ulem}

\newtheorem{thm}{Theorem}

\newtheorem{lem}{Lemma}

\newcommand{\red}[1]{#1}

\renewcommand{\hbar}[1]{h_{\overline{#1}}}
\newcommand{\bbar}[1]{b_{\overline{#1}}}

\newcommand{\indepset}[2]{{\mathcal{S}}_{#1}^{(#2)}}
\newcommand{\indepsetwhole}[1]{{\mathcal{S}}_{#1}}
\newcommand{\stilde}[1]{{S}^{(#1)}}

\newcommand{\floor}[1]{\left\lfloor#1\right\rfloor}

\newcommand{\vfrz}[1]{h^{\mathrm{frz},#1}}
\newcommand{\vspo}{h^{\mathrm{spo}}}

\newcommand{\vrear}[1]{h^{\mathcal{R}, #1}}
\newcommand{\vleafdep}[1]{h^{\mathcal{L},{#1}}}
\newcommand{\vleafindep}{h^{\mathcal{L}}}
\newcommand{\leaves}[2]{\mathcal{L}_{#1}^{#2}}
\newcommand{\rearrangeable}[2]{\mathcal{R}_{#1}^{#2}}

\newcommand{\ccc}[2]{\mathcal{A}^{#1}_{#2}}
\newcommand{\ccl}[3]{\mathcal{B}^{#1, #2}_{#3}}
\newcommand{\cclwtotype}[2]{\mathcal{B}^{#1}_{#2}}
\newcommand{\cc}[3]{\mathcal{C}_{#3}^{#1,#2}}
\newcommand{\LopR}{L \oplus R}

\newcommand{\s}[1]{\mathbb{S}_{#1}}
\renewcommand{\a}[1]{\mathbb{A}_{#1}}
\newcommand{\C}[1]{\mathbb{C}_{#1}}
\newcommand{\Q}[2]{Q_{#1}^{(#2)}}

\usetikzlibrary{positioning}

\newcommand{\placevertex}[4]{
  \pgfmathsetmacro{\nodevalue}{int(#1 + #2)}
  \node[#4] (v#1) at (#1, #3) {$\nodevalue$};
}

\newcommand{\placevertexbar}[4]{
  \pgfmathsetmacro{\nodevalue}{int(#1 + #2)}
  \node[#4] (bar#1) at (#1, #3) {$\overline{\nodevalue}$};
}

\newcommand{\placevertexLR}[6]{
  \pgfmathsetmacro{\nodevalue}{int(#5)}
  \ifnum#6=0
    \node[#4] (#3) at (#1, #2) {$\nodevalue$};
  \else
    \node[#4] (#3) at (#1, #2) {$\overline{\nodevalue}$};
  \fi
}

\newcommand{\placevertexwithoutlabel}[4]{
  \node[#4] (#3) at (#1, #2) {};
}

\newcommand{\drawedge}[3]{
  \pgfmathsetmacro{\rowone}{int(substr("#1",1,1))}
  \pgfmathsetmacro{\colone}{int(substr("#1",3,10))}
  \pgfmathsetmacro{\rowtwo}{int(substr("#2",1,1))}
  \pgfmathsetmacro{\coltwo}{int(substr("#2",3,10))}

  \pgfmathsetmacro{\iscurved}{(abs(\rowone-\rowtwo)==2) ? 1 : 0}

  \ifnum\iscurved=1
    \pgfmathsetmacro{\coldiff}{\coltwo-\colone}
    \ifnum\coldiff<0
      \draw[#3, bend right=10] (#1) to (#2);
    \else
      \draw[#3, bend left=10] (#1) to (#2);
    \fi
  \else
    \draw[#3] (#1) -- (#2);
  \fi
}

\newcommand{\getvertexstyle}[5]{%
  \pgfmathsetmacro{\iscenter}{(#1==#3 && #2==#4) ? 1 : 0}%
  \ifnum\iscenter=1
    clawcenter%
  \else
    \pgfmathsetmacro{\isleaf}{0}%
    \foreach \leaf in #5 {%
      \pgfmathparse{(\leaf == 10*#1+#2) ? 1 : \isleaf}%
      \xdef\isleaf{\pgfmathresult}%
    }%
    \pgfmathsetmacro{\isleafint}{int(\isleaf)}%
    \ifnum\isleafint=1
      clawleaves%
    \else
      vertex%
    \fi
  \fi
}

\newcommand{\curveratio}{13.5}

\newcommand{\curveratioright}{10}

\tikzset{
	 triangular lattice small/.style={
		vertexminimal/.style={circle, minimum size=6mm, inner sep=0pt},
		vertex/.style={vertexminimal, draw},
		vertexleft/.style={vertexminimal,draw=MidnightBlue,line width=2.5pt},
		vertexright/.style={vertexminimal,draw=OliveGreen,line width=2.5pt},
		rearrangeablevertex/.style={vertex,  line width=2pt, draw=Apricot},
		frozeneven/.style={vertex, minimum size=5.25mm,rectangle,  line width=1pt},
		specialodd/.style={vertex, draw=CadetBlue,line width=2pt},
		clawcenter/.style={
			vertexminimal,
			text=white,
			fill=CadetBlue,
			font=\boldmath,
		},
		clawcenter_small/.style={
			circle, minimum size=4mm, inner sep=0pt,
			fill=CadetBlue,
		},
		clawleaves/.style={
			vertexminimal,
			fill=OliveGreen,
			font=\boldmath,
			text=white
		},
		clawleaves_small/.style={
			circle, minimum size=4mm, inner sep=0pt,
			fill=OliveGreen,
		},
		edge/.style={thick, draw=Gray},
			edged/.style={thick, draw, line width=1.25pt},
		edgefar/.style={dashed, draw=Gray, bend left=\curveratio},
			edgefar_right/.style={dashed, draw=Gray, bend right=\curveratioright},
		clawedge/.style={line width=3pt, draw=OliveGreen},
			extendedge/.style={line width=3pt, draw=RoyalBlue},
		curved/.style={bend left=\curveratio},
		curved_right/.style={bend right=\curveratioright},
		highlightededge/.style={line width=2pt, draw},
			scale=1,
			groupellipse/.style={dotted, thick},
			potensialedge/.style={dotted, ultra thick},
			vertexgreen/.style={vertex,draw=OliveGreen,line width=2pt},
			vertexblue/.style={vertex,draw=CadetBlue,line width=2pt},
		vs/.style={circle,draw = CadetBlue,minimum size=2mm,inner sep=1mm, line width=0.5mm}
	},
}

\newcommand{\drawpartdotted}[2]{
  \draw[edged, ultra thick] (#1) -- ($(#1)!0.7!(#2)$);
  \draw[edged, dotted, ultra thick] ($(#1)!0.3!(#2)$) -- (#2);
}

\newcommand{\frustrationgraph}[3]{

\pgfmathsetmacro{\minvalone}{#1-1}
\pgfmathsetmacro{\minvaltwo}{#1}
\pgfmathsetmacro{\minvalthree}{#1+1}

\pgfmathsetmacro{\maxvalone}{#2+1}
\pgfmathsetmacro{\maxvaltwo}{#2+2}
\pgfmathsetmacro{\maxvalthree}{#2+3}

\pgfmathsetmacro{\start}{int(#1/2)}
\pgfmathsetmacro{\last}{int(#2/2)+1}

\pgfmathsetmacro{\lastminusone}{int(#2/2)}
\pgfmathsetmacro{\lastplussone}{int(#2/2)+2}

\def\offsetinfunc{#3}

\pgfmathsetmacro{\yrowone}{0}
\pgfmathsetmacro{\yrowtwo}{-1.732}
\pgfmathsetmacro{\yrowthree}{-3.464}

\foreach \i in {\start,...,\last} {
  \pgfmathsetmacro{\firstv}{int(2*\i-1)}
	\pgfmathsetmacro{\secondv}{int(2*\i)}
	\pgfmathsetmacro{\thirdv}{int(2*\i+1)}
	\placevertex{\firstv}{\offsetinfunc}{\yrowone}{vertex}
	\placevertex{\secondv}{\offsetinfunc}{\yrowtwo}{vertexminimal}
	\placevertexbar{\thirdv}{\offsetinfunc}{\yrowthree}{vertex}
}

\foreach \i in {\start,...,\lastminusone} {
  \pgfmathsetmacro{\firstv}{int(2*\i-1)}
  \pgfmathsetmacro{\firstvnext}{int(2*\i+1)}
	\pgfmathsetmacro{\secondv}{int(2*\i)}
  \pgfmathsetmacro{\secondvnext}{int(2*\i+2)}
	\pgfmathsetmacro{\thirdv}{int(2*\i+1)}
  \pgfmathsetmacro{\thirdvnext}{int(2*\i+3)}
  \draw[edge] (v\firstv) -- (v\firstvnext);
	\draw[edge] (v\secondv) -- (v\secondvnext);
	\draw[edge] (bar\thirdv) -- (bar\thirdvnext);
}

\foreach \i in {\start,...,\last} {
  \pgfmathsetmacro{\firstv}{int(2*\i-1)}
	\pgfmathsetmacro{\secondv}{int(2*\i)}
	\pgfmathsetmacro{\thirdv}{int(2*\i+1)}
  \draw[edge] (v\firstv) -- (v\secondv) -- (bar\thirdv);
}

\foreach \i in {\start,...,\lastplussone} {
  \pgfmathsetmacro{\firstv}{int(2*\i-1)}
	\pgfmathsetmacro{\secondv}{int(2*\i-2)}
	\pgfmathsetmacro{\thirdv}{int(2*\i-3)}
	\ifnum\secondv>\minvaltwo
		\ifnum\firstv<\maxvalone
    	\draw[edge] (v\firstv) -- (v\secondv);
		\fi
  \fi
	\ifnum\thirdv>\minvalthree
		\ifnum\secondv<\maxvaltwo
    	\draw[edge] (v\secondv) -- (bar\thirdv);
		\fi
  \fi
}

\foreach \i in {\start,...,\last} {
  \pgfmathsetmacro{\firstv}{int(2*\i-1)}
	\pgfmathsetmacro{\thirdv}{int(2*\i-1)}
	\pgfmathsetmacro{\thirdvnext}{int(2*\i+1)}
	\pgfmathsetmacro{\thirdvnextnext}{int(2*\i+3)}
	\pgfmathsetmacro{\thirdvnextnextnext}{int(2*\i+5)}
		\ifnum\thirdv<\maxvalthree
			\ifnum\thirdv>\minvalthree
    		\draw[edgefar, curved_right] (v\firstv) to (bar\thirdv);
			\fi
		\fi

		\ifnum\thirdvnext<\maxvalthree
			\ifnum\thirdvnext>\minvalthree
    		\draw[edgefar] (v\firstv) to (bar\thirdvnext);
			\fi
		\fi

		\ifnum\thirdvnextnext<\maxvalthree
			\ifnum\thirdvnextnext>\minvalthree
    		\draw[edgefar] (v\firstv) to (bar\thirdvnextnext);
			\fi
		\fi

		\ifnum\thirdvnextnextnext<\maxvalthree
			\ifnum\thirdvnextnextnext>\minvalthree
    		\draw[edgefar] (v\firstv) to (bar\thirdvnextnextnext);
			\fi
		\fi

}

}

\begin{document}

\begin{center}{\Large \textbf{
			A free fermions in disguise model with claws
		}}\end{center}

% TODO: write the author list here. Use first name (+ other initials) + surname format.
% Separate subsequent authors by a comma, omit comma and use "and" for the last author.
% Mark the corresponding author with a superscript star.
\begin{center}
	Kohei Fukai\textsuperscript{1,2$\star$},
	Istv\'an Vona\textsuperscript{3,4$\dag$} and
	Bal\'azs Pozsgay\textsuperscript{3$\ddag$}
\end{center}

\title{A free fermions in disguise model with claws}
\author[1,2]{Kohei Fukai \thanks{kohei.fukai@riken.jp}}
\author[3,4]{Istv\'an Vona \thanks{vona.istvan@wigner.hun-ren.hu}}
\author[3]{Bal\'azs Pozsgay \thanks{pozsgay.balazs@ttk.elte.hu}}

\begin{center}
	{\bf 1} Department of Physics, Graduate School of Science, \authorcr The University of Tokyo, \authorcr7-3-1, Hongo, Bunkyo-ku, Tokyo, 113-0033, Japan
	\\
	{\bf 2} RIKEN iTHEMS, Wako, Saitama 351-0198, Japan
	\\
	{\bf 3} MTA-ELTE ``Momentum'' Integrable Quantum Dynamics Research Group, \authorcr ELTE E\"otv\"os Lor\'and University,\authorcr P\'azm\'any P. s\'et\'any 1/A, H-1117 Budapest, Hungary
	\\
	{\bf 4}  Holographic Quantum Field Theory Research Group,\authorcr
	HUN-REN Wigner Research Centre for Physics, Budapest, Hungary
	\\
	${}^\star$ {\small \sf kohei.fukai@riken.jp},
	${}^\dag$ {\small \sf vona.istvan@wigner.hun-ren.hu},
	${}^\ddag$ {\small \sf pozsgay.balazs@ttk.elte.hu}
\end{center}

\begin{center}
	\today
\end{center}

% For convenience during refereeing (optional),
% you can turn on line numbers by uncommenting the next line:
%\linenumbers

\section*{Abstract}
 {\bf
  \noindent
  Recently, several spin chain models have been discovered that admit solutions in terms of "free fermions in disguise."
  A graph-theoretical treatment of such models was also established, giving sufficient conditions for free fermionic solvability.
  These conditions involve a particular property of the so-called frustration graph of the Hamiltonian, namely that it must be claw-free.
  Additionally, one set of sufficient conditions also requires the absence of so-called even holes.
  In this paper, we present a model with disguised free fermions where the frustration graph contains both claws and even holes.
  Special relations between coupling constants ensure that the free fermionic property still holds.
  \red{Notably, the central elements associated with the even holes can be removed by fixing the gauge, revealing our model to be an integrable deformation within the original algebra of free fermions in disguise.}
  The transfer matrix of this model can be factorized in a special case, thereby proving the conjectured free fermionic nature of a special quantum circuit published recently by two of the present authors.
   }

\vspace{10pt}
\noindent\rule{\textwidth}{1pt}
\tableofcontents\thispagestyle{fancy}
\noindent\rule{\textwidth}{1pt}
\vspace{10pt}

\section{Introduction}
The quest for exactly solvable quantum many-body systems has driven much of the progress in theoretical physics.
Among the most tractable classes are systems mappable to free fermions, where the absence of interactions permits exact computation of many physical observables.
The Jordan-Wigner transformation \cite{jordan-wigner} stands as the archetypal example, converting local spin variables into fermionic operators and rendering numerous spin chain models exactly solvable \cite{XX-original,Schulz-Mattis-Lieb}.

Recent advances have expanded our understanding of free fermionic solvability beyond the traditional Jordan-Wigner paradigm.
In~\cite{fendley-fermions-in-disguise}, Fendley discovered the ``free fermions in disguise'' (FFD) model, a spin chain that appears to have genuine four-fermion interactions, yet possesses hidden free fermionic structures.
The model exhibits the remarkable property of remaining free fermionic for arbitrary coupling constants, including fully disordered configurations.

Subsequently, graph-theoretical frameworks were established to characterize free fermionic solvability systematically.
Two independent works \cite{chapman-jw} and \cite{japan-free-fermion-JW} first derived graph-theoretical criteria for generalized Jordan-Wigner transformations.
Afterwards, a comprehensive framework was established in \cite{fermions-behind-the-disguise,unified-graph-th}, which unified both Jordan-Wigner and FFD models.
Extensions to parafermionic commutation relations were treated in \cite{alcaraz-medium-fermion-1,alcaraz-medium-fermion-2,free-parafermionic-graphs}.

Importantly, the papers \cite{fermions-behind-the-disguise,unified-graph-th} derived sufficient conditions for free fermionic solvability that apply for all possible choices of the coupling constants.
However, these conditions are not necessary, as the original works themselves noted~\cite{fermions-behind-the-disguise,unified-graph-th}.
This leaves room for models that do not satisfy all of these conditions, while maintaining free fermionic structures through special selections of the coupling constants.
One of the sufficient conditions is the so-called ``claw-free'' property of the frustration graph.
A concrete example of a free fermionic model violating this condition was given in \cite{sajat-FP-model}, which achieved solvability through a particular extension of the FFD algebra.

In this work, we present a further example of a free fermionic model ``with claws''.
Most notably, our model also contains ``even holes'', which pose additional algebraic challenges for constructing free fermionic solutions~\cite{unified-graph-th}.
\red{Around each even hole, the product of the participating operators is central and generates a ``gauge symmetry'', whose effect can be removed by fixing the gauge.}
\red{The model of \cite{sajat-FP-model} mentioned above in fact also contains even holes; there, however, the gauge is fixed from the outset, so the even-hole structure remains implicit.}
\red{In the present work, we treat this even-hole gauge structure explicitly and show that gauge fixing yields a nontrivial integrable deformation of Fendley's FFD Hamiltonian~\cite{fendley-fermions-in-disguise}.}

The model Hamiltonian was inspired by the recent work \cite{sajat-floquet} of two of the present authors, where selected quantum circuits were conjectured to have free fermionic solvability.
We prove one of these conjectures by showing that the transfer matrix of our model becomes proportional to the quantum circuit in question for a special choice of the spectral parameter.

The solvability is achieved through carefully constructed algebraic relations among the coupling constants that compensate for the graph-theoretical obstructions.
We derive these relations by introducing an auxiliary extension of the FFD algebra subject to additional algebraic conditions.
\red{After fixing the gauge, this extended algebra can be realized in terms of Fendley's original FFD generators, so that the resulting Hamiltonian lies within the original FFD algebra.}
We successfully solve the inverse problem by expressing local operators in terms of free fermionic modes, following the methodology of \cite{sajat-ffd-corr}.
We further compute real-time correlation functions for selected operators.
The number of fermionic eigenmodes in our model is the same as in the original FFD model, resulting in an exponentially large symmetry algebra analogous to that recently characterized for the FFD case in \cite{eric-lorenzo-ffd-1}.

The structure of this paper is as follows.
In Sec.~\ref{sec:H}, we review previous progress on the FFD and introduce a new free fermionic Hamiltonian.
In Sec.~\ref{sec:extended-FFD}, we propose an extension of the FFD algebra and derive the new Hamiltonian from this extended algebra.
In Sec.~\ref{sec:charges}, we provide a family of commuting conserved charges.
We introduce the transfer matrix in Sec.~\ref{sec:transfer} and the fermionic eigenmodes in Sec.~\ref{sec:fermion}.
In Sec.~\ref{sec:corr}, we solve the inverse problem and compute correlation functions.
In Sec.~\ref{sec:fact}, we discuss the factorization of the transfer matrix, and finally in Sec.~\ref{sec:concl}, we present our conclusions.
Certain technical parts of the proofs are presented in the Appendices.

\section{The FFD algebra}
\label{sec:H}
In this section, we first review the graph-theoretical framework of the FFD model~\cite{fermions-behind-the-disguise} and then introduce our spin chain Hamiltonian, explaining how it extends Fendley's original Hamiltonian~\cite{fendley-fermions-in-disguise}.
In later sections, we present the frustration graph for our FFD Hamiltonian, which contains both claws and even holes, and prove that the Hamiltonian is integrable and can be solved using free fermionic operators.

\subsection{Review of the FFD algebra}
We first review the algebra introduced by Fendley in~\cite{fendley-fermions-in-disguise}, which we dub FFD algebra.
This algebra is made of the generators $\qty{h_j}_{j=1}^{M}$, where $M\ge 1$ is some integer.
The squares of the generators are scalars\footnote{Some of the previous work on the topic used a different convention, namely to have $h_m^2=1$ and to add the coupling constants separately to the Hamiltonian.}:
\begin{align}
	h_m^2 = b_m^2
	\,,
\end{align}
where $b_m$ are arbitrary coupling constants.
The generators are anticommuting when the indices are neighboring or next to neighboring, and commuting otherwise:
\begin{align}
	h_m h_{m+1} & = - h_{m+1} h_{m}                   \\
	h_m h_{m+2} & = - h_{m+2} h_{m}                   \\
	h_m h_{l}   & = h_{l} h_{m} \quad(\abs{l-m}>2)\,.
\end{align}
There is no periodicity condition for the indices, therefore the algebra naturally describes models with open boundary conditions.
The FFD algebra is a particular case of a generalized Clifford algebra.
\begin{figure}[tbp]
	\centering
	\begin{tikzpicture}[triangular lattice small, vs/.style={circle,draw = CadetBlue,minimum size=2mm,inner sep=1mm, line width=0.5mm}]

		\pgfmathsetmacro{\yrowone}{0}
		\pgfmathsetmacro{\yrowtwo}{-1.732}
		\pgfmathsetmacro{\yrowthree}{-3.464}

		% First row vertices (odd numbers)
		\foreach \i in {1,...,7} {
				\pgfmathsetmacro{\v}{int(2*\i-1)}
				\pgfmathsetmacro{\vlabel}{int(2*\i-1)}
				\node[vs] (v\v) at (\v, \yrowone) {};
				\node[above = 1mm of v\v] {$h_{\vlabel}$};
			}

		% Second row vertices (even numbers)
		\foreach \i in {1,...,6} {
				\pgfmathsetmacro{\v}{int(2*\i)}
				\pgfmathsetmacro{\vlabel}{int(2*\i)}
				\node[vs] (v\v) at (\v, \yrowtwo) {};
				\node[below = 1mm of v\v] {$h_{\vlabel}$};
			}

		\foreach \i in {1,...,11} {
				\pgfmathsetmacro{\inext}{int(\i+2)}
				\draw[edge] (v\i) -- (v\inext);
			}
		\foreach \i in {1,...,12} {
				\pgfmathsetmacro{\inext}{int(\i+1)}
				\draw[edge] (v\i) -- (v\inext);
			}

	\end{tikzpicture}

	\caption{Frustration graph for the original FFD model ($M=13$ case).
		Each vertex represents a generator of the FFD algebra.
		Vertices are connected by an edge when the corresponding generators anticommute.}
	\label{fig:ffd-frustration-graph}
\end{figure}
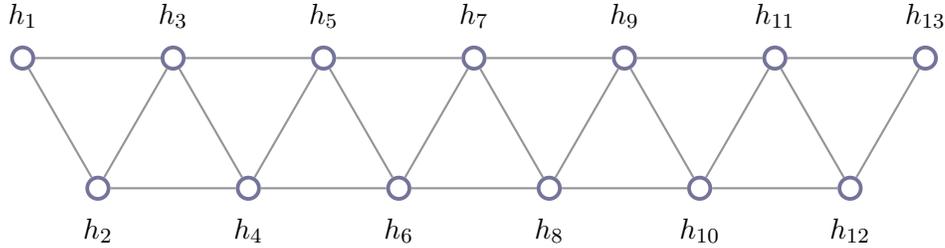
%%%%%%%%%%%%%%%%%%5
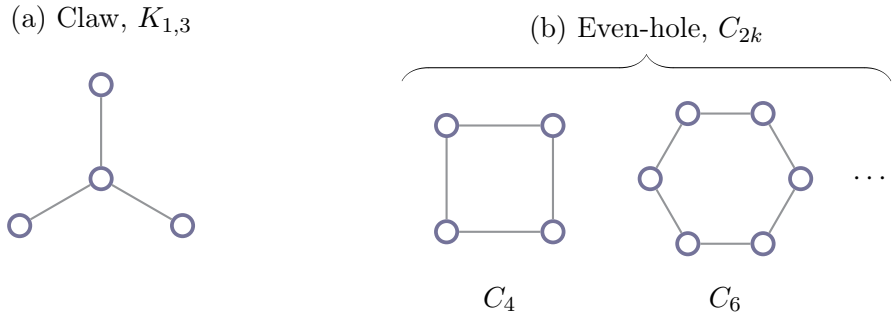
\begin{figure}
	\centering
	\begin{tikzpicture}[
			node/.style={circle,draw = CadetBlue,minimum size=2mm,inner sep=1mm, line width=0.5mm}
		]
		% (b) Claw
		\begin{scope}[xshift=0]
			% Clawの上にキャプション
			\node at (0,2.1) {(a) Claw, $K_{1,3}$};

			% Clawの中心
			\node[node] (b1) at (0,0) {};
			% Clawの3つの葉
			\node[node] (b2) at ({1.25*cos(90)}, {1.25*sin(90)}) {};
			\node[node] (b3) at ({1.25*cos(210)}, {1.25*sin(210)}) {};
			\node[node] (b4) at ({1.25*cos(-30)}, {1.25*sin(-30)}) {};
			% Clawの辺
			\draw[thick,draw = Gray] (b1) -- (b2);
			\draw[thick,draw = Gray] (b1) -- (b3);
			\draw[thick,draw = Gray] (b1) -- (b4);
			% ラベル
		\end{scope}

		% ClawとEven-holeを分ける線
		%\draw[thick] (3.7, -2) -- (3.7, 2.75);

		% (c) Even Holes (C4とC6)
		\begin{scope}[xshift=7cm]
			% Even-holeの上に波括弧（上下逆に）とキャプション
			\draw[decorate,decoration={calligraphic brace,amplitude=10pt,mirror}]
			(3.5,1.3) -- (-3,1.3) node[midway,above=10pt] {(b) Even-hole, $C_{2k}$};

			% C4（正方形）
			\begin{scope}[xshift=-1.7cm]
				\foreach \i in {1,...,4} {
						\node[node] (c\i) at ({cos(90*\i+45)}, {sin(90*\i+45)}) {};
					}
				\draw[thick,draw = Gray] (c1) -- (c2) -- (c3) -- (c4) -- (c1);
				\node at (0,-1.6) {$C_{4}$};
			\end{scope}

			% C6（六角形）
			\begin{scope}[xshift=1.3cm]
				\foreach \i in {1,...,6} {
						\node[node] (d\i) at ({1.*cos(60*\i)}, {1.*sin(60*\i)}) {};
					}
				\draw[thick,draw = Gray] (d1) -- (d2) -- (d3) -- (d4) -- (d5) -- (d6) -- (d1);
				\node at (0,-1.6) {$C_{6}$};
			\end{scope}

			% C6の横に \cdots を追加
			\node at (3.25,0) {$\cdots$};

		\end{scope}
	\end{tikzpicture}
	\caption{
		Forbidden structures in a frustration graph for free fermionic solutions of~\cite{fermions-behind-the-disguise}: (a) the claw $K_{1,3}$ and (b) an even hole $C_{2k}$.
		A graph is claw-free or even-hole-free if none of its induced subgraphs contain a claw or an even hole.
		The frustration graph in Figure~\ref{fig:ffd-frustration-graph} is (even-hole, claw)-free.
	}
	\label{fig:even-hole-claw}
\end{figure}

The standard presentation of the FFD algebra can be given for a spin-1/2 chain as follows.
Let the Pauli matrix acting on the $j$-th site be denoted by $\sigma_{j}^{\alpha}$, where $\alpha \in \{x,y,z\}$ indicates the component of the Pauli matrix.
Then we have
\begin{align}
	h_m = b_m \sigma_{m}^{z} \sigma_{m+1}^{z} \sigma_{m+2}^{x}.
	\label{eq:zzx-rep}
\end{align}
For other presentations, see \cite{fendley-fermions-in-disguise,sajat-floquet}.

The original FFD Hamiltonian of Fendley~\cite{fendley-fermions-in-disguise} reads simply
\begin{equation}
	H_{\mathrm{FFD}}=\sum_{m=1}^M h_m,
	\label{eq:original-FFD-Hamiltonian}
\end{equation}
where the Hamiltonian is the sum of all FFD generators.
The higher-order charges are constructed from products of FFD generators~\cite{fendley-fermions-in-disguise}.

We next explain the frustration graph.
Generally, the frustration graph is defined for spin chain Hamiltonians where the individual terms either commute or anticommute, and this condition is fulfilled if the Hamiltonian is expressed as a sum of products of Pauli matrices.
Each term in the Hamiltonian corresponds to a vertex in the graph, and two vertices are connected if and only if the two terms in question anticommute.
For an extension to the more general case of parafermionic commutation relations, see~\cite{free-parafermionic-graphs,para-frustration-graph-2}.

The frustration graph for the FFD algebra is shown in Figure~\ref{fig:ffd-frustration-graph}.
%The vertices represent the generators of the FFD algebra, and those which anticommute are connected by edges.
This graph has the property of being \emph{(even-hole, claw)-free}~\cite{fermions-behind-the-disguise}.
The structures of even holes and the claw are illustrated in Figure~\ref{fig:even-hole-claw}.
The seminal result in~\cite{fermions-behind-the-disguise} is that any Hamiltonian with a frustration graph that is (even-hole, claw)-free can be solved by free fermions in disguise~\cite{fendley-fermions-in-disguise}.
Thus, we can solve the FFD Hamiltonian~\eqref{eq:original-FFD-Hamiltonian} using these methods.

In~\cite{unified-graph-th}, the framework of (even-hole, claw)-free graphs was generalized to the \emph{simplicial, claw-free} case, which can admit even holes in the frustration graph.
\red{Nevertheless, the free-fermionic solvability beyond the claw-free setting remains much less developed, especially when claws and even holes coexist.}

\red{In the following, we present a new hidden free-fermion model whose frustration graph contains both claws and even holes and which can be constructed essentially within Fendley's original FFD algebra.}

\subsection{A new integrable deformation of the FFD Hamiltonian}
We propose a new integrable Hamiltonian in terms of the FFD algebra:
\begin{equation}
	H_M = \sum_{m=1}^{M} h_m + \sum_{m=2}^{\floor{M/2}}  \beta_{2m-3} h_{2m-2} h_{2m-3} h_{2m}\,,
	\label{eq:new-model}
\end{equation}
where $\{\beta_1, \beta_3, \beta_5,\ldots\}$ are additional coupling constants that satisfy the relations
\begin{align}
	\beta_{2m+1}
	=
	\frac{\beta_{2m-1}}{b_{2m-2}^2\beta_{2m-3} - b_{2m+2}^2 \beta_{2m-1} + 1}
	\quad(m \ge 2)
	\label{eq:beta_relation}
	\,.
\end{align}
The initial values for the recursion, $\beta_{1}$ and $\beta_{3}$ can be chosen arbitrarily.
The ordering of the triple product in~\eqref{eq:new-model} is adopted for the purpose of simplification presented below.

The new model~\eqref{eq:new-model} can be seen as a deformation of the original FFD model~\cite{fendley-fermions-in-disguise}.
The first term in~\eqref{eq:new-model} is equal to the Hamiltonian of the original FFD model, and the second term is the new interaction term, which retains the free fermionic solvability.
The original FFD Hamiltonian~\eqref{eq:original-FFD-Hamiltonian} is reproduced in the limit $\beta_{2m-1}\equiv 0$.

In the next section, we will show the frustration graph for the Hamiltonian~\eqref{eq:new-model} and demonstrate that it contains both claws and even holes.
We will also derive this new integrable Hamiltonian from an algebraic extension of the original FFD algebra.

\section{The extended FFD algebra \label{sec:extended-FFD}}
\red{In this section, we introduce the extended FFD algebra, the resulting model, and its frustration graph.}
\red{The proof of its free-fermionic solvability is given later: the conserved charges are constructed in Section~\ref{sec:charges}, the transfer matrix is introduced in Section~\ref{sec:transfer}, and the fermionic eigenmodes are constructed in Section~\ref{sec:fermion}.}
We first introduce additional generators and the frustration graph for the extended algebra.
\red{We then impose an algebraic relation among them, whose role is crucial in proving the conservation law in Section~\ref{sec:charges} and the mutual commutativity of the charges in Appendix~\ref{app:mutual-commutativity}.}
Finally, we reproduce the new Hamiltonian~\eqref{eq:new-model} as a special case.
We consider the case where the number of generators of the original FFD algebra is odd: $M = 2M' + 1$.
However, our argument can be extended to the even $M$ case with $M = 2M'$.
%%%%%%%%%%%%%%%%%%%%%%%%
%%%%%%%%%%%%%%%%%%%%%%%%
%%%%%%%%%%%%%%%%%%%%%%%%
\begin{figure}[tbp]
	\centering
	\begin{tikzpicture}[triangular lattice small, extendededge/.style={draw=RoyalBlue, bend right=\curveratio, line width=1.25pt}, extendededgeright/.style={draw=RoyalBlue, bend right=\curveratioright, line width=1.25pt}]

		\pgfmathsetmacro{\yrowone}{0}
		\pgfmathsetmacro{\yrowtwo}{-1.732}
		\pgfmathsetmacro{\yrowthree}{-3.464}

		% First row vertices (odd numbers)
		\foreach \i in {1,...,7} {
				\pgfmathsetmacro{\v}{int(2*\i-1)}
				\pgfmathsetmacro{\vlabel}{int(2*\i-1)}
				\node[vs] (v\v) at (\v, \yrowone) {};
				\node[vs] (bar\v) at (\v, \yrowthree) {};
			}

		% Second row vertices (even numbers)
		\foreach \i in {1,...,6} {
				\pgfmathsetmacro{\v}{int(2*\i)}
				\pgfmathsetmacro{\vlabel}{int(2*\i)}
				\node[vs] (v\v) at (\v, \yrowtwo) {};
			}
		\foreach \i in {1,...,11} {
				\pgfmathsetmacro{\inext}{int(\i+2)}
				\draw[edge] (v\i) -- (v\inext);
			}
		\foreach \i in {1,...,12} {
				\pgfmathsetmacro{\inext}{int(\i+1)}
				\draw[edge] (v\i) -- (v\inext);
			}

		\foreach \i in {1,...,6} {
				\pgfmathsetmacro{\v}{int(2*\i-1)}
				\pgfmathsetmacro{\vnext}{int(2*\i+1)}
				\pgfmathsetmacro{\even}{int(2*\i)}
				\draw[edge] (bar\v) -- (bar\vnext);
				\draw[edge] (v\even) -- (bar\vnext);
				\draw[edge] (v\even) -- (bar\v);
			}

		\foreach \i in {1,...,7} {
				\pgfmathsetmacro{\v}{int(2*\i-1)}
				\draw[edgefar, curved_right] (v\v) to (bar\v);
			}

		\foreach \i in {1,...,6} {
				\pgfmathsetmacro{\v}{int(2*\i-1)}
				\pgfmathsetmacro{\vnext}{int(2*\i+1)}
				\draw[edgefar] (v\v) to (bar\vnext);
			}

		\foreach \i in {1,...,5} {
				\pgfmathsetmacro{\v}{int(2*\i-1)}
				\pgfmathsetmacro{\vnext}{int(2*\i+3)}
				\draw[edgefar] (v\v) to (bar\vnext);
			}

		\foreach \i in {1,...,4} {
				\pgfmathsetmacro{\v}{int(2*\i-1)}
				\pgfmathsetmacro{\vnext}{int(2*\i+5)}
				\draw[edgefar] (v\v) to (bar\vnext);
			}

		% First row vertices (odd numbers)
		\foreach \i in {1,...,7} {
				\pgfmathsetmacro{\v}{int(2*\i-1)}
				\pgfmathsetmacro{\vlabel}{int(2*\i-1)}
				\node[above = 1mm of v\v] {$h_{\vlabel}$};
				\node[below = 1mm of bar\v] {$\hbar{\vlabel}$};
			}

		% Second row vertices (even numbers)
		\foreach \i in {1,...,6} {
				\pgfmathsetmacro{\v}{int(2*\i)}
				\pgfmathsetmacro{\vlabel}{int(2*\i)}
				\node[above left = 0.01mm and 0.1mm of v\v] {\small $h_{\vlabel}$};
			}

		\draw[extendededgeright] (v9) to (bar9);
		\draw[extendededge] (bar9) to (v7);
		\draw[extendededge] (bar9) to (v5);
		\draw[extendededge] (bar9) to (v3);

	\end{tikzpicture}

	\caption{
		Frustration graph for the extended FFD algebra $G_{M=13}$.
		Vertices are connected by an edge when the corresponding generators anticommute.
		The edges between $\{h_m\}$ and $\{\hbar{m'}\}$ are represented by dotted lines.
		To highlight the basic structure, we have highlighted the dotted edges from $\hbar{9}$ with solid blue lines.
		The other dotted lines are obtained by translation.
	}
	\label{fig:extended-ffd-frustration-graph}
\end{figure}
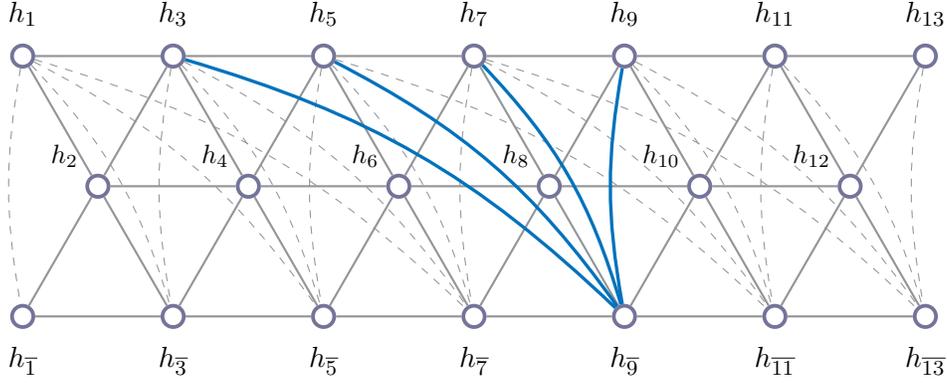
%%%%%%%%%%%%%%%%%%%%%%%%%%%%5
\begin{figure}[tbp]
	\centering
	\begin{tikzpicture}[triangular lattice small]
		% Define common y-coordinates
		\pgfmathsetmacro{\yrowone}{0}
		\pgfmathsetmacro{\yrowtwo}{-1.732}  % -sqrt(3)
		\pgfmathsetmacro{\yrowthree}{-3.464} % -2*sqrt(3)
		\pgfmathsetmacro{\offset}{-2}

		\def\captionx{7}
		\def\captiony{-5}
		\def\shifty{-6.5}

		% Type A claw
		\begin{scope}[shift={(0,0)}]
			\foreach \i in {1,...,6} {
					\pgfmathsetmacro{\v}{int(2*\i-1)}
					\pgfmathsetmacro{\vlabel}{int(2*\i-1)}
					\node[vs] (v\v) at (\v, \yrowone) {};
				}

			\foreach \i in {2,...,7} {
					\pgfmathsetmacro{\v}{int(2*\i-1)}
					\pgfmathsetmacro{\vlabel}{int(2*\i-1)}
					\node[vs] (bar\v) at (\v, \yrowthree) {};
				}

			% Second row vertices (even numbers)
			\foreach \i in {1,...,6} {
					\pgfmathsetmacro{\v}{int(2*\i)}
					\pgfmathsetmacro{\vlabel}{int(2*\i)}
					\node[vs] (v\v) at (\v, \yrowtwo) {};
				}

			% Horizontal edges
			\foreach \i in {1,...,10} {
					\pgfmathsetmacro{\inext}{int(\i+2)}
					\draw[edge] (v\i) -- (v\inext);
				}
			\foreach \i in {1,...,11} {
					\pgfmathsetmacro{\inext}{int(\i+1)}
					\draw[edge] (v\i) -- (v\inext);
				}

			% Diagonal edges
			\foreach \i in {1,...,6} {
					\pgfmathsetmacro{\v}{int(2*\i-1)}
					\pgfmathsetmacro{\vnext}{int(2*\i+1)}
					\pgfmathsetmacro{\even}{int(2*\i)}
					\draw[edge] (v\even) -- (bar\vnext);
				}

			\foreach \i in {2,...,6} {
					\pgfmathsetmacro{\v}{int(2*\i-1)}
					\pgfmathsetmacro{\vnext}{int(2*\i+1)}
					\pgfmathsetmacro{\even}{int(2*\i)}
					\draw[edge] (v\even) -- (bar\v);
				}

			\foreach \i in {2,...,6} {
					\pgfmathsetmacro{\v}{int(2*\i-1)}
					\pgfmathsetmacro{\vnext}{int(2*\i+1)}
					\pgfmathsetmacro{\even}{int(2*\i)}
					\draw[edge] (bar\v) -- (bar\vnext);
				}

			% Curved edges
			\foreach \i in {2,...,6} {
					\pgfmathsetmacro{\v}{int(2*\i-1)}
					\draw[edgefar, curved_right] (v\v) to (bar\v);
				}

			\foreach \i in {1,...,6} {
					\pgfmathsetmacro{\v}{int(2*\i-1)}
					\pgfmathsetmacro{\vnext}{int(2*\i+1)}
					\draw[edgefar] (v\v) to (bar\vnext);
				}

			\foreach \i in {1,...,5} {
					\pgfmathsetmacro{\v}{int(2*\i-1)}
					\pgfmathsetmacro{\vnext}{int(2*\i+3)}
					\draw[edgefar] (v\v) to (bar\vnext);
				}

			\foreach \i in {1,...,4} {
					\pgfmathsetmacro{\v}{int(2*\i-1)}
					\pgfmathsetmacro{\vnext}{int(2*\i+5)}
					\draw[edgefar] (v\v) to (bar\vnext);
				}

			% Left extensions
			\foreach \v in {v1, v2, bar3}{
					\coordinate[left=5.5mm of \v] (l\v);
					\draw[edge] (l\v) -- (\v);
					\coordinate[left=5.5mm of l\v] (ll\v);
					\draw[edge, dotted] (ll\v) -- (l\v);
				}

			% Right extensions
			\foreach \v in {v11, v12, bar13}{
					\coordinate[right=5.5mm of \v] (r\v);
					\draw[edge] (\v) -- (r\v);
					\coordinate[right=5.5mm of r\v] (rr\v);
					\draw[edge, dotted] (r\v) -- (rr\v);
				}

			% Claw structure
			\node[clawcenter_small] (clawcenter) at (5, \yrowone) {};
			\node[clawleaves_small] (clawleaf1) at (4, \yrowtwo) {};
			\node[clawleaves_small] (clawleaf2) at (7, \yrowthree) {};
			\node[clawleaves_small] (clawleaf3) at (11, \yrowthree) {};

			\draw[clawedge] (clawcenter) -- (clawleaf1);
			\draw[clawedge, curved] (clawcenter) to (clawleaf2);
			\draw[clawedge, curved] (clawcenter) to (clawleaf3);

			% Labels
			\node[above=2mm of clawcenter, font=\large] {$h_{2j-1}$};
			\node[above left=1mm and 1mm of clawleaf1, font=\large] {$h_{2j-2}$};
			\node[below=2mm of clawleaf2, font=\large] {$\hbar{2j+1}$};
			\node[below=2mm of clawleaf3, font=\large] {$\hbar{2j+5}$};

			% Caption
			\node at (\captionx,\captiony) {(a) Type A claw at $h_{2j-1}$};
		\end{scope}

		% Type B claw
		\begin{scope}[shift={(0,-7.5)}]
			\foreach \i in {1,...,6} {
					\pgfmathsetmacro{\v}{int(2*\i-1)}
					\pgfmathsetmacro{\vlabel}{int(2*\i-1)}
					\node[vs] (v\v) at (\v, \yrowone) {};
				}

			\foreach \i in {2,...,7} {
					\pgfmathsetmacro{\v}{int(2*\i-1)}
					\pgfmathsetmacro{\vlabel}{int(2*\i-1)}
					\node[vs] (bar\v) at (\v, \yrowthree) {};
				}

			% Second row vertices (even numbers)
			\foreach \i in {1,...,6} {
					\pgfmathsetmacro{\v}{int(2*\i)}
					\pgfmathsetmacro{\vlabel}{int(2*\i)}
					\node[vs] (v\v) at (\v, \yrowtwo) {};
				}

			% Horizontal edges
			\foreach \i in {1,...,10} {
					\pgfmathsetmacro{\inext}{int(\i+2)}
					\draw[edge] (v\i) -- (v\inext);
				}
			\foreach \i in {1,...,11} {
					\pgfmathsetmacro{\inext}{int(\i+1)}
					\draw[edge] (v\i) -- (v\inext);
				}

			% Diagonal edges
			\foreach \i in {1,...,6} {
					\pgfmathsetmacro{\v}{int(2*\i-1)}
					\pgfmathsetmacro{\vnext}{int(2*\i+1)}
					\pgfmathsetmacro{\even}{int(2*\i)}
					\draw[edge] (v\even) -- (bar\vnext);
				}

			\foreach \i in {2,...,6} {
					\pgfmathsetmacro{\v}{int(2*\i-1)}
					\pgfmathsetmacro{\vnext}{int(2*\i+1)}
					\pgfmathsetmacro{\even}{int(2*\i)}
					\draw[edge] (v\even) -- (bar\v);
				}

			\foreach \i in {2,...,6} {
					\pgfmathsetmacro{\v}{int(2*\i-1)}
					\pgfmathsetmacro{\vnext}{int(2*\i+1)}
					\pgfmathsetmacro{\even}{int(2*\i)}
					\draw[edge] (bar\v) -- (bar\vnext);
				}

			% Curved edges
			\foreach \i in {2,...,6} {
					\pgfmathsetmacro{\v}{int(2*\i-1)}
					\draw[edgefar, curved_right] (v\v) to (bar\v);
				}

			\foreach \i in {1,...,6} {
					\pgfmathsetmacro{\v}{int(2*\i-1)}
					\pgfmathsetmacro{\vnext}{int(2*\i+1)}
					\draw[edgefar] (v\v) to (bar\vnext);
				}

			\foreach \i in {1,...,5} {
					\pgfmathsetmacro{\v}{int(2*\i-1)}
					\pgfmathsetmacro{\vnext}{int(2*\i+3)}
					\draw[edgefar] (v\v) to (bar\vnext);
				}

			\foreach \i in {1,...,4} {
					\pgfmathsetmacro{\v}{int(2*\i-1)}
					\pgfmathsetmacro{\vnext}{int(2*\i+5)}
					\draw[edgefar] (v\v) to (bar\vnext);
				}

			% Left extensions
			\foreach \v in {v1, v2, bar3}{
					\coordinate[left=5.5mm of \v] (l\v);
					\draw[edge] (l\v) -- (\v);
					\coordinate[left=5.5mm of l\v] (ll\v);
					\draw[edge, dotted] (ll\v) -- (l\v);
				}

			% Right extensions
			\foreach \v in {v11, v12, bar13}{
					\coordinate[right=5.5mm of \v] (r\v);
					\draw[edge] (\v) -- (r\v);
					\coordinate[right=5.5mm of r\v] (rr\v);
					\draw[edge, dotted] (r\v) -- (rr\v);
				}

			% Claw structure
			\node[clawcenter_small] (clawcenter) at (5, \yrowone) {};
			\node[clawleaves_small] (clawleaf1) at (3, \yrowone) {};
			\node[clawleaves_small] (clawleaf2) at (6, \yrowtwo) {};
			\node[clawleaves_small] (clawleaf3) at (11, \yrowthree) {};

			\draw[clawedge] (clawcenter) -- (clawleaf1);
			\draw[clawedge] (clawcenter) -- (clawleaf2);
			\draw[clawedge, curved] (clawcenter) to (clawleaf3);

			% Labels
			\node[above=2mm of clawcenter, font=\large] {$h_{2j-1}$};
			\node[above=2mm of clawleaf1, font=\large] {$h_{2j-3}$};
			\node[below right=0.2mm and 1mm of clawleaf2, font=\large] {$h_{2j}$};
			\node[below=2mm of clawleaf3, font=\large] {$\hbar{2j+5}$};

			% Caption
			\node at (\captionx,\captiony) {(b) Type B claw at $h_{2j-1}$};
		\end{scope}
	\end{tikzpicture}
	\caption{
		Structures of claws centered at $h_{2j-1}$.
		The bold green edges connect the claw center (blue circle) to its three leaves (green circles), while gray edges indicate all other edges in the frustration graph.
		(a) Type A claw: leaves are located at vertices $h_{2j-2}$, $\hbar{2j+1}$, and $\hbar{2j+5}$.
		(b) Type B claw: leaves are located at vertices $h_{2j-3}$, $h_{2j}$, and $\hbar{2j+5}$.
		Note that both claw types share the common leaf $\hbar{2j+5}$.
		For a more detailed structural analysis, see Figure~\ref{fig:claw-configurations}.
	}
	\label{fig:claw}
\end{figure}
%%%%%%%%%%%%%%%%%%%%%%%%%%%%5
\begin{figure}[tbp]
	\centering
	\begin{tikzpicture}[triangular lattice small]
		% Define common y-coordinates
		\pgfmathsetmacro{\yrowone}{0}
		\pgfmathsetmacro{\yrowtwo}{-1.732}  % -sqrt(3)
		\pgfmathsetmacro{\yrowthree}{-3.464} % -2*sqrt(3)
		\pgfmathsetmacro{\offset}{-2}

		\def\captionx{7}
		\def\captiony{-5}
		\def\shifty{-6.5}

		% Type A claw
		\begin{scope}[shift={(0,0)}]
			\foreach \i in {1,...,6} {
					\pgfmathsetmacro{\v}{int(2*\i-1)}
					\pgfmathsetmacro{\vlabel}{int(2*\i-1)}
					\node[vs] (v\v) at (\v, \yrowone) {};
				}

			\foreach \i in {2,...,7} {
					\pgfmathsetmacro{\v}{int(2*\i-1)}
					\pgfmathsetmacro{\vlabel}{int(2*\i-1)}
					\node[vs] (bar\v) at (\v, \yrowthree) {};
				}

			% Second row vertices (even numbers)
			\foreach \i in {1,...,6} {
					\pgfmathsetmacro{\v}{int(2*\i)}
					\pgfmathsetmacro{\vlabel}{int(2*\i)}
					\node[vs] (v\v) at (\v, \yrowtwo) {};
				}

			% Horizontal edges
			\foreach \i in {1,...,10} {
					\pgfmathsetmacro{\inext}{int(\i+2)}
					\draw[edge] (v\i) -- (v\inext);
				}
			\foreach \i in {1,...,11} {
					\pgfmathsetmacro{\inext}{int(\i+1)}
					\draw[edge] (v\i) -- (v\inext);
				}

			% Diagonal edges
			\foreach \i in {1,...,6} {
					\pgfmathsetmacro{\v}{int(2*\i-1)}
					\pgfmathsetmacro{\vnext}{int(2*\i+1)}
					\pgfmathsetmacro{\even}{int(2*\i)}
					\draw[edge] (v\even) -- (bar\vnext);
				}

			\foreach \i in {2,...,6} {
					\pgfmathsetmacro{\v}{int(2*\i-1)}
					\pgfmathsetmacro{\vnext}{int(2*\i+1)}
					\pgfmathsetmacro{\even}{int(2*\i)}
					\draw[edge] (v\even) -- (bar\v);
				}

			\foreach \i in {2,...,6} {
					\pgfmathsetmacro{\v}{int(2*\i-1)}
					\pgfmathsetmacro{\vnext}{int(2*\i+1)}
					\pgfmathsetmacro{\even}{int(2*\i)}
					\draw[edge] (bar\v) -- (bar\vnext);
				}

			% Curved edges
			\foreach \i in {2,...,6} {
					\pgfmathsetmacro{\v}{int(2*\i-1)}
					\draw[edgefar, curved_right] (v\v) to (bar\v);
				}

			\foreach \i in {1,...,6} {
					\pgfmathsetmacro{\v}{int(2*\i-1)}
					\pgfmathsetmacro{\vnext}{int(2*\i+1)}
					\draw[edgefar] (v\v) to (bar\vnext);
				}

			\foreach \i in {1,...,5} {
					\pgfmathsetmacro{\v}{int(2*\i-1)}
					\pgfmathsetmacro{\vnext}{int(2*\i+3)}
					\draw[edgefar] (v\v) to (bar\vnext);
				}

			\foreach \i in {1,...,4} {
					\pgfmathsetmacro{\v}{int(2*\i-1)}
					\pgfmathsetmacro{\vnext}{int(2*\i+5)}
					\draw[edgefar] (v\v) to (bar\vnext);
				}

			% Left extensions
			\foreach \v in {v1, v2, bar3}{
					\coordinate[left=5.5mm of \v] (l\v);
					\draw[edge] (l\v) -- (\v);
					\coordinate[left=5.5mm of l\v] (ll\v);
					\draw[edge, dotted] (ll\v) -- (l\v);
				}

			% Right extensions
			\foreach \v in {v11, v12, bar13}{
					\coordinate[right=5.5mm of \v] (r\v);
					\draw[edge] (\v) -- (r\v);
					\coordinate[right=5.5mm of r\v] (rr\v);
					\draw[edge, dotted] (r\v) -- (rr\v);
				}

			% Claw structure
			\node[clawcenter_small] (clawcenter) at (9, \yrowthree) {};
			\node[clawleaves_small] (clawleaf1) at (3, \yrowone) {};
			\node[clawleaves_small] (clawleaf2) at (7, \yrowone) {};
			\node[clawleaves_small] (clawleaf3) at (10, \yrowtwo) {};

			\draw[clawedge, curved] (clawleaf1) to (clawcenter);
			\draw[clawedge, curved] (clawleaf2) to (clawcenter);
			\draw[clawedge] (clawcenter) -- (clawleaf3);

			% Labels
			\node[below=2mm of clawcenter, font=\large] {$\hbar{2j+3}$};
			\node[above=2mm of clawleaf1, font=\large] {$h_{2j-3}$};
			\node[above=2mm of clawleaf2, font=\large] {$h_{2j+1}$};
			\node[above right=-2mm and 0.75mm of clawleaf3, font=\large] {$h_{2j+4}$};

			% Caption
			\node at (\captionx,\captiony) {(a) Type A claw at $\hbar{2j+3}$};
		\end{scope}

		% Type B claw
		\begin{scope}[shift={(0,-7.5)}]
			\foreach \i in {1,...,6} {
					\pgfmathsetmacro{\v}{int(2*\i-1)}
					\pgfmathsetmacro{\vlabel}{int(2*\i-1)}
					\node[vs] (v\v) at (\v, \yrowone) {};
				}

			\foreach \i in {2,...,7} {
					\pgfmathsetmacro{\v}{int(2*\i-1)}
					\pgfmathsetmacro{\vlabel}{int(2*\i-1)}
					\node[vs] (bar\v) at (\v, \yrowthree) {};
				}

			% Second row vertices (even numbers)
			\foreach \i in {1,...,6} {
					\pgfmathsetmacro{\v}{int(2*\i)}
					\pgfmathsetmacro{\vlabel}{int(2*\i)}
					\node[vs] (v\v) at (\v, \yrowtwo) {};
				}

			% Horizontal edges
			\foreach \i in {1,...,10} {
					\pgfmathsetmacro{\inext}{int(\i+2)}
					\draw[edge] (v\i) -- (v\inext);
				}
			\foreach \i in {1,...,11} {
					\pgfmathsetmacro{\inext}{int(\i+1)}
					\draw[edge] (v\i) -- (v\inext);
				}

			% Diagonal edges
			\foreach \i in {1,...,6} {
					\pgfmathsetmacro{\v}{int(2*\i-1)}
					\pgfmathsetmacro{\vnext}{int(2*\i+1)}
					\pgfmathsetmacro{\even}{int(2*\i)}
					\draw[edge] (v\even) -- (bar\vnext);
				}

			\foreach \i in {2,...,6} {
					\pgfmathsetmacro{\v}{int(2*\i-1)}
					\pgfmathsetmacro{\vnext}{int(2*\i+1)}
					\pgfmathsetmacro{\even}{int(2*\i)}
					\draw[edge] (v\even) -- (bar\v);
				}

			\foreach \i in {2,...,6} {
					\pgfmathsetmacro{\v}{int(2*\i-1)}
					\pgfmathsetmacro{\vnext}{int(2*\i+1)}
					\pgfmathsetmacro{\even}{int(2*\i)}
					\draw[edge] (bar\v) -- (bar\vnext);
				}

			% Curved edges
			\foreach \i in {2,...,6} {
					\pgfmathsetmacro{\v}{int(2*\i-1)}
					\draw[edgefar, curved_right] (v\v) to (bar\v);
				}

			\foreach \i in {1,...,6} {
					\pgfmathsetmacro{\v}{int(2*\i-1)}
					\pgfmathsetmacro{\vnext}{int(2*\i+1)}
					\draw[edgefar] (v\v) to (bar\vnext);
				}

			\foreach \i in {1,...,5} {
					\pgfmathsetmacro{\v}{int(2*\i-1)}
					\pgfmathsetmacro{\vnext}{int(2*\i+3)}
					\draw[edgefar] (v\v) to (bar\vnext);
				}

			\foreach \i in {1,...,4} {
					\pgfmathsetmacro{\v}{int(2*\i-1)}
					\pgfmathsetmacro{\vnext}{int(2*\i+5)}
					\draw[edgefar] (v\v) to (bar\vnext);
				}

			% Left extensions
			\foreach \v in {v1, v2, bar3}{
					\coordinate[left=5.5mm of \v] (l\v);
					\draw[edge] (l\v) -- (\v);
					\coordinate[left=5.5mm of l\v] (ll\v);
					\draw[edge, dotted] (ll\v) -- (l\v);
				}

			% Right extensions
			\foreach \v in {v11, v12, bar13}{
					\coordinate[right=5.5mm of \v] (r\v);
					\draw[edge] (\v) -- (r\v);
					\coordinate[right=5.5mm of r\v] (rr\v);
					\draw[edge, dotted] (r\v) -- (rr\v);
				}

			% Claw structure
			\node[clawcenter_small] (clawcenter) at (9, \yrowthree) {};
			\node[clawleaves_small] (clawleaf1) at (3, \yrowone) {};
			\node[clawleaves_small] (clawleaf2) at (8, \yrowtwo) {};
			\node[clawleaves_small] (clawleaf3) at (11, \yrowthree) {};

			\draw[clawedge, curved] (clawleaf1) to (clawcenter);
			\draw[clawedge] (clawleaf2) -- (clawcenter);
			\draw[clawedge] (clawcenter) -- (clawleaf3);

			% Labels
			\node[below=2mm of clawcenter, font=\large] {$\hbar{2j+3}$};
			\node[above=2mm of clawleaf1, font=\large] {$h_{2j-3}$};
			\node[above right=0mm and 1mm of clawleaf2, font=\large] {$h_{2j+2}$};
			\node[below=2mm of clawleaf3, font=\large] {$\hbar{2j+5}$};

			% Caption
			\node at (\captionx,\captiony) {(b) Type B claw at $\hbar{2j+3}$};
		\end{scope}
	\end{tikzpicture}
	\caption{
		Structures of claws centered at $\hbar{2j+3}$.
		The bold green edges connect the claw center (blue circle) to its three leaves (green circles), while gray edges indicate all other edges in the frustration graph.
		(a) Type A claw: leaves are located at vertices $h_{2j-3}$, $h_{2j+1}$, and $h_{2j+4}$.
		(b) Type B claw: leaves are located at vertices $h_{2j-3}$, $h_{2j+2}$, and $\hbar{2j+5}$.
		Note that both claw types share the common leaf $h_{2j-3}$.
		For a more detailed structural analysis, see Figure~\ref{fig:claw-configurations-bar}.
	}
	\label{fig:claw-bar}
\end{figure}
\begin{figure}[tbp]
	\centering
	\begin{tikzpicture}[triangular lattice small, evenhole/.style={draw=CadetBlue, line width=2pt},]
		% Define common y-coordinates
		\pgfmathsetmacro{\yrowone}{0}
		\pgfmathsetmacro{\yrowtwo}{-1.732}  % -sqrt(3)
		\pgfmathsetmacro{\yrowthree}{-3.464} % -2*sqrt(3)
		\pgfmathsetmacro{\offset}{-2}

		\def\captionx{7}
		\def\captiony{-5}
		\def\shifty{-6.5}

		% even hole
		\foreach \i in {1,...,4} {
				\pgfmathsetmacro{\v}{int(2*\i-1)}
				\node[vs] (v\v) at (\v, \yrowone) {};
			}

		\foreach \i in {2,...,5} {
				\pgfmathsetmacro{\v}{int(2*\i-1)}
				\node[vs] (bar\v) at (\v, \yrowthree) {};
			}

		% Second row vertices (even numbers)
		\foreach \i in {1,...,4} {
				\pgfmathsetmacro{\v}{int(2*\i)}
				\pgfmathsetmacro{\vlabel}{int(2*\i)}
				\node[vs] (v\v) at (\v, \yrowtwo) {};
			}

		% Horizontal edges
		\foreach \i in {1,...,6} {
				\pgfmathsetmacro{\inext}{int(\i+2)}
				\draw[edge] (v\i) -- (v\inext);
			}
		\foreach \i in {1,...,7} {
				\pgfmathsetmacro{\inext}{int(\i+1)}
				\draw[edge] (v\i) -- (v\inext);
			}

		% Diagonal edges
		\foreach \i in {1,...,4} {
				\pgfmathsetmacro{\v}{int(2*\i-1)}
				\pgfmathsetmacro{\vnext}{int(2*\i+1)}
				\pgfmathsetmacro{\even}{int(2*\i)}
				\draw[edge] (v\even) -- (bar\vnext);
			}

		\foreach \i in {2,...,4} {
				\pgfmathsetmacro{\v}{int(2*\i-1)}
				\pgfmathsetmacro{\vnext}{int(2*\i+1)}
				\pgfmathsetmacro{\even}{int(2*\i)}
				\draw[edge] (v\even) -- (bar\v);
			}

		\foreach \i in {2,...,4} {
				\pgfmathsetmacro{\v}{int(2*\i-1)}
				\pgfmathsetmacro{\vnext}{int(2*\i+1)}
				\pgfmathsetmacro{\even}{int(2*\i)}
				\draw[edge] (bar\v) -- (bar\vnext);
			}

		% Curved edges
		\foreach \i in {2,...,4} {
				\pgfmathsetmacro{\v}{int(2*\i-1)}
				\draw[edgefar, curved_right] (v\v) to (bar\v);
			}

		\foreach \i in {1,...,4} {
				\pgfmathsetmacro{\v}{int(2*\i-1)}
				\pgfmathsetmacro{\vnext}{int(2*\i+1)}
				\draw[edgefar] (v\v) to (bar\vnext);
			}

		\foreach \i in {1,...,3} {
				\pgfmathsetmacro{\v}{int(2*\i-1)}
				\pgfmathsetmacro{\vnext}{int(2*\i+3)}
				\draw[edgefar] (v\v) to (bar\vnext);
			}

		\foreach \i in {1,...,2} {
				\pgfmathsetmacro{\v}{int(2*\i-1)}
				\pgfmathsetmacro{\vnext}{int(2*\i+5)}
				\draw[edgefar] (v\v) to (bar\vnext);
			}

		% Left extensions
		\foreach \v in {v1, v2, bar3}{
				\coordinate[left=5.5mm of \v] (l\v);
				\draw[edge] (l\v) -- (\v);
				\coordinate[left=5.5mm of l\v] (ll\v);
				\draw[edge, dotted] (ll\v) -- (l\v);
			}

		% Right extensions
		\foreach \v in {v7, v8, bar9}{
				\coordinate[right=5.5mm of \v] (r\v);
				\draw[edge] (\v) -- (r\v);
				\coordinate[right=5.5mm of r\v] (rr\v);
				\draw[edge, dotted] (r\v) -- (rr\v);
			}

		% Claw structure
		\node[clawcenter_small] (e1) at (3, \yrowone) {};
		\node[clawcenter_small] (e2) at (4, \yrowtwo) {};
		\node[clawcenter_small] (e3) at (6, \yrowtwo) {};
		\node[clawcenter_small] (e4) at (7, \yrowthree) {};

		\draw[evenhole] (e1) to (e2) to (e3) to (e4);
		\draw[evenhole, curved] (e1) to (e4);

		% Labels
		\node[above=2mm of e1, font=\large] {$h_{2j-3}$};
		\node[above left=-0.75mm and 1mm of e2, font=\large] {$h_{2j-2}$};
		\node[above right=-0.75mm and 0.5mm of e3, font=\large] {$h_{2j}$};
		\node[below=2mm of e4, font=\large] {$\hbar{2j+1}$};

	\end{tikzpicture}
	\caption{
		Structure of an even hole $C_4$ in the frustration graph for the extended FFD algebra: $\mu_{2j-1} = h_{2j-3} h_{2j} h_{2j-2} \hbar{2j+1}$.
		The vertices in the even hole are represented by blue-filled circles and the four edges in the even hole are represented by bold blue lines.
	}
	\label{fig:even-hole}
\end{figure}
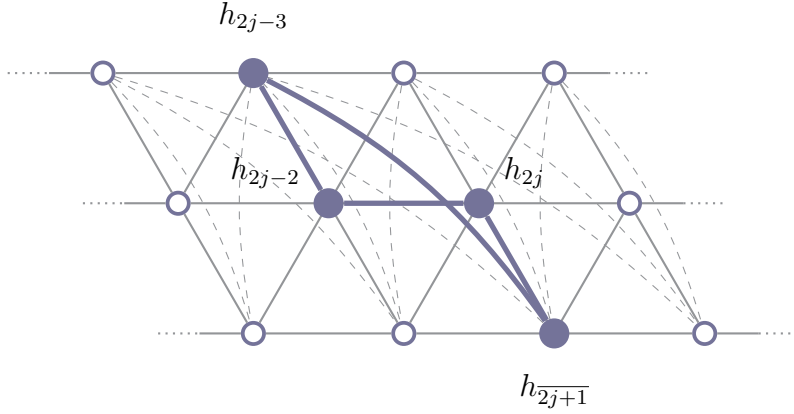

\subsection{Additional generators and the frustration graph}
We propose an algebraic extension of the FFD algebra with new generators $\{\hbar{2m+1}\}_{m=0}^{M^\prime}$.
As in the original FFD algebra, the squares are scalars:
\begin{align}
	\hbar{2m-1}^2 = \bbar{2m-1}^2,
\end{align}
where $\bbar{2m-1}$ is a coupling constant.
The generator $\hbar{2m-1}$ anticommutes with the following FFD generators with odd indices:
\begin{align}
	\{\hbar{2m-1},h_{2m-7}\} = \{\hbar{2m-1},h_{2m-5}\} = \{\hbar{2m-1},h_{2m-3}\} = \{\hbar{2m-1},h_{2m-1}\} = 0,
	\label{eq:extended-anticommutator-odd}
\end{align}
and also with those with even indices:
\begin{align}
	\{\hbar{2m-1},h_{2m-2}\} = \{\hbar{2m-1},h_{2m}\} = 0,
\end{align}
as well as within the extended algebra:
\begin{align}
	\{\hbar{2m-1},\hbar{2m+1}\} = 0,
\end{align}
and commutes with all others.

Following~\cite{fermions-behind-the-disguise,unified-graph-th}, we introduce the frustration graph for the extended FFD algebra.
We show the frustration graph $G_M$ for this extended FFD algebra in Figure~\ref{fig:extended-ffd-frustration-graph}, where the additional vertices for $\hbar{2j-1}$ are placed in the lower row, and the dotted edges represent the anticommutation relations~\eqref{eq:extended-anticommutator-odd}.

\subsection{Claw in the extended frustration graph}
\label{sec:claw-notations}
A notable property of the new frustration graph is that it contains claws.
We show the claws in the frustration graph for the extended FFD algebra in Figures~\ref{fig:claw} and \ref{fig:claw-bar}.
The claw center can be $h_{2j-1}$ or $\hbar{2j-1}$, with leaves shown in Figure~\ref{fig:claw} for claws centered at $h_{2j-1}$ and in Figure~\ref{fig:claw-bar} for those centered at $\hbar{2j-1}$.
Hereafter, we refer to a claw with center $h_c$ as a ``claw at $h_c$''.

Claws with the same center have two patterns: type A and type B.
Let the three leaves of a claw at center $h_c$ be denoted by $\leaves{c}{t} = (\vleafdep{t,1}_{c}, \vleafdep{t,2}_{c}, \vleafdep{t,3}_{c})$.
For claws at $h_{2j-1}$ (Figure~\ref{fig:claw}): type A has leaves at the vertices $\leaves{2j-1}{A} \equiv (h_{\overline{2j+1}}, h_{2j-2}, h_{\overline{2j+5}})$, while type B has leaves at the vertices $\leaves{2j-1}{B} \equiv (h_{2j-3}, h_{2j}, h_{\overline{2j+5}})$.
For claws at $\hbar{2j+3}$ (Figure~\ref{fig:claw-bar}): type A has leaves at the vertices $\leaves{\overline{2j+3}}{A} \equiv (h_{2j+1}, h_{2j+4}, h_{2j-3})$, while type B has leaves at the vertices $\leaves{\overline{2j+3}}{B} \equiv (h_{\overline{2j+5}}, h_{2j+2}, h_{2j-3})$.
Note that the leaf $\vleafdep{t,3}_{c}$ remains the same for both claw types with the same center $c$, and we denote $\vleafindep_{c} \equiv \vleafdep{t,3}_{c}$.

The range of claw centers is as follows: when $h_c$ with $c = 2m-1$ is a claw center, then $3 \le c \le 2M'-5$; when $\hbar{c}$ is a claw center, then $7 \le c \le 2M'-1$.

The sufficient conditions of~\cite{fermions-behind-the-disguise,unified-graph-th} for integrability and free fermionic solvability include claw-freeness of the frustration graph.
At first sight, having claws seems to imply that the methods of~\cite{fermions-behind-the-disguise,unified-graph-th} cannot be applied.
However, the situation is different.
The works~\cite{fermions-behind-the-disguise,unified-graph-th} treated models where the terms in the Hamiltonian were algebraically independent of each other.
In our case, this is not true because we have the relation~\eqref{eq:ac-ca-relation}.
We will argue that this extra relation compensates for the claw obstruction and salvages the applicability of the methods of~\cite{fermions-behind-the-disguise,unified-graph-th}.

\subsection{Even hole in the extended frustration graph}
Another notable feature of the new frustration graph is the existence of even holes $C_{4}$.
We show an even hole in the frustration graph for the extended FFD algebra in Figure~\ref{fig:even-hole}.
These even holes are formed by the following four generators: $\{h_{2m-3}, h_{2m}, h_{2m-2}, h_{\overline{2m+1}}\}$ for $1 < m \le M'$.
We define $\mu_{2m-1}$ as the product of these generators:
\begin{align}
	\mu_{2m-1}
	\equiv
	h_{2m-3} h_{2m} h_{2m-2} \hbar{2m+1}
	\quad \text{for $1 < m \le M'$}.
	\label{eq:center-even-hole}
\end{align}

The important feature of~\eqref{eq:center-even-hole} is that it is central in the extended FFD algebra; $\mu_{2m-1}$ commutes with all generators in the extended FFD algebra: $[\mu_{2m-1}, h_{m'}] = 0$ for $1 \le m' \le M$, and $[\mu_{2m-1}, \hbar{2m'+1}] = 0$ for $0 \le m' \le M'$.
These properties can be proven from the fact that each vertex in the frustration graph has zero, two, or four connections to the vertices in the even hole, as can be confirmed in Figure~\ref{fig:even-hole}.

Moreover, its square is a scalar: $\mu_{2m-1}^{2} = b_{2m-3}^2 b_{2m}^2 b_{2m-2}^2 \bbar{2m+1}^2$.
\red{After fixing a central sector, we may replace $\mu_{2m-1}$ by a scalar eigenvalue:}
\begin{align}
	\mu_{2m-1} = \pm b_{2m} b_{2m-3} b_{2m-2} \bbar{2m+1}
	\,,
	\label{eq:mu-scalar}
\end{align}
where the sign is chosen independently for each central element $\mu_{2m-1}$.
\red{This sign is a gauge choice: choosing the opposite sign changes the corresponding effective parameter $\beta_{2m-3}$ and is absorbed into the gauge-fixed couplings.}
Therefore, we can express the extra generator $\hbar{2m+1}$ as
\begin{align}
	\hbar{2m+1} = \beta_{2m-3} h_{2m-2} h_{2m-3} h_{2m} \quad \text{for $1 < m \le M'$},
	\label{eq:triples-from-algebra}
\end{align}
where we define
\begin{align}
	\beta_{2m-3} \equiv \pm \frac{\bbar{2m+1}}{b_{2m-2} b_{2m-3} b_{2m}}.
	\label{eq:beta-origin}
\end{align}
Here, we can see that the triple products in the new Hamiltonian~\eqref{eq:new-model} can be derived from the extended FFD algebra as in~\eqref{eq:triples-from-algebra}.
Note that $\hbar{1}$ and $\hbar{3}$ cannot be expressed in terms of the generators of the original FFD algebra.

\red{The central even holes of the present model are a special case of the generalized cycle symmetries of~\cite{unified-graph-th}: each $\mu_{2m-1}$ is the product of the generators around a single even hole, which here is not merely a symmetry of the Hamiltonian but is central in the algebra, commuting with every generator.}

\subsection{Extra relations for free fermionic integrability}
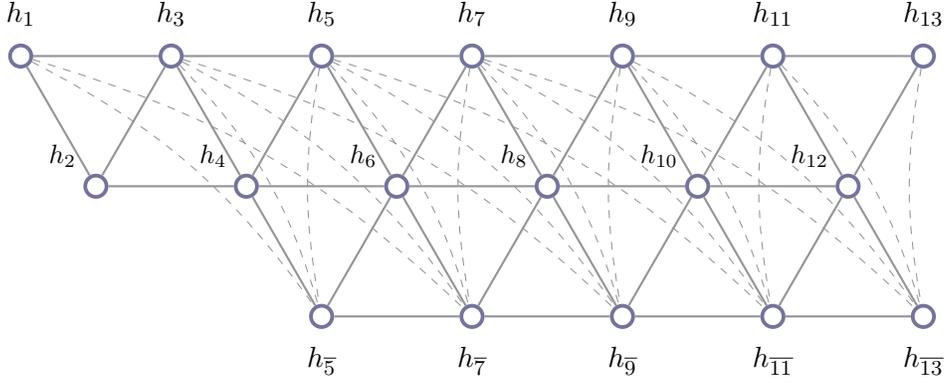
\begin{figure}[tbp]
	\centering
	\begin{tikzpicture}[triangular lattice small, extendededge/.style={draw=RoyalBlue, bend right=\curveratio, line width=1.25pt}, extendededgeright/.style={draw=RoyalBlue, bend right=\curveratioright, line width=1.25pt}]

		\pgfmathsetmacro{\yrowone}{0}
		\pgfmathsetmacro{\yrowtwo}{-1.732}
		\pgfmathsetmacro{\yrowthree}{-3.464}

		% First row vertices (odd numbers)
		\foreach \i in {1,...,7} {
				\pgfmathsetmacro{\v}{int(2*\i-1)}
				\pgfmathsetmacro{\vlabel}{int(2*\i-1)}
				\node[vs] (v\v) at (\v, \yrowone) {};
			}

		\foreach \i in {3,...,7} {
				\pgfmathsetmacro{\v}{int(2*\i-1)}
				\pgfmathsetmacro{\vlabel}{int(2*\i-1)}
				\node[vs] (bar\v) at (\v, \yrowthree) {};
			}

		% Second row vertices (even numbers)
		\foreach \i in {1,...,6} {
				\pgfmathsetmacro{\v}{int(2*\i)}
				\pgfmathsetmacro{\vlabel}{int(2*\i)}
				\node[vs] (v\v) at (\v, \yrowtwo) {};
			}
		\foreach \i in {1,...,11} {
				\pgfmathsetmacro{\inext}{int(\i+2)}
				\draw[edge] (v\i) -- (v\inext);
			}
		\foreach \i in {1,...,12} {
				\pgfmathsetmacro{\inext}{int(\i+1)}
				\draw[edge] (v\i) -- (v\inext);
			}

		\foreach \i in {2,...,6} {
				\pgfmathsetmacro{\v}{int(2*\i-1)}
				\pgfmathsetmacro{\vnext}{int(2*\i+1)}
				\pgfmathsetmacro{\even}{int(2*\i)}
				\draw[edge] (v\even) -- (bar\vnext);
			}

		\foreach \i in {3,...,6} {
				\pgfmathsetmacro{\v}{int(2*\i-1)}
				\pgfmathsetmacro{\vnext}{int(2*\i+1)}
				\pgfmathsetmacro{\even}{int(2*\i)}
				\draw[edge] (bar\v) -- (bar\vnext);
			}

		\foreach \i in {3,...,6} {
				\pgfmathsetmacro{\v}{int(2*\i-1)}
				\pgfmathsetmacro{\vnext}{int(2*\i+1)}
				\pgfmathsetmacro{\even}{int(2*\i)}
				\draw[edge] (v\even) -- (bar\v);
			}

		\foreach \i in {3,...,7} {
				\pgfmathsetmacro{\v}{int(2*\i-1)}
				\draw[edgefar, curved_right] (v\v) to (bar\v);
			}

		\foreach \i in {2,...,6} {
				\pgfmathsetmacro{\v}{int(2*\i-1)}
				\pgfmathsetmacro{\vnext}{int(2*\i+1)}
				\draw[edgefar] (v\v) to (bar\vnext);
			}

		\foreach \i in {1,...,5} {
				\pgfmathsetmacro{\v}{int(2*\i-1)}
				\pgfmathsetmacro{\vnext}{int(2*\i+3)}
				\draw[edgefar] (v\v) to (bar\vnext);
			}

		\foreach \i in {1,...,4} {
				\pgfmathsetmacro{\v}{int(2*\i-1)}
				\pgfmathsetmacro{\vnext}{int(2*\i+5)}
				\draw[edgefar] (v\v) to (bar\vnext);
			}

		% First row vertices (odd numbers)
		\foreach \i in {1,...,7} {
				\pgfmathsetmacro{\v}{int(2*\i-1)}
				\pgfmathsetmacro{\vlabel}{int(2*\i-1)}
				\node[above = 1mm of v\v] {$h_{\vlabel}$};
			}

		\foreach \i in {3,...,7} {
				\pgfmathsetmacro{\v}{int(2*\i-1)}
				\pgfmathsetmacro{\vlabel}{int(2*\i-1)}
				\node[below = 1mm of bar\v] {$\hbar{\vlabel}$};
			}

		% Second row vertices (even numbers)
		\foreach \i in {1,...,6} {
				\pgfmathsetmacro{\v}{int(2*\i)}
				\pgfmathsetmacro{\vlabel}{int(2*\i)}
				\node[above left = 0.01mm and 0.1mm of v\v] {\small $h_{\vlabel}$};
			}

	\end{tikzpicture}

	\caption{
		Frustration graph for the Hamiltonian~\eqref{eq:new-model} ($M=13$).
		The vertices $\hbar{1}$ and $\hbar{3}$ are not present for this case.
	}
	\label{fig:extended-ffd-frustration-graph-for-M2}
\end{figure}
\begin{figure}[tbp]
	\centering
	\begin{tikzpicture}[triangular lattice small]
		% Define common y-coordinates
		\pgfmathsetmacro{\yrowone}{0}
		\pgfmathsetmacro{\yrowtwo}{-1.732}  % -sqrt(3)
		\pgfmathsetmacro{\yrowthree}{-3.464} % -2*sqrt(3)
		\pgfmathsetmacro{\offset}{-2}

		\def\captionx{6.5}
		\def\captiony{-4.425}
		\def\shifty{-6.5}

		\begin{scope}[shift={(0,0)}]
			\foreach \i in {1,...,6} {
					\pgfmathsetmacro{\v}{int(2*\i-1)}
					\pgfmathsetmacro{\vlabel}{int(2*\i-1)}
					\node[vs] (v\v) at (\v, \yrowone) {};
				}

			\foreach \i in {2,...,6} {
					\pgfmathsetmacro{\v}{int(2*\i-1)}
					\pgfmathsetmacro{\vlabel}{int(2*\i-1)}
					\node[vs] (bar\v) at (\v, \yrowthree) {};
				}

			% Second row vertices (even numbers)
			\foreach \i in {1,...,5} {
					\pgfmathsetmacro{\v}{int(2*\i)}
					\pgfmathsetmacro{\vlabel}{int(2*\i)}
					\node[vs] (v\v) at (\v, \yrowtwo) {};
				}

			% Horizontal edges
			\foreach \i in {1,...,9} {
					\pgfmathsetmacro{\inext}{int(\i+2)}
					\draw[edge] (v\i) -- (v\inext);
				}
			\foreach \i in {1,...,10} {
					\pgfmathsetmacro{\inext}{int(\i+1)}
					\draw[edge] (v\i) -- (v\inext);
				}

			% Diagonal edges
			\foreach \i in {1,...,5} {
					\pgfmathsetmacro{\v}{int(2*\i-1)}
					\pgfmathsetmacro{\vnext}{int(2*\i+1)}
					\pgfmathsetmacro{\even}{int(2*\i)}
					\draw[edge] (v\even) -- (bar\vnext);
				}

			\foreach \i in {2,...,5} {
					\pgfmathsetmacro{\v}{int(2*\i-1)}
					\pgfmathsetmacro{\vnext}{int(2*\i+1)}
					\pgfmathsetmacro{\even}{int(2*\i)}
					\draw[edge] (v\even) -- (bar\v);
				}

			\foreach \i in {2,...,5} {
					\pgfmathsetmacro{\v}{int(2*\i-1)}
					\pgfmathsetmacro{\vnext}{int(2*\i+1)}
					\pgfmathsetmacro{\even}{int(2*\i)}
					\draw[edge] (bar\v) -- (bar\vnext);
				}

			% Curved edges
			\foreach \i in {2,...,6} {
					\pgfmathsetmacro{\v}{int(2*\i-1)}
					\draw[edgefar, curved_right] (v\v) to (bar\v);
				}

			\foreach \i in {1,...,5} {
					\pgfmathsetmacro{\v}{int(2*\i-1)}
					\pgfmathsetmacro{\vnext}{int(2*\i+1)}
					\draw[edgefar] (v\v) to (bar\vnext);
				}

			\foreach \i in {1,...,4} {
					\pgfmathsetmacro{\v}{int(2*\i-1)}
					\pgfmathsetmacro{\vnext}{int(2*\i+3)}
					\draw[edgefar] (v\v) to (bar\vnext);
				}

			\foreach \i in {1,...,3} {
					\pgfmathsetmacro{\v}{int(2*\i-1)}
					\pgfmathsetmacro{\vnext}{int(2*\i+5)}
					\draw[edgefar] (v\v) to (bar\vnext);
				}

			% Left extensions
			\foreach \v in {v1, v2, bar3}{
					\coordinate[left=5.5mm of \v] (l\v);
					\draw[edge] (l\v) -- (\v);
					\coordinate[left=5.5mm of l\v] (ll\v);
					\draw[edge, dotted] (ll\v) -- (l\v);
				}

			% Labels
			\node[right=2mm of v10, font=\large] {$h_{2m}$};
			\node[right=2mm of v11, font=\large] {$h_{2m+1}$};
			\node[right=2mm of bar11, font=\large] {$\hbar{2m+1}$};

			% Caption
			\node at (\captionx,\captiony) {(a) Free fermionic frustration graph $G_{2m+1}$};
		\end{scope}

		\begin{scope}[shift={(0,-6)}]
			\foreach \i in {1,...,5} {
					\pgfmathsetmacro{\v}{int(2*\i-1)}
					\pgfmathsetmacro{\vlabel}{int(2*\i-1)}
					\node[vs] (v\v) at (\v, \yrowone) {};
				}

			\foreach \i in {2,...,6} {
					\pgfmathsetmacro{\v}{int(2*\i-1)}
					\pgfmathsetmacro{\vlabel}{int(2*\i-1)}
					\node[vs] (bar\v) at (\v, \yrowthree) {};
				}

			% Second row vertices (even numbers)
			\foreach \i in {1,...,5} {
					\pgfmathsetmacro{\v}{int(2*\i)}
					\pgfmathsetmacro{\vlabel}{int(2*\i)}
					\node[vs] (v\v) at (\v, \yrowtwo) {};
				}

			% Horizontal edges
			\foreach \i in {1,...,8} {
					\pgfmathsetmacro{\inext}{int(\i+2)}
					\draw[edge] (v\i) -- (v\inext);
				}
			\foreach \i in {1,...,9} {
					\pgfmathsetmacro{\inext}{int(\i+1)}
					\draw[edge] (v\i) -- (v\inext);
				}

			% Diagonal edges
			\foreach \i in {1,...,5} {
					\pgfmathsetmacro{\v}{int(2*\i-1)}
					\pgfmathsetmacro{\vnext}{int(2*\i+1)}
					\pgfmathsetmacro{\even}{int(2*\i)}
					\draw[edge] (v\even) -- (bar\vnext);
				}

			\foreach \i in {2,...,5} {
					\pgfmathsetmacro{\v}{int(2*\i-1)}
					\pgfmathsetmacro{\vnext}{int(2*\i+1)}
					\pgfmathsetmacro{\even}{int(2*\i)}
					\draw[edge] (v\even) -- (bar\v);
				}

			\foreach \i in {2,...,5} {
					\pgfmathsetmacro{\v}{int(2*\i-1)}
					\pgfmathsetmacro{\vnext}{int(2*\i+1)}
					\pgfmathsetmacro{\even}{int(2*\i)}
					\draw[edge] (bar\v) -- (bar\vnext);
				}

			% Curved edges
			\foreach \i in {2,...,5} {
					\pgfmathsetmacro{\v}{int(2*\i-1)}
					\draw[edgefar, curved_right] (v\v) to (bar\v);
				}

			\foreach \i in {1,...,5} {
					\pgfmathsetmacro{\v}{int(2*\i-1)}
					\pgfmathsetmacro{\vnext}{int(2*\i+1)}
					\draw[edgefar] (v\v) to (bar\vnext);
				}

			\foreach \i in {1,...,4} {
					\pgfmathsetmacro{\v}{int(2*\i-1)}
					\pgfmathsetmacro{\vnext}{int(2*\i+3)}
					\draw[edgefar] (v\v) to (bar\vnext);
				}

			\foreach \i in {1,...,3} {
					\pgfmathsetmacro{\v}{int(2*\i-1)}
					\pgfmathsetmacro{\vnext}{int(2*\i+5)}
					\draw[edgefar] (v\v) to (bar\vnext);
				}

			% Left extensions
			\foreach \v in {v1, v2, bar3}{
					\coordinate[left=5.5mm of \v] (l\v);
					\draw[edge] (l\v) -- (\v);
					\coordinate[left=5.5mm of l\v] (ll\v);
					\draw[edge, dotted] (ll\v) -- (l\v);
				}

			% Labels
			\node[right=2mm of v10, font=\large] {$h_{2m}$};
			\node[right=2mm of v9, font=\large] {$h_{2m-1}$};
			\node[right=2mm of bar11, font=\large] {$\hbar{2m+1}$};

			% Caption
			\node at (\captionx,\captiony) {(b) Free fermionic frustration graph $G_{2m}$};
		\end{scope}

		\begin{scope}[shift={(0,-12)}]
			\foreach \i in {1,...,4} {
					\pgfmathsetmacro{\v}{int(2*\i-1)}
					\pgfmathsetmacro{\vlabel}{int(2*\i-1)}
					\node[vs] (v\v) at (\v, \yrowone) {};
				}

			\foreach \i in {2,...,6} {
					\pgfmathsetmacro{\v}{int(2*\i-1)}
					\pgfmathsetmacro{\vlabel}{int(2*\i-1)}
					\node[vs] (bar\v) at (\v, \yrowthree) {};
				}

			% Second row vertices (even numbers)
			\foreach \i in {1,...,5} {
					\pgfmathsetmacro{\v}{int(2*\i)}
					\pgfmathsetmacro{\vlabel}{int(2*\i)}
					\node[vs] (v\v) at (\v, \yrowtwo) {};
				}

			% Horizontal edges
			\foreach \i in {1,...,6} {
					\pgfmathsetmacro{\inext}{int(\i+2)}
					\draw[edge] (v\i) -- (v\inext);
				}
			\draw[edge] (v8) -- (v10);

			\foreach \i in {1,...,7} {
					\pgfmathsetmacro{\inext}{int(\i+1)}
					\draw[edge] (v\i) -- (v\inext);
				}

			% Diagonal edges
			\foreach \i in {1,...,5} {
					\pgfmathsetmacro{\v}{int(2*\i-1)}
					\pgfmathsetmacro{\vnext}{int(2*\i+1)}
					\pgfmathsetmacro{\even}{int(2*\i)}
					\draw[edge] (v\even) -- (bar\vnext);
				}

			\foreach \i in {2,...,5} {
					\pgfmathsetmacro{\v}{int(2*\i-1)}
					\pgfmathsetmacro{\vnext}{int(2*\i+1)}
					\pgfmathsetmacro{\even}{int(2*\i)}
					\draw[edge] (v\even) -- (bar\v);
				}

			\foreach \i in {2,...,5} {
					\pgfmathsetmacro{\v}{int(2*\i-1)}
					\pgfmathsetmacro{\vnext}{int(2*\i+1)}
					\pgfmathsetmacro{\even}{int(2*\i)}
					\draw[edge] (bar\v) -- (bar\vnext);
				}

			% Curved edges
			\foreach \i in {2,...,4} {
					\pgfmathsetmacro{\v}{int(2*\i-1)}
					\draw[edgefar, curved_right] (v\v) to (bar\v);
				}

			\foreach \i in {1,...,4} {
					\pgfmathsetmacro{\v}{int(2*\i-1)}
					\pgfmathsetmacro{\vnext}{int(2*\i+1)}
					\draw[edgefar] (v\v) to (bar\vnext);
				}

			\foreach \i in {1,...,4} {
					\pgfmathsetmacro{\v}{int(2*\i-1)}
					\pgfmathsetmacro{\vnext}{int(2*\i+3)}
					\draw[edgefar] (v\v) to (bar\vnext);
				}

			\foreach \i in {1,...,3} {
					\pgfmathsetmacro{\v}{int(2*\i-1)}
					\pgfmathsetmacro{\vnext}{int(2*\i+5)}
					\draw[edgefar] (v\v) to (bar\vnext);
				}

			% Left extensions
			\foreach \v in {v1, v2, bar3}{
					\coordinate[left=5.5mm of \v] (l\v);
					\draw[edge] (l\v) -- (\v);
					\coordinate[left=5.5mm of l\v] (ll\v);
					\draw[edge, dotted] (ll\v) -- (l\v);
				}

			% Labels
			\node[right=2mm of v7, font=\large] {$h_{2m-3}$};
			\node[above right=0mm and -2mm of v8, font=\large] {$h_{2m-2}$};
			\node[right=2mm of v10, font=\large] {$h_{2m}$};
			\node[right=2mm of v9, font=\large] {$h_{2m-1}$};
			\node[right=2mm of bar11, font=\large] {$\hbar{2m+1}$};

			% Caption
			\node at (\captionx,\captiony) {(c) Free fermionic frustration graph $G_{2m}^\prime$};
		\end{scope}
	\end{tikzpicture}
		\caption{
			Frustration graphs for the free-fermionic Hamiltonians.
			(a) The frustration graph $G_{2m+1}$.
			(b) The frustration graph $G_{2m+1}|_{b_{2m+1} = 0} = G_{2m}$.
			This frustration graph corresponds to the new Hamiltonian~\eqref{eq:new-model} with even $M$.
			(c) The frustration graph $\eval{G_{2m+1}}_{b_{2m+1}, b_{2m-1} = 0} = G_{2m}^\prime$.
			\red{The free-fermionic solvability of the corresponding Hamiltonians is proved in Sections~\ref{sec:charges}--\ref{sec:fermion}.}
			The left edge of the frustration graph should also have these configurations rotated by 180 degrees.
		}
		\label{fig:free-fermionic-frustration-graph}
	\end{figure}
%%%%%%%%%%%%%%%%%%%%%%%%%%
%%%%%%%%%%%%%%%%%%%%%%%%%%
%%%%%%%%%%%%%%%%%%%%%%%%%%
We further impose relations among the generators of the extended algebra.
We first define the following quantities:
\begin{align}
	\a{2m-1} & \equiv h_{2m} h_{2m-3} + h_{2m-2} h_{\overline{2m+1}},
	\label{eq:def-a}
	\\
	\C{2m}   & \equiv h_{2m-3} h_{\overline{2m+3}}.
	\label{eq:def-c}
\end{align}
% Note that the square of $\a{2m-1}$ is a scalar:
% \begin{align}
% 	\a{2m-1}^2 = (b_{2m} b_{2m-3} + b_{2m-2} \bbar{2m+1})^2,
% \end{align}
% which can be proven using~\eqref{eq:mu-scalar}.
\red{The following relation is the key algebraic input that compensates for the claw obstruction in the frustration graph.
Its role will become explicit in Theorem~\ref{thm:charge-conservation-law}, where the claw contributions to the commutator $[H_M,\Q{M}{k}]$ cancel precisely because of this relation.
The same relation is also the key input in the proof of mutual commutativity in Appendix~\ref{app:mutual-commutativity}.}
We require
\begin{align}
	\a{2m-1} \C{2m+2}
	=
	\C{2m} \a{2m+3}
	\quad \text{for $1 < m < M' - 1$}.
	\label{eq:ac-ca-relation}
\end{align}
Substituting~\eqref{eq:triples-from-algebra} into~\eqref{eq:ac-ca-relation}, we obtain the relation~\eqref{eq:beta_relation}:
\begin{align}
	 & \a{2m-1} \C{2m+2} - \C{2m} \a{2m+3}
	% \nonumber\\
	% =&
	% \qty(h_{2m-3} h_{2m} + h_{2m-2} \hbar{2m+1})
	% h_{2m-1} \hbar{2m+5}
	% -
	% h_{2m-3} \hbar{2m+3}
	% \qty(h_{2m+1} h_{2m+4} + h_{2m+2} \hbar{2m+5})
	% \nonumber\\
	% =&
	% \beta_{2m+1} \qty(h_{2m-3} h_{2m} + \beta_{2m-3} h_{2m-2} h_{2m-2} h_{2m-3} h_{2m})
	% h_{2m-1} h_{2m+2} h_{2m+1} h_{2m+4}
	% \nonumber\\
	% &\hspace{2em}
	% -
	% \beta_{2m-1} h_{2m-3} h_{2m} h_{2m-1} h_{2m+2}
	% \qty(h_{2m+1} h_{2m+4} + \beta_{2m+1} h_{2m+2} h_{2m+2} h_{2m+1} h_{2m+4})
	% \nonumber\\
	% =&
	% - \beta_{2m+1} \qty(h_{2m-3} h_{2m} + \beta_{2m-3} b_{2m-2}^2 h_{2m-3} h_{2m})
	% h_{2m-1} h_{2m+2} h_{2m+1} h_{2m+4}
	% \nonumber\\
	% &\hspace{2em}
	% +
	% \beta_{2m-1} h_{2m-3} h_{2m-1} h_{2m} h_{2m+2}
	% \qty(h_{2m+1} h_{2m+4} + \beta_{2m+1} b_{2m+2}^2 h_{2m+1} h_{2m+4})
	% \nonumber\\
	=
	\qty[\beta_{2m-1} (1 + \beta_{2m+1} b_{2m+2}^2) - \beta_{2m+1} (1 + \beta_{2m-3} b_{2m-2}^2)]
	\nonumber                                                                   \\
	 & \hspace{16em} \times h_{2m-3} h_{2m-1} h_{2m} h_{2m+2} h_{2m+1} h_{2m+4}
	\,,
\end{align}
thus we have
\begin{align}
	\beta_{2m-1} (1 + \beta_{2m+1} b_{2m+2}^2) = \beta_{2m+1} (1 + \beta_{2m-3} b_{2m-2}^2)
	\,,
\end{align}
which is equivalent to~\eqref{eq:beta_relation}, except when $1 + \beta_{2m-3} b_{2m-2}^2 - \beta_{2m-1} b_{2m+2}^2 = 0$.
The latter case is not considered here.

The condition~\eqref{eq:ac-ca-relation} does not constrain the following coupling constants: $\bbar{1}$, $\bbar{3}$, $b_{2M'-1}$, and $b_{2M'+1}$.
Thus, we can choose these coupling constants freely.

\red{Equivalently, one may choose all $\{b_m\}_{m=1}^{2M'+1}$ freely, together with the unconstrained boundary couplings $\bbar{1}$ and $\bbar{3}$ and the two initial values $\bbar{5}$ and $\bbar{7}$.
Then all remaining couplings $\{\bbar{2m+1}\}_{m=4}^{M'}$ are determined by the recursion relations~\eqref{eq:beta_relation} and~\eqref{eq:beta-origin}.
Thus the continuous parameter count is $2M'+5$, apart from the discrete gauge signs in~\eqref{eq:mu-scalar}.}

Since the coupling constants $b_{2M'-1}$ and $b_{2M'+1}$ can be set freely, we can set them to zero.
Figure~\ref{fig:free-fermionic-frustration-graph}(a) shows the frustration graph $G_{2m+1}$.
Setting $b_{2m+1} = 0$ (i.e., removing the vertex $h_{2m+1}$) yields the frustration graph $G_{2m} = \eval{G_{2m+1}}_{b_{2m+1} = 0}$ for the Hamiltonian~\eqref{eq:new-model} with even $M = 2m$ (Figure~\ref{fig:free-fermionic-frustration-graph}(b)).
Furthermore, removing the vertices $h_{2m-1}$ and $h_{2m+1}$ yields the frustration graph $G_{2m}^\prime = \eval{G_{2m}}_{b_{2m-1} = 0}$ shown in Figure~\ref{fig:free-fermionic-frustration-graph}(c).
\red{We call a frustration graph \emph{free fermionic} if the corresponding Hamiltonian is solvable by free fermions.}
The frustration graphs in Figure~\ref{fig:free-fermionic-frustration-graph} are free fermionic frustration graphs, \red{which will be proved in Sections~\ref{sec:charges}--\ref{sec:fermion}}.
For Figure~\ref{fig:free-fermionic-frustration-graph}(c), we can further eliminate the vertices $h_{2m-3}$ and $\hbar{2m+1}$ simultaneously, as depicted in Figure~\ref{fig:frustrationgraph-for-even-Q-recursion}, which gives another free fermionic frustration graph.

For the left edge of the graph, corresponding to $\bbar{1}$ and $\bbar{3}$, the frustration graph for the free fermionic model must have the same structure as in Figure~\ref{fig:free-fermionic-frustration-graph}, but rotated by 180 degrees.

\subsection{Representation of the extended FFD algebra}
We now show the representation for the extended algebra.
From~\eqref{eq:triples-from-algebra}, for $2 \le m \le M'$, $\hbar{2m+1}$ can be represented using the representation of the original FFD algebra~\eqref{eq:zzx-rep}.

However, $\hbar{1}$ and $\hbar{3}$ cannot be represented by the original FFD representation~\eqref{eq:zzx-rep}.
They require representations outside the original one~\eqref{eq:zzx-rep}:
\begin{align}
	\hbar{1} & = \bbar{1} \sigma_{1}^{z} \sigma_{2}^{y}\,,
	\\
	\hbar{3} & = \bbar{3} \sigma_{1}^{y} \sigma_{2}^{y} \sigma_{3}^{z} \sigma_{4}^{x}\,.
\end{align}

\subsection{General integrable Hamiltonian}
We define the abstract Hamiltonian in terms of the extended FFD algebra:
\begin{align}
	H_M
	 & =
	\sum_{m=1}^{M} h_m + \sum_{m=0}^{\floor{M/2}} \hbar{2m+1}\,,
	\label{eq:general-extended-FFD-hamiltonian}
\end{align}
and the corresponding frustration graph is $G_M$.
This Hamiltonian includes the previously introduced Hamiltonian~\eqref{eq:new-model} as a special case with coupling constants $\bbar{1} = \bbar{3} = 0$.
Again, the Hamiltonian for even $M=2M'$ can be obtained from $\eval{H_{2M'+1}}_{b_{2M'+1} = 0}$, whose corresponding frustration graph is $G_{2M'}$.

In Section~\ref{sec:charges}, we prove that the Hamiltonian~\eqref{eq:general-extended-FFD-hamiltonian} has an extensive number of mutually commuting charges.
The condition~\eqref{eq:ac-ca-relation} is crucial for the existence of these higher-order charges and their mutual commutativity.
In Section~\ref{sec:transfer}, we introduce the transfer matrix for the model~\eqref{eq:general-extended-FFD-hamiltonian} and derive its recursion relation.

From Section~\ref{sec:fermion} onward, we consider only the Hamiltonian with $\bbar{1} = \bbar{3} = 0$, i.e., the Hamiltonian~\eqref{eq:new-model}, which can be constructed within the original FFD algebra.
The frustration graph for the Hamiltonian~\eqref{eq:new-model} is shown in Figure~\ref{fig:extended-ffd-frustration-graph-for-M2}.
Nevertheless, the more general Hamiltonian~\eqref{eq:general-extended-FFD-hamiltonian} can also be solved in the same way.

\section{Conserved charges}
\label{sec:charges}
In this section, we give an extensive number of conserved charges for the Hamiltonian~\eqref{eq:general-extended-FFD-hamiltonian}.
We first introduce notations for a graph-theoretical description of the charges.
Then we define the extensive number of charges and provide a proof of the conservation laws.
In Appendix~\ref{app:mutual-commutativity}, we provide the proof for the mutual commutativity of the charges.

\subsection{Notations}
We introduce the graph-theoretical notation used throughout this paper in a similar manner as in~\cite{fermions-behind-the-disguise}.
A graph $G = (V, E)$ consists of a finite vertex set $V$ and an edge set $E \subseteq V^{\times 2}$, where each edge is a 2-element subset of $V$.
All graphs considered are simple (no self-loops or multiple edges).

For any subset $S \subseteq V$, the induced subgraph is denoted by $G[S] = (S, E \cap S^{\times 2})$, which contains all vertices in $S$ and all edges from $E$ with both endpoints in $S$.
We also define the shorthand notation $E_{S} \equiv E \cap S^{\times 2}$.
Throughout this paper, we often refer to a subset of vertices interchangeably with the subgraph it induces.

The open neighborhood of a vertex $v$ is $\Gamma(v) = \{u \in V \mid \{u,v\} \in E\}$, consisting of all vertices neighboring $v$.
The closed neighborhood is $\Gamma[v] = \Gamma(v) \cup \{v\}$, which includes $v$ itself.

An independent set of a graph $G=(V,E)$ is a subset of vertices $S \subseteq V$ which induces a subgraph with no edges, i.e., $G[S] = (S, \emptyset)$.
Equivalently, $S$ is independent if and only if no two vertices in $S$ are adjacent.

A \emph{clique} $K$ in a graph $G=(V,E)$ is a subset of vertices such that every pair of vertices in $K$ is neighboring.

Now we introduce the notation for the frustration graph of our model.
First, we introduce the vertex set for the original FFD Hamiltonian~\eqref{eq:original-FFD-Hamiltonian}: $V_M^{\mathrm{FFD}} \equiv \{1,2,\ldots,M\}$, where each vertex corresponds to a term in the Hamiltonian~\eqref{eq:original-FFD-Hamiltonian}, i.e., the vertex labeled as $j$ denotes the FFD generator $h_j$.
The original edge set of the FFD algebra is given by
\begin{align}
	E_M^{\mathrm{FFD}} \equiv \{(m, n) \in V_M^{\mathrm{FFD}} \times V_M^{\mathrm{FFD}} \mid \{h_m, h_n\} = 0\}.
\end{align}

We next introduce the additional vertices for the additional generators in the Hamiltonian~\eqref{eq:general-extended-FFD-hamiltonian}: $\overline{V}_M \equiv \{\overline{1}, \overline{3}, \ldots, \overline{2\floor{M/2}+1}\}$.
Note that the indices with overlines consist of odd integers.
Each vertex corresponds to a term in the Hamiltonian.
The vertex labeled by an overlined index $\overline{2m+1}$ represents the extra generator $\hbar{2m+1}$.
We also define the union of these two sets as ${V}_M \equiv V_M^{\mathrm{FFD}} \cup \overline{V}_M$.
We introduce the notation for the frustration graph of the extended FFD algebra ${G}_M = ({V}_M, {E}_M)$, where the set of edges is defined by
\begin{align}
	{E}_M \equiv \{(\boldsymbol{i}, \boldsymbol{j}) \in {V}_M \times {V}_M \mid \{h_{\boldsymbol{i}}, h_{\boldsymbol{j}}\} = 0\},
	\label{eq:pseudo-edge}
\end{align}
which correspond to the edges in Figure~\ref{fig:extended-ffd-frustration-graph}.
Thus, $G_M = (V_M, E_M)$ is the frustration graph for the extended FFD algebra introduced in the previous section.
With this notation, the Hamiltonian~\eqref{eq:general-extended-FFD-hamiltonian} is rewritten as
\begin{align}
	H_M = \sum_{\boldsymbol{j} \in V_M} h_{\boldsymbol{j}}.
	\label{eq:new-model-graph}
\end{align}
In the following, we often identify the vertex $\boldsymbol{j} \in V_M$ with the corresponding generator $h_{\boldsymbol{j}}$.

\subsection{Higher-order charges}
In this subsection, we introduce the higher-order charges of the Hamiltonian~\eqref{eq:general-extended-FFD-hamiltonian} and prove their conservation law.

Let $\indepset{M}{k}$ denote the collection of all independent sets of order $k$ in $G_M = (V_M, E_M)$.
For $S \in \indepset{M}{k}$, any two generators from $S$ mutually commute:
\begin{align}
	[h_{\boldsymbol{i}}, h_{\boldsymbol{j}}] = 0\quad \text{for all } \boldsymbol{i}, \boldsymbol{j} \in S.
\end{align}

The charges are defined as sums of products of commuting generators:
\begin{align}
	\Q{M}{k} \equiv \sum_{S \in \indepset{M}{k}} \prod_{\boldsymbol{i} \in S} h_{\boldsymbol{i}}.
	\label{eq:charges}
\end{align}
Note that the first charge is the Hamiltonian itself: $\Q{M}{1} = H_M$.
The highest value for $k$ is the independence number of the frustration graph, denoted by $S_M$, that is, the cardinality of its largest independent vertex set.
For $k > S_M$, we have $Q_M^{(k)} = 0$.
In our case, $S_M = \lfloor (M+2)/3 \rfloor$, similarly to the original FFD case, which can be seen from the recursion equations for the polynomials~\eqref{eq:m2-poly-recursion-even} and~\eqref{eq:m2-poly-recursion-odd}; the proof is given in Appendix~\ref{app:proof-Sk}.

\begin{thm}
	\label{thm:charge-conservation-law}
	The charges defined in~\eqref{eq:charges} are conserved:
	\begin{equation}
		[H_M, \Q{M}{k}] = 0
		\qquad \forall k \in \{1, 2, \ldots, S_M\}.
		\label{eq:conservation-law}
	\end{equation}
\end{thm}

\begin{proof}
	The proof follows the strategy of Lemma~1 in~\cite{fermions-behind-the-disguise}.
	For any independent set ${S} \subset {V}_M$, we define
	\begin{align}
		h_{{S}} \equiv \prod_{\boldsymbol{j} \in {S}} h_{\boldsymbol{j}}.
		\label{eq:shorthand-notation-for-prod}
	\end{align}
	For any $\boldsymbol{j} \in {V}_M$ and independent set ${S} \subset {V}_M$, we have
	\begin{align}
		[h_{\boldsymbol{j}}, h_{{S}}]
		=
		\begin{cases}
			2 h_{\boldsymbol{j}} h_{{S}}
			  &
			\text{if } |\Gamma(\boldsymbol{j}) \cap {S}| \text{ is odd},
			\\[15pt]
			0 & \text{if } |\Gamma(\boldsymbol{j}) \cap {S}| \text{ is even}.
		\end{cases}
	\end{align}
	The commutator of the Hamiltonian with the charge is then
	\begin{align}
		[H_M , \Q{M}{k}]
		 & =
		2
		\sum_{\boldsymbol{j} \in {V}_M}
		\sum_{\substack{
				{S} \in \indepset{M}{k}
		\\
				|\Gamma(\boldsymbol{j}) \cap {S}| = 1
			}
		}
		h_{\boldsymbol{j}} h_{{S}}
		+
		2
		\sum_{\boldsymbol{j} \in {V}_M}
		\sum_{\substack{
				{S} \in  \indepset{M}{k}
		\\
				|\Gamma(\boldsymbol{j}) \cap {S}| = 3
			}
		}
		h_{\boldsymbol{j}} h_{{S}}.
		\label{eq:commutator-hamiltonian-charge}
	\end{align}
	We show that both terms vanish separately.

	\noindent
	\textbf{Cancellation of the first term:} This follows the same argument as in the original FFD case~\cite{fendley-fermions-in-disguise,fermions-behind-the-disguise}.
	Consider $\boldsymbol{j} \in {V}_M$ and ${S} \in  \indepset{M}{k}$ with $|\Gamma(\boldsymbol{j}) \cap {S}| = 1$.
	We can write $h_{{S}} = h_{\boldsymbol{j}'} h_{{S} \setminus \{\boldsymbol{j}'\}}$ where $\boldsymbol{j}' \in \Gamma(\boldsymbol{j}) \cap {S}$.
	The summation includes another pair with $\boldsymbol{j}' \in {V}_M$ and ${S}' = \{\boldsymbol{j}\} \cup ({S} \setminus \{\boldsymbol{j}'\}) \in  \indepset{M}{k}$.
	These pairs cancel:
	\begin{align}
		h_{\boldsymbol{j}} h_{{S}} + h_{\boldsymbol{j}'} h_{{S}'}
		=
		(h_{\boldsymbol{j}} h_{\boldsymbol{j}'} + h_{\boldsymbol{j}'} h_{\boldsymbol{j}}) h_{{S} \setminus \{\boldsymbol{j}'\}}
		=
		0.
		\label{eq:cancelation-first-term}
	\end{align}

	\noindent
	\textbf{Cancellation of the second term:} This is a new feature of our model.
	Consider $\boldsymbol{j} = \boldsymbol{c} \in {V}_M$ and ${S} \in  \indepset{M}{k}$ with $|\Gamma(\boldsymbol{c}) \cap {S}| = 3$, meaning $h_{\boldsymbol{c}}$ is a claw center with leaves $\leaves{\boldsymbol{c}}{t} \subset {S}$ for some $t \in \{A, B\}$.
	Since the leaves form an independent set, they can be simultaneously included in ${S}$.

	We can decompose $h_{{S}} = h_{\vleafdep{t,1}_{\boldsymbol{c}}} h_{\vleafdep{t,2}_{\boldsymbol{c}}} h_{\vleafdep{3}_{\boldsymbol{c}}} h_{{S} \setminus \leaves{\boldsymbol{c}}{t}}$.
	Define $\boldsymbol{c}' = \overline{2m+3}$ if $\boldsymbol{c} = 2m-1$, and vice versa.
	For $\boldsymbol{d} \in \{\boldsymbol{c}, \boldsymbol{c}'\}$ and $\tau \in \{A, B\}$, define ${S}_{\boldsymbol{d}, \tau} \equiv \leaves{\boldsymbol{d}}{\tau} \cup ({S} \setminus \leaves{\boldsymbol{c}}{t})$.

	The summation includes all four pairs $(\boldsymbol{d}, {S}_{\boldsymbol{d}, \tau})$ (see Appendix~\ref{app:cc-part-cancellation} for details).
	These cancel as follows:
	\begin{align}
		\sum_{\boldsymbol{d} \in \{\boldsymbol{c}, \boldsymbol{c}'\}} \sum_{\tau \in \{A, B\}}
		h_{\boldsymbol{d}} h_{{S}_{\boldsymbol{d}, \tau}}
		 & =
		\sum_{\tau \in \{A, B\}}
		\left[
			h_{\boldsymbol{c}}
			h_{\vleafdep{\tau,1}_{\boldsymbol{c}}}
			h_{\vleafdep{\tau,2}_{\boldsymbol{c}}}
			h_{\vleafdep{3}_{\boldsymbol{c}}}
			+
			h_{\boldsymbol{c}'}
			h_{\vleafdep{\tau,1}_{\boldsymbol{c}'}}
			h_{\vleafdep{\tau,2}_{\boldsymbol{c}'}}
			h_{\vleafdep{3}_{\boldsymbol{c}'}}
			\right]
		h_{{S} \setminus \leaves{\boldsymbol{c}}{t}}
		\nonumber \\
		 & =
		(\a{2m-1} \C{2m+2} - \C{2m} \a{2m+3}) h_{{S} \setminus \leaves{\boldsymbol{c}}{t}}
		\nonumber \\
		 & =
		0,
	\end{align}
	where the last equality follows from the relation~\eqref{eq:ac-ca-relation}.
\end{proof}

The next theorem establishes that the higher-order charges are mutually commuting:

\begin{thm}
	\label{thm:mutual-commutativity}
	The charges~\eqref{eq:charges} mutually commute:
	\begin{align}
		[\Q{M}{k}, \Q{M}{l}] = 0 \quad \forall k, l \in \{1, 2, \ldots, S_M\}.
	\end{align}
\end{thm}

The proof of Theorem~\ref{thm:mutual-commutativity} is given in Appendix~\ref{app:mutual-commutativity}.
While it requires a more careful analysis of the frustration graph, the essential point remains the same: the relation~\eqref{eq:ac-ca-relation} plays a crucial role.

\subsection{Recursion for the higher-order charges}
We here show the recursion relations that the higher-order charges satisfy.
We can define similar quantities to those in~\eqref{eq:charges} when given a frustration graph $G$ that is a subgraph of $G_M = (V_M, E_M)$ for some $M$.
Let $\indepset{G}{k}$ denote the collection of all independent sets of order $k$ in $G$.
We also define a \emph{pseudo-charge} that reproduces the charges~\eqref{eq:charges} for the appropriate choice of frustration graph $G$:
\begin{align}
	\Q{G}{k}
	\equiv
	\sum_{S \in \indepset{G}{k}} h_S,
	\label{eq:pseudo-charges}
\end{align}
where we have used the abbreviated notation~\eqref{eq:shorthand-notation-for-prod}.
Note that $\{\Q{G}{k}\}_{k=1,2,\ldots}$ are not necessarily mutually commuting charges.
They form a set of mutually commuting charges when the frustration graph is free fermionic, as in Figure~\ref{fig:free-fermionic-frustration-graph} and Figure~\ref{fig:frustrationgraph-for-even-Q-recursion}.

The pseudo-charge~\eqref{eq:pseudo-charges} satisfies the following recursion:
\begin{align}
	\Q{G}{k}
	 & =
	\Q{G \setminus K}{k}
	+
	\sum_{\boldsymbol{j} \in K} h_{\boldsymbol{j}} \Q{G \setminus \Gamma[\boldsymbol{j}]}{k-1},
	\label{eq:general-Q-recursion}
\end{align}
where $K$ is a clique in $G$.
The proof of~\eqref{eq:general-Q-recursion} follows trivially from the definition of the pseudo-charge~\eqref{eq:pseudo-charges}.
Please refer to~\cite{fermions-behind-the-disguise,unified-graph-th} for a more detailed explanation.

Using the recursion for the pseudo-charges~\eqref{eq:general-Q-recursion}, we can derive the following recursion for the charges~\eqref{eq:charges}:
\begin{equation}
	\begin{split}
		\Q{2m}{k}
		 & =
		\Q{2m-1}{k}
		+
		h_{2m} \Q{2m-3}{k-1}
		+
		h_{\overline{2m+1}} \Q{2m-4}{k-1}
		\\
		 & \qquad
		+
		h_{2m-2} h_{\overline{2m+1}} \Q{2m-6}{k-2}
		+
		h_{2m-5} h_{\overline{2m+1}} \Q{2m-8}{k-2},
		\label{eq:m2-charge-recursion-even}
	\end{split}
\end{equation}
and
\begin{equation}
	\Q{2m-1}{k}
	=
	\Q{2m-2}{k}
	+	h_{2m-1} \Q{2m-4}{k-1}.
	\label{eq:m2-charge-recursion-odd}
\end{equation}
The initial conditions are $\Q{0}{0} = \Q{-1}{0} = \Q{-2}{0} = \mathbb{1}$ and $\Q{m}{k} = 0$ for $k<0$ and $\Q{m}{k} = 0$ for $m < 0$ and $k>0$.
The proofs of~\eqref{eq:m2-charge-recursion-even} and~\eqref{eq:m2-charge-recursion-odd} are given in Appendix~\ref{app:charge-recursion}.
From these charge recursions, we can confirm that $S_M = \floor{(M+2)/3}$.
\section{Transfer matrix\label{sec:transfer}}
In this section, we introduce the transfer matrix for the Hamiltonian~\eqref{eq:general-extended-FFD-hamiltonian}, from which we can derive the polynomial for the excitation spectrum.
We define the transfer matrix from the conserved charges~\eqref{eq:charges} as
\begin{align}
	\label{eq:transfer-mat}
	T_M(u) \equiv \sum_{k \ge 0} (-u)^k \Q{M}{k},
\end{align}
where $u$ is the spectral parameter.

The transfer matrix satisfies the following recursion relations for even and odd $M$, respectively:
\begin{equation}
	\begin{split}
		T_{2m}(u) = & T_{2m-1}(u) -
		u h_{2m}T_{2m-3}(u)
		- u h_{\overline{2m+1}} T_{2m-4}(u)                        \\
		            & +u^2 h_{2m-2} h_{\overline{2m+1}}T_{2m-6}(u)
		+u^2 h_{2m-5} h_{\overline{2m+1}} T_{2m-8}(u),
		\label{eq:m2-transfermat-recursion-even}
	\end{split}
\end{equation}
and
\begin{equation}
	T_{2m-1}(u)
	=
	T_{2m-2}(u)
	-
	u h_{2m-1}T_{2m-4}(u).
	\label{eq:m2-transfermat-recursion-odd}
\end{equation}
The initial conditions are $T_{0}(u) = T_{-2}(u) = \mathbb{1}$, with the convention that $h_{m} = h_{\overline{2m-1}} = 0$ for all $m \le 0$.
The recursion relations for the transfer matrix~\eqref{eq:m2-transfermat-recursion-even} and~\eqref{eq:m2-transfermat-recursion-odd} follow directly from the recursions for the charges~\eqref{eq:m2-charge-recursion-even} and~\eqref{eq:m2-charge-recursion-odd}.

In the limit of the FFD model we set $h_{\bar k}\equiv 0$, and we get the recursion relation of the FFD model~\cite{fendley-fermions-in-disguise}:
\begin{equation}
	T_{k}(u) = T_{k-1}(u) - u {h}_{k}T_{k-3}(u).
\end{equation}

Combining the two equations above, we obtain a recursion relation that only contains the even $M$ transfer matrix:
\begin{align}
	T_{2m}(u)
	= &
	T_{2m-2}(u)
	-u \s{2m} T_{2m-4}(u)
	+u^2 \a{2m-1} T_{2m-6}(u)
	+u^2 \C{2m-2} T_{2m-8}(u)
	\,,
	\label{eq:m2-transfermat-recursion-only-even}
\end{align}
where we defined
\begin{align}
	\s{2m} \equiv {h}_{2m-1} + {h}_{2m} + \hbar{2m+1}\,.
\end{align}
One can further simplify this recursion substantially to read
\begin{equation}\label{eq:shortrec}
	T_{2 m}(u) = T_{2 m -2}(u) - \frac{u}{2} \lbrace \s{2m}, T_{2 m -2}(u)\rbrace,
\end{equation}
as shown in Appendix~\ref{subsec:mainfermion}.

The following theorem provides the polynomial whose roots determine the free fermionic spectrum:
\begin{thm}\label{thm:poly-recursion}
	The transfer matrix satisfies the following simple inversion relation
	\begin{equation}\label{eq:poly}
		T_M(u)T_M(-u) = P_M(u^2) \cdot \mathbb{1}
	\end{equation}
	where $P_M(z)$ is a polynomial of degree $S_M$.
	For even $M$ it satisfies the following recursion:
	\begin{align}
		P_{2m}(u^2)
		= &
		P_{2m-2}(u^2)
		-
		u^2 \s{2m}^2 P_{2m-4}(u^2)
		+
		u^4 \a{2m-1}^2 P_{2m-6}(u^2)
		-
		u^4 \C{2m-2}^2 P_{2m-8}(u^2)
		\,.
		\label{eq:m2-poly-recursion-even}
	\end{align}
	For the case of $P_{M}(u^2)$ with odd $M$, we have
	\begin{align}
		P_{2m-1}(u^2)
		= &
		P_{2m-2}(u^2)
		-
		u^2 {b}^2_{2m-1}P_{2m-4}(u^2)
		\,.
		\label{eq:m2-poly-recursion-odd}
	\end{align}
	The starting point is $P_{0}(u^2)= P_{-2}(u^2) = 1$ such that $b_{m}=b_{\overline{2m-1}}=0, ~\forall m\le 0$ formally.
\end{thm}
In Appendix~\ref{app:polyproof}, we provide the proof of Theorem~\ref{thm:poly-recursion}.
We note that the following quantities appearing in~\eqref{eq:m2-poly-recursion-even} in Theorem~\ref{thm:poly-recursion} are all scalars:
\begin{align}
	\s{2m}^2   & = b_{2m-1}^2 + b_{2m}^2 + \bbar{2m+1}^2
	\,,
	\\
	\a{2m-1}^2 & = (b_{2m-3}b_{2m} \pm b_{2m-2} \bbar{2m+1})^2\,,
	\\
	\C{2m-2}^2 & = b_{2m-3}^2 \bbar{2m+3}^2
	\,,
\end{align}
where the sign factor in the second line is determined by which sign we choose for the central element $\mu_{2m-1}$ in~\eqref{eq:mu-scalar}.

In Appendix~\ref{app:proof-Sk}, we give the proof for $S_M = \floor{(M+2)/3}$ using the recursion for the polynomial~\eqref{eq:m2-poly-recursion-even} and~\eqref{eq:m2-poly-recursion-odd}.

\section{Free fermions}

\label{sec:fermion}

In this section, we derive the free fermionic operators for the Hamiltonian~\eqref{eq:new-model}.
First, we introduce the simplicial mode and corresponding simplicial clique.
Then, we construct the free fermion modes using the simplicial mode and simplicial clique.

The simplicial mode and clique are localized around one of the boundaries of the spin chain.
Correspondingly, there are two ways to proceed, by choosing either boundary.
The Hamiltonian is not space reflection symmetric, therefore the two cases will be genuinely different.
This is in contrast with the situation in the original FFD model, where the two boundaries of the chain were equivalent.

Although in the following we only treat the Hamiltonian~\eqref{eq:new-model}, whose frustration graph is given in Figure~\ref{fig:extended-ffd-frustration-graph-for-M2} and which is a special case of the more general Hamiltonian~\eqref{eq:general-extended-FFD-hamiltonian}, the free fermion solution can be constructed in the same way for the general case.

\subsection{Simplicial cliques and simplicial modes}
We now introduce two key structures that are central to our free-fermion solution, referring to the explanation in~\cite{fendley-fermions-in-disguise,unified-graph-th}.
A \emph{simplicial clique} $K_s$ is a clique with the additional property that, for every vertex $v \in K_s$, the neighborhood of $v$ outside $K_s$ induces a clique in $G \setminus K_s$.
More precisely, for each $v \in K_s$, the set $\Gamma(v) \setminus K_s$ forms a clique.
A graph is called \emph{simplicial} if it contains at least one simplicial clique.
For a simplicial clique $K_s$ in $G$, we can define for each $v \in K_s$ the clique $K_v = (\Gamma(v) \setminus K_s) \cup \{v\}$, such that $\Gamma[v] = K_s \cup K_v$.
All frustration graphs treated in this paper are simplicial.

Associated with each simplicial clique, we define a \emph{simplicial mode}.
A simplicial mode with respect to $K_s$ is an additional generator $\chi$ which anticommutes only with the generators corresponding to vertices in $K_s$.
The simplicial mode acts as a seed for constructing the fermionic eigenmodes.
By repeatedly commuting $\chi$ with the Hamiltonian, we generate a Krylov subspace whose structure is intimately connected to the induced path tree rooted at the simplicial clique.
This connection between the algebraic structure (commutation relations) and the graph structure (induced paths) is what enables the exact free-fermion solution.

In the following, we show the simplicial mode for our frustration graph (Figure~\ref{fig:extended-ffd-frustration-graph-for-M2}).

\subsection{\label{subsec:rightedge}Right-end simplicial mode}

First, we define the simplicial mode at the right-end of the chain.
In this case, the definition of simplicial mode differs for even $M$ and odd $M$.
The simplicial mode for even $M$ is
\begin{align}
	\qty{\chi_M, h_M} = \qty{\chi_M, h_{M-1}} = 0\,,\quad \qty[\chi_M, h_l] = 0\quad(l<M-1)\,.
\end{align}
The simplicial mode for odd $M$ is the same as the original FFD case~\cite{fendley-fermions-in-disguise}:
\begin{align}
	\qty{\chi_M, h_M} = 0\,,\quad \qty[\chi_M, h_l] = 0\quad(l<M)\,.
\end{align}
Then, the corresponding simplicial clique is
\begin{align}\label{eq:simplicial}
	K_{M}^{(s)} \equiv
	\begin{cases}
		\qty{M-1, M, \overline{M+1}} & \text{$M$ is even}
		\\
		\qty{M}                      & \text{$M$ is odd}
	\end{cases}
	\,.
\end{align}

\subsection{\label{subsec:leftedge}Left-end simplicial mode}
Next, we explain the simplicial mode at the left end of the chain.
In this case, the definition of the simplicial mode is the same for both even and odd $M$.
The simplicial mode is defined by
\begin{align}
	\qty{\chi_M, h_1} = \qty{\chi_M, h_{2}} = 0\,,\quad \qty[\chi_M, h_l] = 0\quad(2<l)\,.
\end{align}
Then, the corresponding simplicial clique is
\begin{align}
	K_{M}^{(s)} \equiv
	\qty{1, 2}
	\,.
\end{align}

\subsection{Construction of the free fermion mode}

Here we construct the free fermion mode using the simplicial clique introduced above.
We can use either of the left-end and right-end simplicial mode.

We would like to prove that the operators
\begin{equation}\label{eq:fermdef}
	\Psi_k \equiv \frac{T_M(-u_k)\chi_M T_M(u_k)}{\mathcal{N}_k}
\end{equation}
act as fermionic ladder operators that satisfy the eigenvalue problem of the adjoint action with the Hamiltonian as
\begin{equation}\label{eq:eigenpsi}
	[H, \Psi_k ] = 2\epsilon_k \Psi_k,
\end{equation}
where $\epsilon_k = u_k^{-1}$ are the energies of the fermion modes, and the $u_k$-s are the roots of the polynomials $P_M(u_k^2)=0$ in \eqref{eq:poly}.
As the latter come in opposite-sign pairs, we may define $u_{-k} = -u_k$ where $u_k>0$ for $k=1,2,\ldots,S_M$.
The $\Psi^\dagger_k = \Psi_{-k}$ are creating while $\Psi_k$ are annihilating a fermionic mode with positive energy for $k\geq 1$.
They also need to anticommute as
\begin{equation}\label{eq:anticomm}
	\lbrace \Psi_k , \Psi_{-k'}\rbrace = \delta_{k,k'} \cdot \mathbb{1}
\end{equation}
where the $k=k'$ case sets the normalization constant $\mathcal{N}_k$ in \eqref{eq:fermdef}.
Below we will show how the basic properties of the simplicial mode (edge operator) and the recursion for the transfer matrix lead to the above statements.

We start from the trivial identity
\begin{align}\label{eq:eigen}
	\qty[H_{M}, T_{M}(-u) \chi_M T_{M}(u)] = 2 \sum_{v \in K_{M}^{(s)}}T_{M}(-u) h_v \chi_M T_{M}(u)
	\,,
\end{align}
that follows from commuting $H_M$ through the transfer matrices due to Theorem~\ref{thm:charge-conservation-law}, and the properties of the simplicial mode.
Then we apply the next theorem to arrive at \eqref{eq:eigenpsi}.
\begin{lem}\label{lem:mainfermion}
	%	The free fermionic mode can be made from the following identity:
	The following modified inversion relation holds:
	\begin{align} \label{eq:mainfermion}
		T_{M}(-u)\qty(1 - u\sum_{v \in K_{M}^{(s)}}h_v )\chi_M T_{M}(u)
		=
		P_M(u^2)
		\qty(1 + u\sum_{v \in K_{M}^{(s)}}h_v) \chi_M
		\,.
	\end{align}
\end{lem}
The proof of Lemma~\ref{lem:mainfermion} is given in~Appendix \ref{subsec:mainfermion}.
Then, we have the following theorem:
\begin{thm}
	The fermionic operators defined by \eqref{eq:fermdef} satisfy the anticommutation relations and also the commutation relation \eqref{eq:eigenpsi}.
\end{thm}
\begin{proof}
	Replacing the term on the l.h.s. of \eqref{eq:mainfermion} that appears on the r.h.s. of \eqref{eq:eigen} using the latter relation gives
	\begin{align}
		\qty[H_{M}, T_{M}(-u) \chi_M T_{M}(u)] = 2 u^{-1} T_{M}(-u) \chi_M T_{M}(u)
		-  2 P_M(u^2) \qty(u^{-1} + \sum_{v \in K_{M}^{(s)}}h_v) \chi_M,
	\end{align}
	and evaluating it at a root $u_k$ of the polynomial $P_M(u_k^2)=0$ we get \eqref{eq:eigenpsi} (up to the undetermined normalization factor $\mathcal{N}_k$).
	For the proof of the canonical anticommutations in \eqref{eq:anticomm} see Appendix \ref{subsec:CAR}.
\end{proof}

As the number of fermionic modes obtained this way depends on the number of roots of the polynomial $P_M(u^2)$, their number is $S_M$.
The number of the different energy levels is thus $2^{\lfloor (M+2)/3\rfloor}$ and since the dimension of the Hilbert space is $2^M$ it means the energy eigenstates are exponentially degenerate.
An understanding of these degeneracies in the case of the FFD model has been achieved in \cite{eric-lorenzo-ffd-1}.

Following the arguments of Section 3.4 in~\cite{fendley-fermions-in-disguise}, and using \eqref{eq:poly} and \eqref{eq:tmpsiflip}, it is possible to show that the basis of fermionic modes is complete in the sense that one can reconstruct the Hamiltonian and the transfer matrix itself as
\begin{equation}\label{eq:reconstr}
	H = \sum_{k=1}^{S_M} \epsilon_k [\Psi_k,\Psi_{-k}],\qquad T_M(u) = \prod_{k=1}^{S_M} \left(1 - u \epsilon_k [\Psi_k,\Psi_{-k}] \right).
\end{equation}

Note that similarly to \cite{sajat-ffd-corr} we may extend our fermion algebra by a Majorana zero mode $\Psi_0$, that only exists for certain system sizes.
If it exists, it may be defined as the $u\to\infty$ limit
\begin{equation}\label{eq:zerodef}
	\Psi_{0}\equiv\frac{1}{2 C_0}\left(\chi_M+\lim_{u\to\infty}\frac{T_{M}(-u)\chi_M T_{M}(u)}{P_{M}(u^{2})}\right),
\end{equation}
where $C_0$ is a known constant defined in \eqref{eq:CcoeffsZero}.
It satisfies the following properties:
\begin{equation}
	[H,\Psi_0] = 0, \quad \Psi_0^2 = \mathbb{1}, \quad \text{and}\quad \{\Psi_0,\Psi_k\} = 0 ~\text{for}~k=\pm 1,\pm 2, \ldots ,\pm S_M.
\end{equation}
%In summary, the fermion modes anticommute as $\{\Psi_{k},\Psi_{k'}\}=2^{\delta_{k,0}}\delta_{k+k'}\mathbb{1}$ with $k,k'\in\{-S,\ldots,-1,0,1,\ldots,S\}$.
Although the zero mode is not needed for the reconstruction of the Hamiltonian and the transfer matrix in \eqref{eq:reconstr} (due to its zero energy $\epsilon_0 = 0$), as we will see in Section \ref{sec:corr}, it appears in the decomposition of the simplicial mode.

\section{\label{sec:corr}Correlation functions}

So far we obtained the solution of the Hamiltonian in terms of the fermionic operators.
As a next step we can also compute certain correlation functions, namely the correlations of those operators which take a sufficiently simple form in terms of the fermions.
This was achieved in \cite{sajat-ffd-corr} for the FFD model, and now we generalize that formalism to our model.

The first step is solving the inverse problem for the edge operator in Subsection \ref{subsec:rightedge}, i.e.
decomposing it into the fermion modes:
\begin{equation}
	\chi_M=\sum_{k=-S_M}^{S_M}C_{k}\Psi_{k},\label{eq:chidecomp}
\end{equation}
where the coefficients $C_{k}$ read
\begin{align}
	C_{j} & =\sqrt{\frac{P_{M-r_{M}}(u_{j}^{2})}{-u_{j}^{2}P_{M}'(u_{j}^{2})}}\quad\text{for $j\neq0$}\,,
	\label{eq:Ccoeffs}
\end{align}
and
\begin{align}
	C_{0} & =\sqrt{\lim_{u\to\infty}\frac{P_{M-r_{M}}(u^{2})}{P_{M}(u^{2})}}
	\,,
	\label{eq:CcoeffsZero}
\end{align}
with $r_{2m} = 2$ and $r_{2m+1}=1$.
The zero mode is not present if $S_{M-r_M} < S_M$, in this case $C_{0}=0$.
We may calculate infinite temperature correlators
$\langle \cdot \rangle \equiv \mathrm{Tr}(\cdot)/\mathrm{Tr}(\mathbb{1})$
of certain operators defined by the recursion
\begin{equation}\label{eq:krylov}
	o_{j}=\frac{1}{2}[H,o_{j-1}],\quad o_{0}\equiv\chi_M
\end{equation}
that are elements of the Krylov-subspace generated by repeated action of the commutator with the Hamiltonian.
Using \eqref{eq:chidecomp}, their time evolution reads
\begin{equation}\label{eq:ojt}
	o_{j}(t)=e^{iHt}o_{j}e^{-iHt}=\sum_{k=-S_M}^{S_M}C_{k}\epsilon_{k}^{j}e^{2i\epsilon_{k}t}\Psi_{k},
\end{equation}
and it is useful to define their anticommutator that is proportional to the identity
\begin{equation}
	\{o_{j_{1}}(t_{1}),o_{j_{2}}(t_{2})\}\equiv2B_{j_{1},j_{2}}(t_{1}-t_{2})\mathbb{1}.
\end{equation}
It is easy to show that Wick theorem leads to
\begin{align}\label{eq:Wick}
	\langle o_{j_{1}}(t_{1})o_{j_{2}}(t_{2})o_{j_{3}}(t_{3})o_{j_{4}}(t_{4})\rangle & =B_{j_{1},j_{2}}(t_{1,2})B_{j_{3},j_{4}}(t_{3,4})-B_{j_{1},j_{3}}(t_{1,3})B_{j_{2},j_{4}}(t_{2,4}) \nonumber \\
	                                                                                & \hspace{5em} +B_{j_{1},j_{4}}(t_{1,4})B_{j_{2},j_{3}}(t_{2,3})
\end{align}
where $t_{i,j}\equiv t_i-t_j$, and the function
\begin{equation}
	B_{j_{1},j_{2}}(t)=\frac{(-1)^{j_{2}}}{(2i)^{j_{1}+j_{2}}}\partial_{t}^{j_{1}+j_{2}}B(t)
\end{equation}
can be computed from the self-correlator of the edge operator, which reads
\begin{equation}
	B(t)\equiv B_{0,0}(t)=\langle\chi_M(t)\chi_M(0)\rangle=\sum_{j=0}^{S_M}C_{j}^{2}\cos(2\epsilon_{j}t).\label{eq:chicorr}
\end{equation}
One of the simplest operators one may construct from bilinears of the Krylov subspace looks like
\begin{equation}\label{eq:simplicialop}
	\mathbb{S}_M \equiv o_{1}o_{0}=\begin{cases}
		h_{2m-1}                            & M=2m-1 \\
		h_{2m-1}+h_{2m}+h_{\overline{2m+1}} & M=2m
	\end{cases}
\end{equation}
that directly corresponds to the simplicial clique in \eqref{eq:simplicial} of the frustration graph.
The self-correlator of this operator is then
\begin{equation}
	\langle\mathbb{S}_M(t)\mathbb{S}_M(0)\rangle=\frac{1}{4}\left(\dot{B}^{2}(t)-\ddot{B}(t)B(t)\right)
\end{equation}
where the dot means time-derivative and the correlator may be calculated very efficiently using \eqref{eq:chicorr} and \eqref{eq:Ccoeffs} after solving for the roots of the polynomial $P_{M}(u^{2})$.

In case of Floquet time evolution with the unitary circuit $\mathcal{V}(\delta) \equiv T_M(-i\delta)/\sqrt{P_M(-\delta^2)}$ the fermion modes pick up the phase
\begin{equation}
	\mathcal{V}^\dagger(\delta) \Psi_j \mathcal{V}(\delta) = \frac{T_M(-i \delta) \Psi_j T_M(i \delta)}{P_M(-\delta^2)} = \frac{u_j + i \delta}{u_j - i \delta} \Psi_j = e^{2 i \arctan{\left(\epsilon_j \delta\right)}} \Psi_j
\end{equation}
after each application of the circuit.
After the $N^{\text{th}}$ time step (that is, total time $t=N \delta $ in the Trotterization picture of \cite{sajat-ffd-corr}), the phase angle acquired is $\theta_j = 2 N \arctan(\epsilon_j \delta)$.
This is in comparison to real time evolution, where the same angle is simply $\theta_j = 2 \epsilon_j t$.
Analogously to \eqref{eq:ojt} the Krylov-basis elements can be time-evolved as $o_j(t = N \delta) = \left(\mathcal{V}^\dagger_M(\delta)\right)^N o_j \mathcal{V}^N_M(\delta)$ and formula \eqref{eq:Wick} describes their correlators by a formal substitution of the phase angles $2 \epsilon_k t \to 2 N \arctan(\epsilon_k \delta)$.
That means the building blocks for the correlators for both real and Floquet time evolution are
\begin{equation}
	B_{j_1,j_2}(t) = \sum_{k=1}^{S_M} C_k^2 \epsilon_k^{j_1+j_2} \frac{(-1)^{j_2} e^{i \theta_k}+ (-1)^{j_1} e^{-i \theta_k}}{2}, \quad\forall ~j_1+j_2 > 0,
\end{equation}
while $B(t) =\sum_{k=0}^{S_M} C_k^2 \cos(\theta_k)$, with the respective phase angles $\theta_k$ as explained above.

\section{Factorization of the transfer matrix}

\label{sec:fact}

In this section, we show that a special case of the transfer matrix~\eqref{eq:transfer-mat} becomes proportional to a quantum unitary circuit with free fermions in disguise reported in~\cite{sajat-floquet}.
In fact, finding the Hamiltonian behind the quantum circuit of~\cite{sajat-floquet} was one of the motivations for the present work.

First, we consider the factorization of our transfer matrix into the product of local operators as
\begin{align}
	\mathcal{V}_{M}=G_{M}\cdot G_{M}^{\top}
	\,,
	\label{eq:quantum-circuit-M2}
\end{align}
where
\begin{align}
	G_{2k}= & (g_{2}g_{4}\cdots g_{2k})(g_{1}g_{3}g_{5}\cdots g_{2k-1}) \\
	=       & (g_{2}g_{1})(g_{4}g_{3})\cdots(g_{2k}g_{2k-1})
	\,,
	\label{eq:G-factor-M2}
\end{align}
and the local gates $g_j$ are defined as
\begin{align}
	g_j & \equiv \cos(\theta_j / 2) + \sin(\theta_j / 2)b_j^{-1} h_j
	\,,
	\label{eq:local-gate}
\end{align}
where for odd $M$, we define $G_{2k-1} \equiv \eval{G_{2k}}_{\theta_{2k} = 0}$.
The subscript $\top$ denotes the transpose operation, with $g_j^\top = g_j$.
Hereafter, for notational clarity, the explicit dependence on $\theta_j$ will be suppressed unless otherwise specified.

Our notation differs slightly from the original paper~\cite{sajat-floquet}.
The notation used in Ref.~\cite{sajat-floquet} can be recovered by $\eval{\mathcal{V}_{2k}}_{\theta_1 = \theta_{2k} = 0}$ and subsequently applying the index shift $j \rightarrow j-1$.

The formula \eqref{eq:quantum-circuit-M2} is analogous to the factorization of the transfer matrix in the original FFD model~\cite{fendley-fermions-in-disguise}.
The only difference is the ordering of the local operators in \eqref{eq:G-factor-M2}.
The work~\cite{sajat-floquet} asked which operator orderings are compatible with the free fermionic structure, and the special ordering of~\eqref{eq:G-factor-M2} was found to be the simplest possibility different from the original one in~\cite{fendley-fermions-in-disguise}.

The local gates in~\eqref{eq:local-gate} are not unitary.
However, one obtains unitary gates using the transformation of Eq.~(5.10) in~\cite{sajat-floquet}, which amounts to an analytic continuation in the angles.
Using this transformation, we provide a rigorous proof of the conjecture in~\cite{sajat-floquet} that one of the quantum circuits proposed in~\cite{sajat-floquet} is indeed free fermionic.
\red{The same circuit was also proved to be free fermionic by a different method in~\cite{szasz-schagrin-qst-circuits}.}

\begin{thm}
	If the parameters $b_j$ are chosen arbitrarily, the parameters $\beta_1$ and $\beta_3$ are set by
	\begin{align}
		\beta_{1} & = \frac{1}{1-b_2^2} \,,\qquad
		%	\label{eq:beta-1-circuit-case}
		\beta_{3}  = \frac{1}{1- b_2^2 - b_4^2} \,.
		\label{eq:beta-3-circuit-case}
	\end{align}
	and the remaining $\beta_{2j-1}$ are calculated using \eqref{eq:beta_relation}, then the transfer matrix can be factorized at the special point $u=-1$.
	More precisely we obtain
	\begin{align}
		T_{2k}(u=-1)=	\mathcal{V}_{2k}
		\,,
		\label{eq:transfer-matrices-coincidence}
	\end{align}
	where the angles are determined by the recursion relation:
	\begin{align}
		\begin{aligned}
			\sin\theta_{2j}   & = \frac{b_{2j}}{\cos\theta_{2j-2}}\,,
			\\
			\sin\theta_{2j-1} & = \frac{b_{2j-1}}{\cos\theta_{2j-2}\cos\theta_{2j-3}\cos\theta_{2j}}\,,
		\end{aligned}
		\label{eq:angles-recursion}
	\end{align}
	where $\theta_{0} = \theta_{-1} = \theta_{M+1} = 0$.
	For the relation between the parameters $\beta_{2j-1}$ and the angles we obtain the direct relation
	\begin{align}
		\beta_{2j-1} & = \frac{1}{\cos^2\theta_{2j-2} \cos^2\theta_{2j}}
		\,.
		\label{eq:beta-from-algles}
	\end{align}
	which is compatible with \eqref{eq:beta_relation}.
\end{thm}
\begin{proof}
	We derive the recursion equation for the quantum circuit~\eqref{eq:quantum-circuit-M2} and show that it is the same as the recursion relations derived above in~\eqref{eq:m2-transfermat-recursion-odd} and~\eqref{eq:m2-transfermat-recursion-even}.
	We can prove the transfer matrix satisfies the following recursion relation:
	\begin{align}
		\mathcal{V}_{2k}
		 & =
		\mathcal{V}_{2k-2}+(c_{2k-2}h_{2k}^{\prime}+c_{2k}c_{2k-2}c_{2k-3}h_{2k-1}^{\prime}+c_{2k-5}h_{2k-2}^{\prime}h_{2k-3}^{\prime}h_{2k}^{\prime})\mathcal{V}_{2k-4}
		\nonumber \\
		 & \qquad
		+c_{2k-4}c_{2k-5}h_{2k}^{\prime}h_{2k-3}^{\prime}\mathcal{V}_{2k-6}+c_{2k-4}c_{2k-5}c_{2k-6}c_{2k-7}h_{2k-5}^{\prime}h_{2k-2}^{\prime}h_{2k-3}^{\prime}h_{2k}^{\prime}\mathcal{V}_{2k-8}
		\,,
		\label{eq:m2-transfermat-recursion-circuit-even}
		\\
		\mathcal{V}_{2k-1}
		 & =
		\mathcal{V}_{2k-2} + c_{2k-2}c_{2k-3}h_{2k-1}^{\prime}\mathcal{V}_{2k-4}
		\,.
		\label{eq:m2-transfermat-recursion-circuit-odd}
	\end{align}
	where we denote $c_j \equiv \cos\theta_{j}$ and $h_{j}^{\prime} \equiv \sin\theta_j  b_j^{-1} h_j$.
	The equation in the second line~\eqref{eq:m2-transfermat-recursion-circuit-odd} is identical to that of the original FFD paper~\cite{fendley-fermions-in-disguise} and can be derived in the same manner.
	The proof of the first line~\eqref{eq:m2-transfermat-recursion-circuit-even} is provided in the Appendix~\ref{app:proof-of-circuit-recursion}.
	By eliminating the angles $\theta_j$ from the above equation using~\eqref{eq:angles-recursion} and~\eqref{eq:beta-from-algles}, we can identify the above recursion~\eqref{eq:m2-transfermat-recursion-circuit-even} with the recursion relation for the transfer matrix~\eqref{eq:m2-transfermat-recursion-even} with $u=-1$.

	Since the recursion equation is identical and the initial condition is the same, we can conclude that the two transfer matrices coincide due to the uniqueness of the solution to the recursion equation, thus proving~\eqref{eq:transfer-matrices-coincidence}.
\end{proof}

We also attempted to find a factorization of the transfer matrix for generic values of the spectral parameter, but we have not yet found a general solution.
This is left for further work.

\section{Conclusion}
\label{sec:concl}

We introduced a new spin chain model that can be solved with free fermions in disguise, despite its frustration graph containing both claws and even holes—structures that violate the sufficient conditions established in previous works~\cite{fermions-behind-the-disguise,unified-graph-th}.
The free fermionic solvability is achieved through carefully constructed algebraic relations among the coupling constants, which we derived by introducing an extension of the FFD algebra.

Our Hamiltonian can be formulated solely within the original FFD algebra for a special case, distinguishing it from other extensions such as \cite{sajat-FP-model}.
When specific coupling constants are set to zero, our model reduces to the original FFD model.
Moreover, for a particular choice of parameters, the transfer matrix can be factorized into a product of three-site operators, yielding a quantum circuit first introduced in \cite{sajat-floquet}.
That circuit was conjectured to be free fermionic, and our construction now proves that conjecture.
\red{Relatedly, Sz{\'a}sz-Schagrin et al.~gave another proof of the free-fermionic nature of this circuit~\cite{szasz-schagrin-qst-circuits}.}

We also attempted to factorize the transfer matrix for the general case but did not succeed.
This remains an open problem.

In future work, it would be interesting to prove the free fermionic nature of other circuits presented in \cite{sajat-floquet}, and eventually to develop a more general theory for such circuits.

An even more challenging task is to derive sufficient conditions for free fermionic solvability when the standard conditions of \cite{fermions-behind-the-disguise,unified-graph-th} are not applicable.
The present model and \cite{sajat-FP-model} are examples where special algebraic relations guarantee integrability.
\red{In our model, once the central even-hole elements are gauged out, the Hamiltonian can be expressed within Fendley's original FFD algebra~\cite{fendley-fermions-in-disguise}.
In this sense, the problem reduces to a ``minimal'' frustration graph in that algebra, now realized through multibody interactions and through algebraic relations among the coupling constants that resolve the claw obstruction.
It would be desirable to develop a general theory in which further even-hole-free and claw-free graphs arise as minimal frustration graphs under such reductions.}

\section*{Acknowledgements}
The authors thank H. Katsura for helpful discussions.
K.F. was supported by MEXT KAKENHI Grant-in-Aid for Transformative Research Areas A “Extreme Universe” (KAKENHI Grant No. JP21H05191).
K.F. was also supported by JSPS KAKENHI Grant-in-Aid for Research Activity Start-up (Grant No. JP25K23354), JSPS KAKENHI Grant-in-Aid for Early-Career Scientists (Grant No. JP26K17045), JSPS KAKENHI Grant-in-Aid for JSPS Fellows (Grant No. JP26KJ0404), and the JSPS Bilateral Program with MTA (Project No. 120263802).
This research was supported by the Hungarian National Research, Development and Innovation Office: I.V. and B.P. were supported by the NKFIH Grant No. K-145904, also B.P. was supported by the NKFIH excellence grant TKP2021-NKTA-64.

\appendix

\section{Proof of mutual commutativity of charges}\label{app:mutual-commutativity}
In this appendix, we prove the mutual commutativity of the charges in Eq.~\eqref{eq:charges}.
In the following, we often omit the subscript of the number of the generators $M$.
Following Eq.~\eqref{eq:shorthand-notation-for-prod}, for any independent set ${S} \subset {V}_{M}$, we define ${h}_{{S}} \equiv \prod_{\boldsymbol{j} \in {S}} {h}_{\boldsymbol{j}}$.

For any independent sets ${S}, {S}^\prime \subset {V}_M$, we have
\begin{align}
	\qty[{h}_{{S}}, {h}_{{S}^\prime}]
	=
	\begin{cases}
		2h_{{S}}h_{{S}^\prime}
		  &
		\text{if }\abs{{E}_{{S} \oplus {S}^\prime}} \text{ is odd}
		\\[15pt]
		0 & \text{if }\abs{{E}_{{S} \oplus {S}^\prime}} \text{ is even}
	\end{cases}
	\,,
\end{align}
where ${S} \oplus {S}^\prime \equiv \qty({S} \cup {S}^\prime) \setminus \qty({S} \cap {S}^\prime)$ denotes the symmetric difference.
Note that the product can be rewritten as
$
	h_{{S}}h_{{S}^\prime}
	=
	\qty(\prod_{\boldsymbol{j} \in {S} \cap {S}^\prime} b_{\boldsymbol{j}}^2) h_{{S} \setminus {S}^\prime}
	h_{{S}^\prime \setminus {S}}
$
where the operators on common vertices become scalar factors.

The commutator of the charges then reads
\begin{align}
	\qty[\Q{M}{k}, \Q{M}{l}]
	 & =
	\qty[
		\sum_{L\in \indepset{M}{k}} h_{L}
		,
		\sum_{R \in \indepset{M}{l}} h_{R}
	]
	=
	2
	\sum_{\substack{
			(L, R)  \in \indepset{M}{k} \times \indepset{M}{l}
	\\
			\abs{{E}_{L\oplus R}} \text{ is odd}
		}}
	h_{L} h_{R}
	\,.
	\label{eq:commutator-cancellation}
\end{align}
We show that all terms in the RHS cancel.
For this, we consider the following subsets of ${E}_{L\oplus R}$: (i) \emph{balanced-odd-edges components} (BOE components) and (ii) \emph{claw-cancellation parts} (CC parts).
We prove that the cancellation is ensured if ${E}_{L\oplus R}$ includes either at least one BOE component or at least one CC part.
Since $\abs{{E}_{L\oplus R}}$ is odd, at least one of these structures must exist.

Below, we explain the two cases of cancellation.
All terms in the RHS of~\eqref{eq:commutator-cancellation} cancel via either BOE components or CC parts.
In Appendix~\ref{app:BOE-cancellation}, we define BOE component and we show that when $\LopR$ has the BOE component, there is another term that cancels $h_L h_R$.
In Appendix~\ref{app:CC-cancellation}, we define CC-part and show that when $\LopR$ has the CC-part, there are three other terms that cancel with $h_L h_R$.
In Appendix~\ref{app:exhausiveness}, we show all terms are canceled at least by either BOE components or CC parts.

\subsection{Cancellation via balanced-odd-edges components}
\label{app:BOE-cancellation}
We define a subset $\mathcal{O} \subset L\oplus R$ to be a BOE component if it satisfies the following conditions: (i) $\abs{{E}_{\mathcal{O}}}$ is odd, (ii) $\abs{\mathcal{O} \cap L} = \abs{\mathcal{O} \cap R}$, and (iii) ${G}[\mathcal{O}]$ is an isolated connected subgraph in ${G}[\LopR]$.

When $L\oplus R$ contains a BOE component $\mathcal{O}$, there exists another pair $(L^\prime, R^\prime) \in \indepset{M}{k} \times \indepset{M}{l}$ defined by
\begin{align}
	L^\prime \equiv (L \setminus \mathcal{O}_{L}) \cup \mathcal{O}_{R},
	\quad
	R^\prime \equiv (R \setminus \mathcal{O}_{R}) \cup \mathcal{O}_{L},
	\label{eq:BOE-component-cancellation}
\end{align}
where $\mathcal{O}_{L} \equiv \mathcal{O} \cap L$ and $\mathcal{O}_{R} \equiv \mathcal{O} \cap R$.
Note that the following anticommutation relation holds:
\begin{align}
	h_{\mathcal{O}_{L}}h_{\mathcal{O}_{R}} + h_{\mathcal{O}_{R}}h_{\mathcal{O}_{L}} = 0,
\end{align}
because there are an odd number of edges between $\mathcal{O}_{L}$ and $\mathcal{O}_{R}$, and both $\mathcal{O}_{L}$ and $\mathcal{O}_{R}$ are independent sets.
Then we have the following cancellation:
\begin{align}
	h_{L} h_{R} + h_{L^\prime} h_{R^\prime}
	 & =
	h_{L \setminus \mathcal{O}_{L}} (h_{\mathcal{O}_{L}}h_{\mathcal{O}_{R}} + h_{\mathcal{O}_{R}}h_{\mathcal{O}_{L}}) h_{R \setminus \mathcal{O}_{R}}
	=
	0.
\end{align}

Thus, we have proved that the term $h_L h_R$ in the RHS of~\eqref{eq:commutator-cancellation} cancels with the term $h_{L^\prime} h_{R^\prime}$ if $L\oplus R$ contains a BOE component.

Note that the cancellation via BOE components follows the same argument as in the claw-free case~\cite{fermions-behind-the-disguise}.
Next, we consider the cancellation arising from CC parts in $L\oplus R$.

\subsection{Cancellation via claw-cancellation part}
\label{app:CC-cancellation}

Here we explain the cancellation of the term $h_{L}h_{R}$ in the RHS of~\eqref{eq:commutator-cancellation} for the case when $L\oplus R$ has a claw-cancellation part (CC-part).
We first explain the structure of the claw.
Then, we define the claw-cancellation part.
And finally, we show the cancellation arising from the claw-cancellation part.

\subsubsection{Structure of claw}
\label{sec:structure-claw}
First, we explain the claw structure.
Please refer to Section~\ref{sec:claw-notations} for the definitions of the notation for claws.
Since the leaves form an independent set, they can be simultaneously included in $L$ (or $R$).

We show the structure of claws at $h_{2j-1}$ ($h_{\overline{2j+3}}$) in Figure~\ref{fig:claw-configurations} (Figure~\ref{fig:claw-configurations-bar}).
To facilitate understanding of the neighbors of claws at center $h_{2j-1}$ ($h_{\overline{2j+3}}$), we illustrate the neighborhood structure in Figure~\ref{fig:claw-neigbors-AB} (Figure~\ref{fig:claw-neigbors-ABbar}), which is closely explained in the next section~\ref{sec:potential-neighbors-claw}.
In these figures, bold green lines indicate edges connecting the claw center to its leaves.
The claw centers are represented by blue-filled circles, while the leaves are shown as green-filled circles.
We assume the claw center belongs to $L \in \stilde{k}$ and the leaves belong to $R \in \stilde{l}$; the converse case follows by a similar argument.

\begin{figure}[tbp]
	\centering
	\begin{tikzpicture}[triangular lattice small]
		% Define common y-coordinates
		\pgfmathsetmacro{\yrowone}{0}
		\pgfmathsetmacro{\yrowtwo}{-1.732}  % -sqrt(3)
		\pgfmathsetmacro{\yrowthree}{-3.464} % -2*sqrt(3)
		\pgfmathsetmacro{\offset}{-2}

		\def\captionx{3}
		\def\captiony{-4.5}
		\def\shifty{-6.5}

		% First subfigure (a)
		\begin{scope}[shift={(0,0)}]
			% rearrangeable clique - left
			\placevertex{-3}{\offset}{\yrowone}{rearrangeablevertex}
			\placevertex{-2}{\offset}{\yrowtwo}{rearrangeablevertex}
			\placevertexbar{-1}{\offset}{\yrowthree}{rearrangeablevertex}
			% right
			\placevertex{7}{\offset}{\yrowone}{rearrangeablevertex}
			\placevertex{8}{\offset}{\yrowtwo}{rearrangeablevertex}
			\placevertexbar{9}{\offset}{\yrowthree}{rearrangeablevertex}

			% frozen evens
			\placevertex{4}{\offset}{\yrowtwo}{frozeneven}
			\placevertex{6}{\offset}{\yrowtwo}{frozeneven}

			% special odd
			\placevertex{5}{\offset}{\yrowone}{specialodd}

			% Place vertices in row 1 (top)
			\placevertex{-1}{\offset}{\yrowone}{vertex}
			\placevertex{1}{\offset}{\yrowone}{clawcenter}
			\placevertex{3}{\offset}{\yrowone}{vertex}

			% Place vertices in row 2 (middle)
			\placevertex{0}{\offset}{\yrowtwo}{clawleaves}
			\placevertex{2}{\offset}{\yrowtwo}{vertex}

			% Place vertices in row 3 (bottom)
			\placevertexbar{1}{\offset}{\yrowthree}{vertex}
			\placevertexbar{3}{\offset}{\yrowthree}{clawleaves}
			\placevertexbar{5}{\offset}{\yrowthree}{vertex}
			\placevertexbar{7}{\offset}{\yrowthree}{clawleaves}

			% Draw edges
			% Row 1 to Row 2 edges
			\draw[edge] (v-3) -- (v-2);
			\draw[edge] (v-1) -- (v0);
			\draw[edge] (v1) -- (v2);
			\draw[edge] (v3) -- (v4);
			\draw[edge] (v5) -- (v6);
			\draw[edge] (v7) -- (v8);

			\draw[edge] (v-1) -- (v-2);
			\draw[edge] (v3) -- (v2);
			\draw[edge] (v5) -- (v4);
			\draw[edge] (v7) -- (v6);

			% Row 2 to Row 3 edges
			\draw[edge] (v-2) -- (bar-1);
			\draw[edge] (v0) -- (bar1);
			\draw[edge] (v2) -- (bar3);
			\draw[edge] (v4) -- (bar5);
			\draw[edge] (v6) -- (bar7);
			\draw[edge] (v8) -- (bar9);

			\draw[edge] (v0) -- (bar-1);
			\draw[edge] (v2) -- (bar1);
			\draw[edge] (v4) -- (bar3);
			\draw[edge] (v6) -- (bar5);
			\draw[edge] (v8) -- (bar7);

			% Row 1 to Row 3 edges (curved)
			\draw[edgefar, curved] (v-3) to (bar-1);
			\draw[edgefar, curved] (v-1) to (bar1);
			\draw[edgefar, curved] (v3) to (bar5);
			\draw[edgefar, curved] (v5) to (bar7);
			\draw[edgefar, curved] (v7) to (bar9);

			% Diagonal connections (one step right)
			\draw[edgefar, curved_right] (v-1) to (bar-1);
			\draw[edgefar, curved_right] (v1) to (bar1);
			\draw[edgefar, curved_right] (v3) to (bar3);
			\draw[edgefar, curved_right] (v5) to (bar5);
			\draw[edgefar, curved_right] (v7) to (bar7);

			% Diagonal connections (one step left)
			\draw[edgefar, curved] (v-3) to (bar1);
			\draw[edgefar, curved] (v-1) to (bar3);
			\draw[edgefar, curved] (v1) to (bar5);
			\draw[edgefar, curved] (v3) to (bar7);
			\draw[edgefar, curved] (v5) to (bar9);

			% Far diagonal connections (two steps left)
			\draw[edgefar, curved] (v-3) to (bar3);
			\draw[edgefar, curved] (v-1) to (bar5);
			\draw[edgefar, curved] (v3) to (bar9);

			% Horizontal edges
			\draw[edge] (v-3) -- (v-1);
			\draw[edge] (v-1) -- (v1);
			\draw[edge] (v1) -- (v3);
			\draw[edge] (v3) -- (v5);
			\draw[edge] (v5) -- (v7);

			\draw[edge] (v-2) -- (v0);
			\draw[edge] (v0) -- (v2);
			\draw[edge] (v2) -- (v4);
			\draw[edge] (v4) -- (v6);
			\draw[edge] (v6) -- (v8);

			\draw[edge] (bar-1) -- (bar1);
			\draw[edge] (bar1) -- (bar3);
			\draw[edge] (bar3) -- (bar5);
			\draw[edge] (bar5) -- (bar7);
			\draw[edge] (bar7) -- (bar9);

			\draw[clawedge, curved] (v1) to (bar3);
			\draw[clawedge, curved] (v1) to (bar7);
			\draw[clawedge] (v1) -- (v0);

			% Caption
			\node at (\captionx,\captiony) {type A claw at $h_{2j-1}$};
		\end{scope}

		% Second subfigure (b)
		\begin{scope}[shift={(0,{\shifty})}]
			% Similar structure with different claw configuration
			% rearrangeable clique
			\placevertex{-3}{\offset}{\yrowone}{rearrangeablevertex}
			\placevertex{-2}{\offset}{\yrowtwo}{rearrangeablevertex}
			\placevertexbar{-1}{\offset}{\yrowthree}{rearrangeablevertex}
			\placevertex{7}{\offset}{\yrowone}{rearrangeablevertex}
			\placevertex{8}{\offset}{\yrowtwo}{rearrangeablevertex}
			\placevertexbar{9}{\offset}{\yrowthree}{rearrangeablevertex}

			% frozen evens
			\placevertex{4}{\offset}{\yrowtwo}{frozeneven}
			\placevertex{6}{\offset}{\yrowtwo}{frozeneven}

			% special odd
			\placevertex{5}{\offset}{\yrowone}{specialodd}

			% Main vertices with different claw configuration
			\placevertex{-1}{\offset}{\yrowone}{clawleaves}
			\placevertex{1}{\offset}{\yrowone}{clawcenter}
			\placevertex{3}{\offset}{\yrowone}{vertex}

			\placevertex{0}{\offset}{\yrowtwo}{vertex}
			\placevertex{2}{\offset}{\yrowtwo}{clawleaves}

			\placevertexbar{1}{\offset}{\yrowthree}{vertex}
			\placevertexbar{3}{\offset}{\yrowthree}{vertex}
			\placevertexbar{5}{\offset}{\yrowthree}{vertex}
			\placevertexbar{7}{\offset}{\yrowthree}{clawleaves}

			% Edge structure (abbreviated for clarity - same pattern as first figure)
			% Row 1 to Row 2
			\draw[edge] (v-3) -- (v-2);
			\draw[edge] (v-1) -- (v0);
			\draw[edge] (v3) -- (v4);
			\draw[edge] (v5) -- (v6);
			\draw[edge] (v7) -- (v8);

			\draw[edge] (v-1) -- (v-2);
			\draw[edge] (v1) -- (v0);
			\draw[edge] (v3) -- (v2);
			\draw[edge] (v5) -- (v4);
			\draw[edge] (v7) -- (v6);

			% Row 2 to Row 3
			\draw[edge] (v-2) -- (bar-1);
			\draw[edge] (v0) -- (bar1);
			\draw[edge] (v2) -- (bar3);
			\draw[edge] (v4) -- (bar5);
			\draw[edge] (v6) -- (bar7);
			\draw[edge] (v8) -- (bar9);

			\draw[edge] (v0) -- (bar-1);
			\draw[edge] (v2) -- (bar1);
			\draw[edge] (v4) -- (bar3);
			\draw[edge] (v6) -- (bar5);
			\draw[edge] (v8) -- (bar7);

			% Row 1 to Row 3 (curved)
			\draw[edgefar] (v-3) to (bar-1);
			\draw[edgefar] (v-1) to (bar1);
			\draw[edgefar] (v1) to (bar3);
			\draw[edgefar] (v3) to (bar5);
			\draw[edgefar] (v5) to (bar7);
			\draw[edgefar] (v7) to (bar9);

			% Diagonal connections
			\draw[edgefar_right] (v-1) to (bar-1);
			\draw[edgefar_right] (v1) to (bar1);
			\draw[edgefar_right] (v3) to (bar3);
			\draw[edgefar_right] (v5) to (bar5);
			\draw[edgefar_right] (v7) to (bar7);

			\draw[edgefar] (v-3) to (bar1);
			\draw[edgefar] (v-1) to (bar3);
			\draw[edgefar] (v1) to (bar5);
			\draw[edgefar] (v3) to (bar7);
			\draw[edgefar] (v5) to (bar9);

			\draw[edgefar] (v-3) to (bar3);
			\draw[edgefar] (v-1) to (bar5);
			\draw[edgefar] (v3) to (bar9);

			% Horizontal edges
			\draw[edge] (v-3) -- (v-1);
			\draw[edge] (v1) -- (v3);
			\draw[edge] (v3) -- (v5);
			\draw[edge] (v5) -- (v7);

			\draw[edge] (v-2) -- (v0);
			\draw[edge] (v0) -- (v2);
			\draw[edge] (v2) -- (v4);
			\draw[edge] (v4) -- (v6);
			\draw[edge] (v6) -- (v8);

			\draw[edge] (bar-1) -- (bar1);
			\draw[edge] (bar1) -- (bar3);
			\draw[edge] (bar3) -- (bar5);
			\draw[edge] (bar5) -- (bar7);
			\draw[edge] (bar7) -- (bar9);

			\draw[clawedge] (v1) -- (v2);
			\draw[clawedge] (v-1) -- (v1);
			\draw[clawedge, curved] (v1) to (bar7);

			% Caption
			\node at (\captionx,\captiony) {type B claw at $h_{2j-1}$};
		\end{scope}

	\end{tikzpicture}
	\caption{
		Structures of claws at $h_{2j-1}$.
		The bold green edges connect the claw center to the claw leaves, while gray edges indicate all other edges on the frustration graph.
		Each vertex labeled with number $i$ represents $h_{2j+i}$ and labeled with $\overline{i}$ represents $\hbar{2j+i}$.
		The blue-filled circle indicates the claw center $h_{2j-1}$.
		Green-filled circles indicate the claw leaves: $h_{2j-2}$, $h_{\overline{2j+1}}$, $h_{\overline{2j+5}}$ for type A claws, and $h_{2j-3}$, $h_{2j}$, $h_{\overline{2j+5}}$ for type B claws.
		Rectangles represent the frozen even vertices $h_{2j+2}$ and $h_{2j+4}$.
		The blue-bordered circle indicates the special odd vertex $h_{2j+3}$.
		Orange-bordered circles indicate the rearrangeable clique: $\rearrangeable{2j-1}{-} = \{h_{2j-5}$, $h_{2j-4}$, $h_{\overline{2j-3}}\}$ on the left, and $\rearrangeable{2j-1}{+} = \{h_{2j+5}$, $h_{2j+6}$, $h_{\overline{2j+7}}\}$ on the right.
		The white circles indicate forbidden vertices, that cannot be connected to any of the claw leaves in $\LopR$.
	}
	\label{fig:claw-configurations}
\end{figure}

\begin{figure}[tbp]
	\centering
	\begin{tikzpicture}[triangular lattice small]
		% Define common y-coordinates
		\pgfmathsetmacro{\yrowone}{0}
		\pgfmathsetmacro{\yrowtwo}{-1.732}  % -sqrt(3)
		\pgfmathsetmacro{\yrowthree}{-3.464} % -2*sqrt(3)
		\pgfmathsetmacro{\offset}{-2}

		\def\captionx{3}
		\def\captiony{-4.5}
		\def\shifty{-6.5}

		% Third subfigure (c)
		\begin{scope}[shift={(0,{2*\shifty})}]
			% rearrangeable clique
			\placevertex{-3}{\offset}{\yrowone}{rearrangeablevertex}
			\placevertex{-2}{\offset}{\yrowtwo}{rearrangeablevertex}
			\placevertexbar{-1}{\offset}{\yrowthree}{rearrangeablevertex}
			\placevertex{7}{\offset}{\yrowone}{rearrangeablevertex}
			\placevertex{8}{\offset}{\yrowtwo}{rearrangeablevertex}
			\placevertexbar{9}{\offset}{\yrowthree}{rearrangeablevertex}

			% frozen evens
			\placevertex{0}{\offset}{\yrowtwo}{frozeneven}
			\placevertex{2}{\offset}{\yrowtwo}{frozeneven}

			% special odd
			\placevertexbar{1}{\offset}{\yrowthree}{specialodd}

			% Place vertices
			\placevertex{-1}{\offset}{\yrowone}{clawleaves}
			\placevertex{1}{\offset}{\yrowone}{vertex}
			\placevertex{3}{\offset}{\yrowone}{clawleaves}
			\placevertex{5}{\offset}{\yrowone}{vertex}

			\placevertex{4}{\offset}{\yrowtwo}{vertex}
			\placevertex{6}{\offset}{\yrowtwo}{clawleaves}

			\placevertexbar{3}{\offset}{\yrowthree}{vertex}
			\placevertexbar{5}{\offset}{\yrowthree}{clawcenter}
			\placevertexbar{7}{\offset}{\yrowthree}{vertex}

			% Edges (abbreviated)
			% Row 1 to Row 2
			\draw[edge] (v-3) -- (v-2);
			\draw[edge] (v-1) -- (v0);
			\draw[edge] (v1) -- (v2);
			\draw[edge] (v3) -- (v4);
			\draw[edge] (v5) -- (v6);
			\draw[edge] (v7) -- (v8);

			\draw[edge] (v-1) -- (v-2);
			\draw[edge] (v1) -- (v0);
			\draw[edge] (v3) -- (v2);
			\draw[edge] (v5) -- (v4);
			\draw[edge] (v7) -- (v6);

			% Row 2 to Row 3
			\draw[edge] (v-2) -- (bar-1);
			\draw[edge] (v0) -- (bar1);
			\draw[edge] (v2) -- (bar3);
			\draw[edge] (v4) -- (bar5);
			\draw[edge] (v6) -- (bar7);
			\draw[edge] (v8) -- (bar9);

			\draw[edge] (v0) -- (bar-1);
			\draw[edge] (v2) -- (bar1);
			\draw[edge] (v4) -- (bar3);
			\draw[edge] (v8) -- (bar7);

			% Row 1 to Row 3 (curved)
			\draw[edgefar, curved] (v-3) to (bar-1);
			\draw[edgefar, curved] (v-1) to (bar1);
			\draw[edgefar, curved] (v1) to (bar3);
			\draw[edgefar, curved] (v5) to (bar7);
			\draw[edgefar, curved] (v7) to (bar9);

			% Diagonal connections
			\draw[edgefar, curved_right] (v-1) to (bar-1);
			\draw[edgefar, curved_right] (v1) to (bar1);
			\draw[edgefar, curved_right] (v3) to (bar3);
			\draw[edgefar, curved_right] (v5) to (bar5);
			\draw[edgefar, curved_right] (v7) to (bar7);

			\draw[edgefar, curved] (v-3) to (bar1);
			\draw[edgefar, curved] (v-1) to (bar3);
			\draw[edgefar, curved] (v1) to (bar5);
			\draw[edgefar, curved] (v3) to (bar7);
			\draw[edgefar, curved] (v5) to (bar9);

			\draw[edgefar, curved] (v-3) to (bar3);
			\draw[edgefar, curved] (v1) to (bar7);
			\draw[edgefar, curved] (v3) to (bar9);

			% Horizontal edges
			\draw[edge] (v-3) -- (v-1);
			\draw[edge] (v-1) -- (v1);
			\draw[edge] (v1) -- (v3);
			\draw[edge] (v3) -- (v5);
			\draw[edge] (v5) -- (v7);

			\draw[edge] (v-2) -- (v0);
			\draw[edge] (v0) -- (v2);
			\draw[edge] (v2) -- (v4);
			\draw[edge] (v4) -- (v6);
			\draw[edge] (v6) -- (v8);

			\draw[edge] (bar-1) -- (bar1);
			\draw[edge] (bar1) -- (bar3);
			\draw[edge] (bar3) -- (bar5);
			\draw[edge] (bar5) -- (bar7);
			\draw[edge] (bar7) -- (bar9);

			\draw[clawedge, curved] (v-1) to (bar5);
			\draw[clawedge, curved] (v3) to (bar5);
			\draw[clawedge] (v6) -- (bar5);

			% Caption
			\node at (\captionx,\captiony) {{type A claw at $h_{\overline{2j+3}}$}};
		\end{scope}

		% Fourth subfigure (d)
		\begin{scope}[shift={(0,{3*\shifty})}]
			% rearrangeable clique
			\placevertex{-3}{\offset}{\yrowone}{rearrangeablevertex}
			\placevertex{-2}{\offset}{\yrowtwo}{rearrangeablevertex}
			\placevertexbar{-1}{\offset}{\yrowthree}{rearrangeablevertex}
			\placevertex{7}{\offset}{\yrowone}{rearrangeablevertex}
			\placevertex{8}{\offset}{\yrowtwo}{rearrangeablevertex}
			\placevertexbar{9}{\offset}{\yrowthree}{rearrangeablevertex}

			% frozen evens
			\placevertex{0}{\offset}{\yrowtwo}{frozeneven}
			\placevertex{2}{\offset}{\yrowtwo}{frozeneven}

			% special odd
			\placevertexbar{1}{\offset}{\yrowthree}{specialodd}

			% Place vertices
			\placevertex{-1}{\offset}{\yrowone}{clawleaves}
			\placevertex{1}{\offset}{\yrowone}{vertex}
			\placevertex{3}{\offset}{\yrowone}{vertex}
			\placevertex{5}{\offset}{\yrowone}{vertex}

			\placevertex{4}{\offset}{\yrowtwo}{clawleaves}
			\placevertex{6}{\offset}{\yrowtwo}{vertex}

			\placevertexbar{3}{\offset}{\yrowthree}{vertex}
			\placevertexbar{5}{\offset}{\yrowthree}{clawcenter}
			\placevertexbar{7}{\offset}{\yrowthree}{clawleaves}

			% Edges (abbreviated)
			% Row 1 to Row 2
			\draw[edge] (v-3) -- (v-2);
			\draw[edge] (v-1) -- (v0);
			\draw[edge] (v1) -- (v2);
			\draw[edge] (v3) -- (v4);
			\draw[edge] (v5) -- (v6);
			\draw[edge] (v7) -- (v8);

			\draw[edge] (v-1) -- (v-2);
			\draw[edge] (v1) -- (v0);
			\draw[edge] (v3) -- (v2);
			\draw[edge] (v5) -- (v4);
			\draw[edge] (v7) -- (v6);

			% Row 2 to Row 3
			\draw[edge] (v-2) -- (bar-1);
			\draw[edge] (v0) -- (bar1);
			\draw[edge] (v2) -- (bar3);
			\draw[edge] (v6) -- (bar7);
			\draw[edge] (v8) -- (bar9);

			\draw[edge] (v0) -- (bar-1);
			\draw[edge] (v2) -- (bar1);
			\draw[edge] (v4) -- (bar3);
			\draw[edge] (v6) -- (bar5);
			\draw[edge] (v8) -- (bar7);

			% Row 1 to Row 3 (curved)
			\draw[edgefar, curved] (v-3) to (bar-1);
			\draw[edgefar, curved] (v-1) to (bar1);
			\draw[edgefar, curved] (v1) to (bar3);
			\draw[edgefar, curved] (v3) to (bar5);
			\draw[edgefar, curved] (v5) to (bar7);
			\draw[edgefar, curved] (v7) to (bar9);

			% Diagonal connections
			\draw[edgefar, curved_right] (v-1) to (bar-1);
			\draw[edgefar, curved_right] (v1) to (bar1);
			\draw[edgefar, curved_right] (v3) to (bar3);
			\draw[edgefar, curved_right] (v5) to (bar5);
			\draw[edgefar, curved_right] (v7) to (bar7);

			\draw[edgefar, curved] (v-3) to (bar1);
			\draw[edgefar, curved] (v-1) to (bar3);
			\draw[edgefar, curved] (v1) to (bar5);
			\draw[edgefar, curved] (v3) to (bar7);
			\draw[edgefar, curved] (v5) to (bar9);

			\draw[edgefar, curved] (v-3) to (bar3);
			\draw[edgefar, curved] (v1) to (bar7);
			\draw[edgefar, curved] (v3) to (bar9);

			% Horizontal edges
			\draw[edge] (v-3) -- (v-1);
			\draw[edge] (v-1) -- (v1);
			\draw[edge] (v1) -- (v3);
			\draw[edge] (v3) -- (v5);
			\draw[edge] (v5) -- (v7);

			\draw[edge] (v-2) -- (v0);
			\draw[edge] (v0) -- (v2);
			\draw[edge] (v2) -- (v4);
			\draw[edge] (v4) -- (v6);
			\draw[edge] (v6) -- (v8);

			\draw[edge] (bar-1) -- (bar1);
			\draw[edge] (bar1) -- (bar3);
			\draw[edge] (bar3) -- (bar5);
			\draw[edge] (bar7) -- (bar9);

			\draw[clawedge] (v4) -- (bar5);
			\draw[clawedge] (bar5) -- (bar7);
			\draw[clawedge, curved] (v-1) to (bar5);

			% Caption
			\node at (\captionx,\captiony) {type B claw at $h_{\overline{2j+3}}$};
		\end{scope}

	\end{tikzpicture}
	\caption{
		Structures of claws at $\hbar{2j+3}$.
		The bold green edges connect the claw center to the claw leaves, while gray edges indicate all other edges.
		Each vertex labeled with number $i$ represents $h_{2j+i}$ and labeled with $\overline{i}$ represents $\hbar{2j+i}$.
		The blue-filled circle indicates the claw center $\hbar{2j+3}$.
		Green-filled circles indicate the claw leaves: $h_{2j-3}$, $h_{2j+1}$, $h_{2j+4}$ for type A claws, and $h_{2j-3}$, $h_{2j+2}$, $\hbar{2j+5}$ for type B claws.
		Solid rectangles represent the frozen even vertices $h_{2j-2}$ and $h_{2j}$.
		The blue-bordered circle indicates the special odd vertex $\hbar{2j-1}$.
		Orange-bordered circles indicate the rearrangeable clique: $\rearrangeable{\overline{2j+3}}{+} = \{h_{2j-5}$, $h_{2j-4}$, $h_{\overline{2j-3}}\}$ on the left, and $\rearrangeable{\overline{2j+3}}{-} = \{h_{2j+5}$, $h_{2j+6}$, $h_{\overline{2j+7}}\}$ on the right, which is the same as those in Figure~\ref{fig:claw-configurations}.
		The white circles indicate forbidden vertices, that cannot be connected to any of the claw leaves in $\LopR$.
	}
	\label{fig:claw-configurations-bar}
\end{figure}

\subsubsection{Potential neighbors of claw}
\label{sec:potential-neighbors-claw}
Here we explain the \emph{potential neighbors} of claw, which is the neighbors of claw in ${G}[\LopR]$.
Note that ${G}[\LopR]$ is a bipartite graph, and $L$ and $R$ are independent sets; vertices in $L$ (or $R$) are not connected to each other.

The potential neighbors of claws are classified into three categories: (i) \emph{rearrangeable clique}, (ii) \emph{special odd vertices}, and (iii) \emph{frozen even vertices}, which will be explained in the following.

Let the rearrangeable clique of a claw at $c$ be denoted by $\rearrangeable{c}{\sigma} = (\vrear{\sigma, 1}_{c}, \vrear{\sigma, 2}_{c}, \vrear{\sigma, 3}_{c})$, where $\sigma \in \{+, -\}$.
The three vertices in each rearrangeable clique are mutually connected, forming a clique.
For claws at $h_{2j-1}$ and $h_{\overline{2j+3}}$, the rearrangeable cliques are defined as
\begin{align}
	\rearrangeable{2j-1}{-} = \rearrangeable{\overline{2j+3}}{+} \equiv (h_{2j-5}, h_{2j-4}, h_{\overline{2j-3}})
	\,,\quad
	\rearrangeable{2j-1}{+} = \rearrangeable{\overline{2j+3}}{-} \equiv (h_{2j+5}, h_{2j+6}, h_{\overline{2j+7}})
	\,.
\end{align}
The rearrangeable clique $\rearrangeable{c}{-}$ is connected to the special odd and frozen vertices of the claw at $c$ in ${G}$, while $\rearrangeable{c}{+}$ is not.
These rearrangeable cliques are represented by orange-bordered circles in Figures~\ref{fig:claw-configurations}--\ref{fig:claw-neigbors-ABbar}.
Note that claws at $h_{2j-1}$ and $h_{\overline{2j+3}}$ share the same set of rearrangeable cliques, independent of the claw type.
Since each rearrangeable clique forms a clique, a claw can connect to at most one vertex in each rearrangeable clique.
Importantly, the possible neighbors of rearrangeable cliques connected to claws at $h_{2j-1}$ and $h_{\overline{2j+3}}$ are also the same, independent of the claw type.

The special odd vertex for the claw at $h_{2j-1}$ is defined as $\vspo_{2j-1} \equiv h_{2j+3}$, and those for the claw at $h_{\overline{2j+3}}$ is defined as $\vspo_{\overline{2j+3}} \equiv h_{\overline{2j-1}}$.
These are represented by blue-bordered circles in Figures~\ref{fig:claw-configurations}--\ref{fig:claw-neigbors-ABbar}.
When the leaves are connected to the special odd vertex, the claw can potentially cancel with an extended claw structure, as is explained in Figure~\ref{fig:CC-part-one-spodd-plusone-neigbor}.

The frozen even vertices for the claw at $h_{2j-1}$ are defined as $\vfrz{1}_{2j-1} \equiv h_{2j+2}$ and $\vfrz{2}_{2j-1} \equiv h_{2j+4}$, and those for the claw with center $\hbar{2j+3}$ are defined as $\vfrz{1}_{\overline{{2j+3}}} \equiv h_{2j}$ and $\vfrz{2}_{\overline{{2j+3}}} \equiv h_{2j-2}$.
$\vfrz{2}_{c}$ is connected to $\rearrangeable{c}{+}$ and $\vfrz{1}_{c}$ is not.
These are represented by rectangles in Figures~\ref{fig:claw-configurations}--\ref{fig:claw-neigbors-ABbar}.
When the leaves are connected to the frozen even vertices, these frozen vertices cannot be connected to any other vertices in $R$.
This property prevents claw cancellation involving this claw, and the cancellation must then be achieved through other CC-parts or reduced to BOE-cancellation, as we will see later.
Note that claws with the same center share the same set of frozen even vertices independent of the claw type.

The other neighbors are represented by white circles and are called forbidden vertices in Figures~\ref{fig:claw-configurations}, \ref{fig:claw-configurations-bar}.
These are neighbors of the claw center that cannot be connected to the leaves, since $R$ is an independent set of vertices.
Also note that the claw center cannot be connected to any vertices other than the leaves; otherwise $L$ would become a non-independent set, contradicting its definition.
Thus, for example, when $\LopR$ contains a type A claw at  $h_{2j-1}$, neither $L$ nor $R$ can include the vertices $h_{2j-2}$ and $h_{2j+2}$, which are leaves of the type B claw with the same center.
This leads to the rearrangeability property of claws: if $\LopR$ contains a claw, there exists another configuration $(L^\prime, R^\prime) \in \indepset{M}{k} \times \indepset{M}{l}$ where $L^\prime \oplus R^\prime$ equals $\LopR$ with the claw replaced by a different type of claw with the same center.

%Clearly, $\LopR$ cannot contain nearby claws simultaneously: if claws whose centers are $h_{2j-1}$ and $h_{2j-1+\delta}$ (where $-10 < \delta < 10$) are both included in $L$, they cannot coexist, nor can claws whose centers are $\hbar{2j-1}$ and $\hbar{2j-1+\delta}$; otherwise, at least one pair of leaves from the two claws would be adjacent, contradicting the assumption that $R$ is an independent set.

We can confirm that there are no other claws in the frustration graph ${G}_M$.
The claws included in ${G}_M$ are exhausted by those given above.
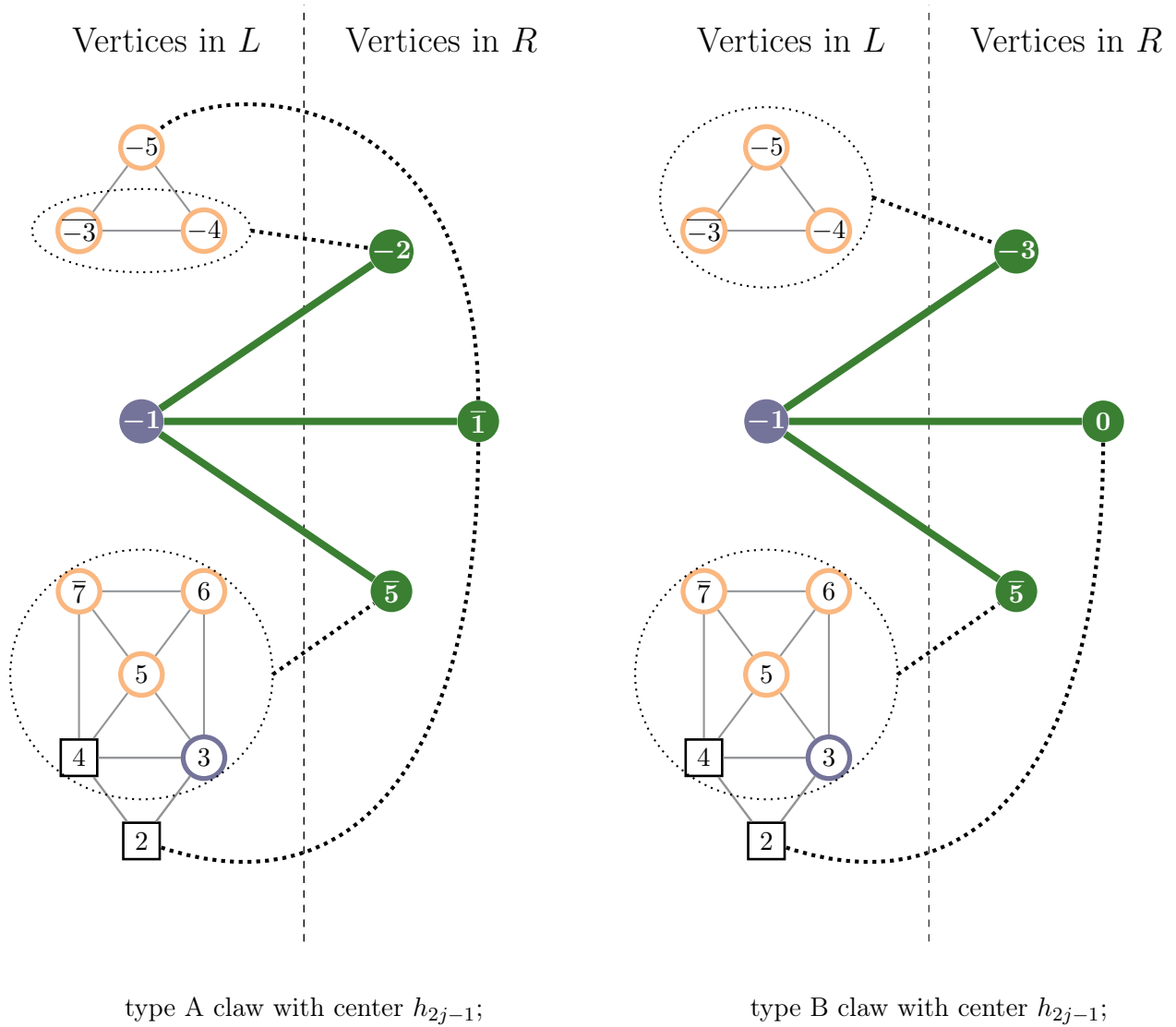
\begin{figure}[tbp]
	\centering
	\begin{tikzpicture}[
			triangular lattice small,
		]
		\pgfmathsetmacro{\offset}{0}

		\def\rightx{3.6}
		\def\righty{2.45}
		\def\delta{0.3}
		\def\deltay{1.2}
		\def\deltax{0.9}
		\def\belowoffsety{0.3}

		\begin{scope}[shift={(0,0)}]

			\node at (0.65*\rightx,-8.5) {{type A claw with center $h_{2j-1}$};};
			\draw[dashed] (0.65*\rightx,-7.5)--(0.65*\rightx,6);

			\node at (0.1*\rightx,5.5) {{\Large Vertices in $L$}};
			\node at (1.2*\rightx,5.5) {{\Large Vertices in $R$}};
			\draw[dashed] (0.65*\rightx,-7.5)--(0.65*\rightx,6);

			% claw center
			\placevertexLR{0}{0}{center}{clawcenter}{-1}{0}

			% claw leaves
			\placevertexLR{\rightx}{\righty}{leaf1}{clawleaves}{-2}{0}
			\placevertexLR{{\rightx+1.25}}{0}{leaf2}{clawleaves}{1}{1}
			\placevertexLR{\rightx}{{-\righty}}{leaf3}{clawleaves}{5}{1}

			% rearrengeble
			\placevertexLR{\deltax}{\righty+\delta}{reup1}{rearrangeablevertex}{-4}{0}
			\placevertexLR{-\deltax}{\righty+\delta}{reup2}{rearrangeablevertex}{-3}{1}
			\placevertexLR{0}{\righty+\delta+\deltay}{reup3}{rearrangeablevertex}{-5}{0}

			% rearrengeble
			\placevertexLR{\deltax}{-\righty-\delta+\belowoffsety}{redown1}{rearrangeablevertex}{6}{0}
			\placevertexLR{-\deltax}{-\righty-\delta+\belowoffsety}{redown2}{rearrangeablevertex}{7}{1}
			\placevertexLR{0}{-\righty-\delta-\deltay+\belowoffsety}{redown3}{rearrangeablevertex}{5}{0}

			% specials
			\placevertexLR{\deltax}{-\righty-\delta-\deltay*2+\belowoffsety}{spodd}{specialodd}{3}{0}
			\placevertexLR{-\deltax}{-\righty-\delta-\deltay*2+\belowoffsety}{speven1}{frozeneven}{4}{0}
			\placevertexLR{0}{-\righty-\delta-\deltay*3+\belowoffsety}{speven2}{frozeneven}{2}{0}

			\draw[clawedge] (center) -- (leaf1);
			\draw[clawedge] (center) -- (leaf2);
			\draw[clawedge] (center) -- (leaf3);

			\draw[edge] (reup1) -- (reup2) -- (reup3) -- (reup1);

			\draw[edge] (redown1) -- (redown2) -- (redown3) -- (redown1);

			\draw[edge] (spodd) -- (speven1) -- (speven2) -- (spodd);

			\draw[edge] (spodd) -- (speven1) -- (speven2) -- (spodd);

			\draw[edge] (spodd) -- (redown3) -- (speven1);

			\draw[edge] (redown2) -- (speven1);
			\draw[edge] (redown1) -- (spodd);

			% bar-3とv-4を点線の楕円で囲む
			% 楕円の中心と半径を計算
			\pgfmathsetmacro{\ellipsecenterx}{(\deltax + (-\deltax))/2}
			\pgfmathsetmacro{\ellipsecentery}{\righty+\delta}
			\pgfmathsetmacro{\ellipsewidth}{\deltax*1.75}
			\pgfmathsetmacro{\ellipseheight}{0.6}

			% 点線の楕円を描画
			\draw[groupellipse] (\ellipsecenterx,\ellipsecentery) ellipse (\ellipsewidth cm and \ellipseheight cm);

			% 楕円からv-2への接続線
			% 楕円の上端からv-2へ
			\pgfmathsetmacro{\connectionx}{\ellipsecenterx+ \ellipsewidth}
			\pgfmathsetmacro{\connectiony}{\ellipsecentery}

			% 接続線を描画（少し曲げて自然に見せる）
			\draw[potensialedge] (\connectionx,\connectiony) -- (leaf1);

			% bar-3とv-4を点線の楕円で囲む
			% 楕円の中心と半径を計算
			\pgfmathsetmacro{\ellipsecenterx}{0}
			\pgfmathsetmacro{\ellipsecentery}{-\righty-\delta-\deltay+\belowoffsety}
			\pgfmathsetmacro{\ellipsewidth}{\deltax*2.1}
			\pgfmathsetmacro{\ellipseheight}{\deltay*1.5}

			% 点線の楕円を描画
			\draw[groupellipse] (\ellipsecenterx,\ellipsecentery) ellipse (\ellipsewidth cm and \ellipseheight cm);

			% 楕円からv-2への接続線
			% 楕円の上端からv-2へ
			\pgfmathsetmacro{\connectionx}{\ellipsecenterx+ \ellipsewidth}
			\pgfmathsetmacro{\connectiony}{\ellipsecentery}

			% 接続線を描画（少し曲げて自然に見せる）
			\draw[potensialedge] (\connectionx,\connectiony) -- (leaf3);

			\draw[potensialedge] (leaf2)  .. controls +(0,5) and +(1,1) .. (reup3);
			\draw[potensialedge] (leaf2) .. controls +(0,-5) and +(3,-1) .. (speven2);

		\end{scope}

		\begin{scope}[shift={(9,0)}]

			\node at (0.65*\rightx,-8.5) {{type B claw with center $h_{2j-1}$};};

			% claw center
			\placevertexLR{0}{0}{center}{clawcenter}{-1}{0}

			\node at (0.1*\rightx,5.5) {{\Large Vertices in $L$}};
			\node at (1.2*\rightx,5.5) {{\Large Vertices in $R$}};

			% claw leaves
			\placevertexLR{\rightx}{\righty}{leaf1}{clawleaves}{-3}{0}
			\placevertexLR{{\rightx+1.25}}{0}{leaf2}{clawleaves}{0}{0}
			\placevertexLR{\rightx}{{-\righty}}{leaf3}{clawleaves}{5}{1}

			% rearrengeble
			\placevertexLR{\deltax}{\righty+\delta}{reup1}{rearrangeablevertex}{-4}{0}
			\placevertexLR{-\deltax}{\righty+\delta}{reup2}{rearrangeablevertex}{-3}{1}
			\placevertexLR{0}{\righty+\delta+\deltay}{reup3}{rearrangeablevertex}{-5}{0}

			% rearrengeble
			\placevertexLR{\deltax}{-\righty-\delta+\belowoffsety}{redown1}{rearrangeablevertex}{6}{0}
			\placevertexLR{-\deltax}{-\righty-\delta+\belowoffsety}{redown2}{rearrangeablevertex}{7}{1}
			\placevertexLR{0}{-\righty-\delta-\deltay+\belowoffsety}{redown3}{rearrangeablevertex}{5}{0}

			% specials
			\placevertexLR{\deltax}{-\righty-\delta-\deltay*2+\belowoffsety}{spodd}{specialodd}{3}{0}
			\placevertexLR{-\deltax}{-\righty-\delta-\deltay*2+\belowoffsety}{speven1}{frozeneven}{4}{0}
			\placevertexLR{0}{-\righty-\delta-\deltay*3+\belowoffsety}{speven2}{frozeneven}{2}{0}

			\draw[dashed] (0.65*\rightx,-7.5)--(0.65*\rightx,6);

			\draw[clawedge] (center) -- (leaf1);
			\draw[clawedge] (center) -- (leaf2);
			\draw[clawedge] (center) -- (leaf3);

			\draw[edge] (reup1) -- (reup2) -- (reup3) -- (reup1);

			\draw[edge] (redown1) -- (redown2) -- (redown3) -- (redown1);

			\draw[edge] (spodd) -- (speven1) -- (speven2) -- (spodd);

			\draw[edge] (spodd) -- (speven1) -- (speven2) -- (spodd);

			\draw[edge] (spodd) -- (redown3) -- (speven1);

			\draw[edge] (redown2) -- (speven1);
			\draw[edge] (redown1) -- (spodd);

			% bar-3とv-4を点線の楕円で囲む
			% 楕円の中心と半径を計算
			\pgfmathsetmacro{\ellipsecenterx}{0}
			\pgfmathsetmacro{\ellipsecentery}{\righty+\delta+0.4*\deltay}
			\pgfmathsetmacro{\ellipsewidth}{\deltax*1.7}
			\pgfmathsetmacro{\ellipseheight}{1.3}

			% 点線の楕円を描画
			\draw[groupellipse] (\ellipsecenterx,\ellipsecentery) ellipse (\ellipsewidth cm and \ellipseheight cm);

			% 楕円からv-2への接続線
			% 楕円の上端からv-2へ
			\pgfmathsetmacro{\connectionx}{\ellipsecenterx+ \ellipsewidth}
			\pgfmathsetmacro{\connectiony}{\ellipsecentery}

			% 接続線を描画（少し曲げて自然に見せる）
			\draw[potensialedge] (\connectionx,\connectiony) -- (leaf1);

			% bar-3とv-4を点線の楕円で囲む
			% 楕円の中心と半径を計算
			\pgfmathsetmacro{\ellipsecenterx}{0}
			\pgfmathsetmacro{\ellipsecentery}{-\righty-\delta-\deltay+\belowoffsety}
			\pgfmathsetmacro{\ellipsewidth}{\deltax*2.1}
			\pgfmathsetmacro{\ellipseheight}{\deltay*1.5}

			% 点線の楕円を描画
			\draw[groupellipse] (\ellipsecenterx,\ellipsecentery) ellipse (\ellipsewidth cm and \ellipseheight cm);

			% 楕円からv-2への接続線
			% 楕円の上端からv-2へ
			\pgfmathsetmacro{\connectionx}{\ellipsecenterx+ \ellipsewidth}
			\pgfmathsetmacro{\connectiony}{\ellipsecentery}

			% 接続線を描画（少し曲げて自然に見せる）
			\draw[potensialedge] (\connectionx,\connectiony) -- (leaf3);

			%\draw[edge] (leaf2)  .. controls +(0,5) and +(1,1) .. (reup3);
			\draw[potensialedge] (leaf2) .. controls +(0,-5) and +(3,-1) .. (speven2);

		\end{scope}

	\end{tikzpicture}
	\caption{
		Potential neighbors of type A and type B claws in $\LopR$ for $(L, R) \in {S}_k \times {S}_l$.
		Here we assume the claw center belongs to the $L$ side and the leaves belong to the $R$ side; the converse case follows by a similar argument.
		The claw center cannot be connected to any vertices other than its three leaves.
		Dotted lines indicate potential edges from the leaves: a dotted line to a dotted circle means the leaf can be connected to any vertex within that dotted circle.
		The descriptions of the vertices are the same as those in Figure~\ref{fig:claw-configurations}.
	}
	\label{fig:claw-neigbors-AB}
\end{figure}
\begin{figure}[tbp]
	\centering
	\begin{tikzpicture}[
			triangular lattice small
		]
		\pgfmathsetmacro{\offset}{0}

		\def\rightx{3.6}
		\def\righty{2.45}
		\def\delta{0.3}
		\def\deltay{1.2}
		\def\deltax{0.9}
		\def\belowoffsety{0.3}

		\begin{scope}[shift={(0,0)}]

			\node at (0.65*\rightx,-6.5) {{type A claw with center $\hbar{2j+3}$};};

			\node at (0.1*\rightx,7.5) {{\Large Vertices in $L$}};
			\node at (1.2*\rightx,7.5) {{\Large Vertices in $R$}};
			\draw[dashed] (0.65*\rightx,8)--(0.65*\rightx,-5.5);

			% claw center
			\placevertexLR{0}{0}{center}{clawcenter}{3}{1}

			% claw leaves
			\placevertexLR{\rightx}{\righty}{leaf1}{clawleaves}{-3}{0}
			\placevertexLR{{\rightx+1.25}}{0}{leaf2}{clawleaves}{1}{0}
			\placevertexLR{\rightx}{{-\righty}}{leaf3}{clawleaves}{4}{0}

			% rearrengeble
			\placevertexLR{\deltax}{-\righty-\delta}{reup1}{rearrangeablevertex}{6}{0}
			\placevertexLR{-\deltax}{-\righty-\delta}{reup2}{rearrangeablevertex}{5}{0}
			\placevertexLR{0}{-\righty-\delta-\deltay}{reup3}{rearrangeablevertex}{7}{1}

			% rearrengeble
			\placevertexLR{\deltax}{\righty+\delta-\belowoffsety}{redown1}{rearrangeablevertex}{-5}{0}
			\placevertexLR{-\deltax}{\righty+\delta-\belowoffsety}{redown2}{rearrangeablevertex}{-4}{0}
			\placevertexLR{0}{\righty+\delta+\deltay-\belowoffsety}{redown3}{rearrangeablevertex}{-3}{1}

			% specials
			\placevertexLR{\deltax}{\righty+\delta+\deltay*2-\belowoffsety}{spodd}{specialodd}{-1}{1}
			\placevertexLR{-\deltax}{\righty+\delta+\deltay*2-\belowoffsety}{speven1}{frozeneven}{-2}{0}
			\placevertexLR{0}{\righty+\delta+\deltay*3-\belowoffsety}{speven2}{frozeneven}{0}{0}

			\draw[clawedge] (center) -- (leaf1);
			\draw[clawedge] (center) -- (leaf2);
			\draw[clawedge] (center) -- (leaf3);

			\draw[edge] (reup1) -- (reup2) -- (reup3) -- (reup1);

			\draw[edge] (redown1) -- (redown2) -- (redown3) -- (redown1);

			\draw[edge] (spodd) -- (speven1) -- (speven2) -- (spodd);

			\draw[edge] (spodd) -- (speven1) -- (speven2) -- (spodd);

			\draw[edge] (spodd) -- (redown3) -- (speven1);

			\draw[edge] (redown2) -- (speven1);
			\draw[edge] (redown1) -- (spodd);

			% bar-3とv-4を点線の楕円で囲む
			% 楕円の中心と半径を計算
			\pgfmathsetmacro{\ellipsecenterx}{0}
			\pgfmathsetmacro{\ellipsecentery}{-\righty-\delta}
			\pgfmathsetmacro{\ellipsewidth}{\deltax*1.75}
			\pgfmathsetmacro{\ellipseheight}{0.6}

			% 点線の楕円を描画
			\draw[groupellipse] (\ellipsecenterx,\ellipsecentery) ellipse (\ellipsewidth cm and \ellipseheight cm);

			% 楕円からv-2への接続線
			% 楕円の上端からv-2へ
			\pgfmathsetmacro{\connectionx}{\ellipsecenterx+ \ellipsewidth}
			\pgfmathsetmacro{\connectiony}{\ellipsecentery}

			% 接続線を描画（少し曲げて自然に見せる）
			\draw[potensialedge] (\connectionx,\connectiony) -- (leaf3);

			% bar-3とv-4を点線の楕円で囲む
			% 楕円の中心と半径を計算
			\pgfmathsetmacro{\ellipsecenterx}{0}
			\pgfmathsetmacro{\ellipsecentery}{\righty+\delta+\deltay-\belowoffsety}
			\pgfmathsetmacro{\ellipsewidth}{\deltax*2.1}
			\pgfmathsetmacro{\ellipseheight}{\deltay*1.5}

			% 点線の楕円を描画
			\draw[groupellipse] (\ellipsecenterx,\ellipsecentery) ellipse (\ellipsewidth cm and \ellipseheight cm);

			% 楕円からv-2への接続線
			% 楕円の上端からv-2へ
			\pgfmathsetmacro{\connectionx}{\ellipsecenterx+ \ellipsewidth}
			\pgfmathsetmacro{\connectiony}{\ellipsecentery}

			% 接続線を描画（少し曲げて自然に見せる）
			\draw[potensialedge] (\connectionx,\connectiony) -- (leaf1);

			\draw[potensialedge] (leaf2)  .. controls +(0,-5) and +(1,-1) .. (reup3);
			\draw[potensialedge] (leaf2) .. controls +(0,5) and +(3,1) .. (speven2);

		\end{scope}

		\begin{scope}[shift={(9,0)}]

			\node at (0.65*\rightx,-6.5) {{type B claw with center $\hbar{2j+3}$};};

			\node at (0.1*\rightx,7.5) {{\Large Vertices in $L$}};
			\node at (1.2*\rightx,7.5) {{\Large Vertices in $R$}};
			\draw[dashed] (0.65*\rightx,8)--(0.65*\rightx,-5.5);

			% claw center
			\placevertexLR{0}{0}{center}{clawcenter}{3}{1}

			% claw leaves
			\placevertexLR{\rightx}{\righty}{leaf1}{clawleaves}{-3}{0}
			\placevertexLR{{\rightx+1.25}}{0}{leaf2}{clawleaves}{2}{0}
			\placevertexLR{\rightx}{{-\righty}}{leaf3}{clawleaves}{5}{1}

			% rearrengeble
			\placevertexLR{\deltax}{-\righty-\delta}{reup1}{rearrangeablevertex}{6}{0}
			\placevertexLR{-\deltax}{-\righty-\delta}{reup2}{rearrangeablevertex}{5}{0}
			\placevertexLR{0}{-\righty-\delta-\deltay}{reup3}{rearrangeablevertex}{7}{1}

			% rearrengeble
			\placevertexLR{\deltax}{\righty+\delta-\belowoffsety}{redown1}{rearrangeablevertex}{-5}{0}
			\placevertexLR{-\deltax}{\righty+\delta-\belowoffsety}{redown2}{rearrangeablevertex}{-4}{0}
			\placevertexLR{0}{\righty+\delta+\deltay-\belowoffsety}{redown3}{rearrangeablevertex}{-3}{1}

			% specials
			\placevertexLR{\deltax}{\righty+\delta+\deltay*2-\belowoffsety}{spodd}{specialodd}{-1}{1}
			\placevertexLR{-\deltax}{\righty+\delta+\deltay*2-\belowoffsety}{speven1}{frozeneven}{-2}{0}
			\placevertexLR{0}{\righty+\delta+\deltay*3-\belowoffsety}{speven2}{frozeneven}{0}{0}

			\draw[clawedge] (center) -- (leaf1);
			\draw[clawedge] (center) -- (leaf2);
			\draw[clawedge] (center) -- (leaf3);

			\draw[edge] (reup1) -- (reup2) -- (reup3) -- (reup1);

			\draw[edge] (redown1) -- (redown2) -- (redown3) -- (redown1);

			\draw[edge] (spodd) -- (speven1) -- (speven2) -- (spodd);

			\draw[edge] (spodd) -- (speven1) -- (speven2) -- (spodd);

			\draw[edge] (spodd) -- (redown3) -- (speven1);

			\draw[edge] (redown2) -- (speven1);
			\draw[edge] (redown1) -- (spodd);

			% bar-3とv-4を点線の楕円で囲む
			% 楕円の中心と半径を計算
			\pgfmathsetmacro{\ellipsecenterx}{0}
			\pgfmathsetmacro{\ellipsecentery}{-\righty-\delta-0.4*\deltay}
			\pgfmathsetmacro{\ellipsewidth}{\deltax*1.7}
			\pgfmathsetmacro{\ellipseheight}{1.3}

			% 点線の楕円を描画
			\draw[groupellipse] (\ellipsecenterx,\ellipsecentery) ellipse (\ellipsewidth cm and \ellipseheight cm);

			% 楕円からv-2への接続線
			% 楕円の上端からv-2へ
			\pgfmathsetmacro{\connectionx}{\ellipsecenterx+ \ellipsewidth}
			\pgfmathsetmacro{\connectiony}{\ellipsecentery}

			% 接続線を描画（少し曲げて自然に見せる）
			\draw[potensialedge] (\connectionx,\connectiony) -- (leaf3);

			% bar-3とv-4を点線の楕円で囲む
			% 楕円の中心と半径を計算
			\pgfmathsetmacro{\ellipsecenterx}{0}
			\pgfmathsetmacro{\ellipsecentery}{\righty+\delta+\deltay-\belowoffsety}
			\pgfmathsetmacro{\ellipsewidth}{\deltax*2.1}
			\pgfmathsetmacro{\ellipseheight}{\deltay*1.5}

			% 点線の楕円を描画
			\draw[groupellipse] (\ellipsecenterx,\ellipsecentery) ellipse (\ellipsewidth cm and \ellipseheight cm);

			% 楕円からv-2への接続線
			% 楕円の上端からv-2へ
			\pgfmathsetmacro{\connectionx}{\ellipsecenterx+ \ellipsewidth}
			\pgfmathsetmacro{\connectiony}{\ellipsecentery}

			% 接続線を描画（少し曲げて自然に見せる）
			\draw[potensialedge] (\connectionx,\connectiony) -- (leaf1);

			%\draw[edge] (leaf2)  .. controls +(0,-5) and +(1,-1) .. (reup3);
			\draw[potensialedge] (leaf2) .. controls +(0,5) and +(3,1) .. (speven2);

		\end{scope}

	\end{tikzpicture}
	\caption{
		Potential neighbors of type A and type B claws in $\LopR$ for $(L, R) \in {S}_k \times {S}_l$.
		Here we assume the claw center belongs to the $L$ side and the leaves belong to the $R$ side; the converse case follows by a similar argument.
		The claw center cannot be connected to any vertices other than its three leaves.
		Dotted lines indicate potential edges from the leaves: a dotted line to a dotted circle means the leaf can be connected to any vertex within that dotted circle.
		The descriptions of the vertices are the same as those in Figure~\ref{fig:claw-configurations-bar}.
	}
	\label{fig:claw-neigbors-ABbar}
\end{figure}

\subsubsection{Structure of extended claws}
We next introduce \emph{extended claws}, which are subgraphs in ${G}[\LopR]$ and natural extensions of claws.
The vertices in an extended claw are denoted by $\cc{n}{t}{c} = \ccc{n}{c} \sqcup \ccl{n}{t}{c} \subset \LopR$ ($t \in \{A, B\}$) and are defined as follows:
\begin{align}
	\ccc{n}{c}    & = \bigsqcup_{i=1}^{n} h_{c^{(i)}},
	\\
	\ccl{n}{t}{c} & = \{\vleafdep{t, 1}_{c}, \vleafdep{t, 2}_{c}\} \sqcup \cclwtotype{n}{c},
	\quad
	\cclwtotype{n}{c} = \bigsqcup_{i=1}^{n} \vleafindep_{c^{(i)}},
\end{align}
where $h_{c^{(1)}} \equiv h_c$ and $h_{c^{(i)}}$ is recursively defined by $h_{c^{(i)}} = \vspo_{c^{(i-1)}}$ for $i=2,\ldots,n$.
Note that the leaves of the claw at $h_c$ are included in $\ccl{n}{t}{c}$: the first component of $\cclwtotype{n}{c}$ is the third leaf $\vleafindep_{c^{(1)}} = \vleafindep_{c}$, which together with the first two vertices comprises the leaves of the claw at $h_c$: $\leaves{c}{t} = \{\vleafdep{t, 1}_{c}, \vleafdep{t, 2}_{c}, \vleafindep_{c}\}$.
More explicitly,
\begin{align}
	\ccc{n}{2j-1}
	 & =
	\{h_{2j-1+4i}\}_{i=0}^{n-1},
	\\
	\ccc{n}{\overline{2j-1}}
	 & =
	\{h_{\overline{2j-1-4i}}\}_{i=0}^{n-1},
	\\
	\cclwtotype{n}{2j-1}
	 & =
	\{h_{\overline{2j+5+4i}}\}_{i=0}^{n-1},
	\\
	\cclwtotype{n}{\overline{2j-1}}
	 & =
	\{h_{2j-7-4i}\}_{i=0}^{n-1}.
\end{align}
Hereafter, we refer to an extended claw $\cc{n}{t}{c}$ as an ``extended claw at $c$''.

For $n\geq 2$, the extended claw can also be defined recursively: $\ccc{n}{c} = \ccc{n-1}{c} \sqcup h_{c^{(n)}}$ and $\cclwtotype{n}{c} = \cclwtotype{n-1}{c} \sqcup \vleafindep_{c^{(n)}}$, thus $\cc{n}{t}{c} = \cc{n-1}{t}{c} \sqcup \{h_{c^{(n)}}, \vleafindep_{c^{(n)}}\}$.
The extended claw $\cc{n}{t}{c}$ for $n=1$ is the usual claw introduced in Section~\ref{sec:structure-claw}.
Note that $\vleafindep_{c^{(n)}}$ is the only possible neighbor of $h_{c^{(n)}}$ from $\cc{n-1}{t}{c} \sqcup \{h_{c^{(n)}}\}$ within the subgraph $\LopR$.

In Figure~\ref{fig:CC-part-one-spodd-plusone-neigbor}, we show the structure of extended claws.
For extended claws at $h_{2j-1}$ and $h_{\overline{2j+7}}$, see Figure~\ref{fig:extended-claw-configuration} and Figure~\ref{fig:extended-claw-configurations-bar}, respectively.

\subsubsection{Potential neighbors of an extended claw}
We next explain the potential neighbors of an extended claw.
The potential neighbors of $\cc{n}{t}{c}$ are (i) the rearrangeable clique $\rearrangeable{c}{-}$, (ii) the rearrangeable clique $\rearrangeable{c^{(n)}}{+}$, (iii) the frozen evens $\{\vfrz{1}_{c^{(n)}}, \vfrz{2}_{c^{(n)}}\}$, and (iv) the special odd $\vspo_{c^{(n)}}$.
In Figure~\ref{fig:CC-part-one-spodd-plusone-neigbor}, we show the possible neighbors of $\cc{n}{t}{c} \subset L\oplus R$ for $n=2$.
The structure of the potential neighbors of the extended claw is similar to that of the original claw.

In the following, we call the neighbors $\{\vfrz{1}_{c^{(n)}}, \vfrz{2}_{c^{(n)}}\}$ the frozen even vertices of the extended claw, and $\vspo_{c^{(n)}}$ the special odd vertex of the extended claw.

The other neighbors of $\cc{n}{t}{c}$ in ${G}$ are forbidden neighbors in $\LopR$, which are denoted by white circles in Figure~\ref{fig:extended-claw-configuration} and Figure~\ref{fig:extended-claw-configurations-bar}.

\subsubsection{Definition of CC-part}
We give the rigorous definition of the CC-part.
$\cc{n}{t}{c} \subset L\oplus R$ is called a CC-part if the neighbors of $\cc{n}{t}{c}$ in ${G}[L\oplus R]$ are contained in $\rearrangeable{c}{-}$ or $\rearrangeable{c^{(n)}}{+}$, if any.

We give a schematic picture of the CC-part with a claw $\cc{n=1}{t}{c}$ in Figure~\ref{fig:CC-part-simple}, and the CC-part with an extended claw $\cc{n=2}{t}{c}$ in Figure~\ref{fig:CC-part-one-spodd-plusone}.

As we will see in Appendix~\ref{app:exhausiveness}, when the extended claw is connected to its frozen even vertices or special odd vertex, the extended claw preserves the balance between the number of vertices from $L$ and $R$ in the connected component where it is included.
In such cases, we can attribute the cancellation of $h_L h_R$ to other BOE components or CC-parts in $L \oplus R$ that can break this balance.
The condition that $\abs{{E}_{L \oplus R}}$ is odd assures the existence of such components.

\begin{figure}[tbp]
	\centering

	\begin{tikzpicture}[triangular lattice small]

		\def\captionx{5}
		\def\captiony{-4.5}

		\begin{scope}[shift={(0,0)}]

			\node at (\captionx,\captiony) {{type A extended claw $\cc{2}{A}{2j-1}$}};

			\pgfmathsetmacro{\offset}{-2}

			\frustrationgraph{-3}{11}{\offset}

			% rearrangeable clique
			% left
			\placevertex{-3}{\offset}{\yrowone}{rearrangeablevertex}
			\placevertex{-2}{\offset}{\yrowtwo}{rearrangeablevertex}
			\placevertexbar{-1}{\offset}{\yrowthree}{rearrangeablevertex}
			% right
			\placevertex{11}{\offset}{\yrowone}{rearrangeablevertex}
			\placevertex{12}{\offset}{\yrowtwo}{rearrangeablevertex}
			\placevertexbar{13}{\offset}{\yrowthree}{rearrangeablevertex}

			% frozen evens
			\placevertex{8}{\offset}{\yrowtwo}{frozeneven}
			\placevertex{10}{\offset}{\yrowtwo}{frozeneven}

			% other evens
			\placevertex{2}{\offset}{\yrowtwo}{vertex}
			\placevertex{4}{\offset}{\yrowtwo}{vertex}
			\placevertex{6}{\offset}{\yrowtwo}{vertex}

			% special odd
			\placevertex{9}{\offset}{\yrowone}{specialodd}

			\placevertex{5}{\offset}{\yrowone}{clawcenter}
			\placevertexbar{11}{\offset}{\yrowthree}{clawleaves}

			% Place vertices in row 1 (top)
			\placevertex{1}{\offset}{\yrowone}{clawcenter, clawcenter}

			% Place vertices in row 2 (middle)
			\placevertex{0}{\offset}{\yrowtwo}{clawleaves}

			% Place vertices in row 3 (bottom)
			\placevertexbar{3}{\offset}{\yrowthree}{clawleaves}
			\placevertexbar{7}{\offset}{\yrowthree}{clawleaves}

			\draw[clawedge] (v1) to (v0);
			\draw[clawedge, curved] (v1) to (bar3);
			\draw[clawedge, curved] (v1) to (bar7);

			\draw[highlightededge, curved] (v5) to (bar7);
			\draw[highlightededge, curved] (v5) to (bar11);

		\end{scope}

		\begin{scope}[shift={(0,-7)}]

			\node at (\captionx,\captiony) {{type B extended claw $\cc{2}{B}{2j-1}$}};

			\pgfmathsetmacro{\offset}{-2}

			\frustrationgraph{-3}{11}{\offset}

			% rearrangeable clique
			% left
			\placevertex{-3}{\offset}{\yrowone}{rearrangeablevertex}
			\placevertex{-2}{\offset}{\yrowtwo}{rearrangeablevertex}
			\placevertexbar{-1}{\offset}{\yrowthree}{rearrangeablevertex}
			% right
			\placevertex{11}{\offset}{\yrowone}{rearrangeablevertex}
			\placevertex{12}{\offset}{\yrowtwo}{rearrangeablevertex}
			\placevertexbar{13}{\offset}{\yrowthree}{rearrangeablevertex}

			% frozen evens
			\placevertex{8}{\offset}{\yrowtwo}{frozeneven}
			\placevertex{10}{\offset}{\yrowtwo}{frozeneven}

			% other evens
			\placevertex{0}{\offset}{\yrowtwo}{vertex}
			\placevertex{4}{\offset}{\yrowtwo}{vertex}
			\placevertex{6}{\offset}{\yrowtwo}{vertex}

			% special odd
			\placevertex{9}{\offset}{\yrowone}{specialodd}

			\placevertex{5}{\offset}{\yrowone}{clawcenter}
			\placevertexbar{11}{\offset}{\yrowthree}{clawleaves}

			% Place vertices in row 1 (top)
			\placevertex{1}{\offset}{\yrowone}{clawcenter, clawcenter}

			% Place vertices in row 2 (middle)
			\placevertex{-1}{\offset}{\yrowone}{clawleaves}

			% Place vertices in row 3 (bottom)
			\placevertex{2}{\offset}{\yrowtwo}{clawleaves}
			\placevertexbar{7}{\offset}{\yrowthree}{clawleaves}

			\draw[clawedge] (v1) to (v-1);
			\draw[clawedge] (v1) to (v2);
			\draw[clawedge, curved] (v1) to (bar7);

			\draw[highlightededge, curved] (v5) to (bar7);
			\draw[highlightededge, curved] (v5) to (bar11);

		\end{scope}

	\end{tikzpicture}

	\caption{Structures of extended claws at the vertex $h_{2j-1}$: $\cc{n=2}{t}{2j-1}$.
		The bold green edges connect the claw center to the claw leaves, and the bold black lines represent the other edges in the extended claws, while gray edges indicate all other edges not present in the extended claws.
		Each vertex labeled with number $i$ represents $h_{2j+i}$ and labeled with $\overline{i}$ represents $\hbar{2j+i}$.
		The blue-filled circle indicates the part $\ccc{2}{2j-1}$, and the green-filled circles indicate the part $\ccl{2}{t}{2j-1}$ in $\cc{n=2}{t}{2j-1}$.
		Solid rectangles represent the frozen even vertices $\vfrz{1}_{2j+3} = h_{2j+6}$ and $\vfrz{2}_{2j+3} = h_{2j+8}$.
		The blue-bordered circle indicates the special odd vertex $\vspo_{2j+3} = h_{2j+7}$.
		The left orange-bordered circles constitute the rearrangeable clique $\rearrangeable{2j-1}{-}$, and the right ones constitute $\rearrangeable{2j+3}{+}$.}

	\label{fig:extended-claw-configuration}
\end{figure}

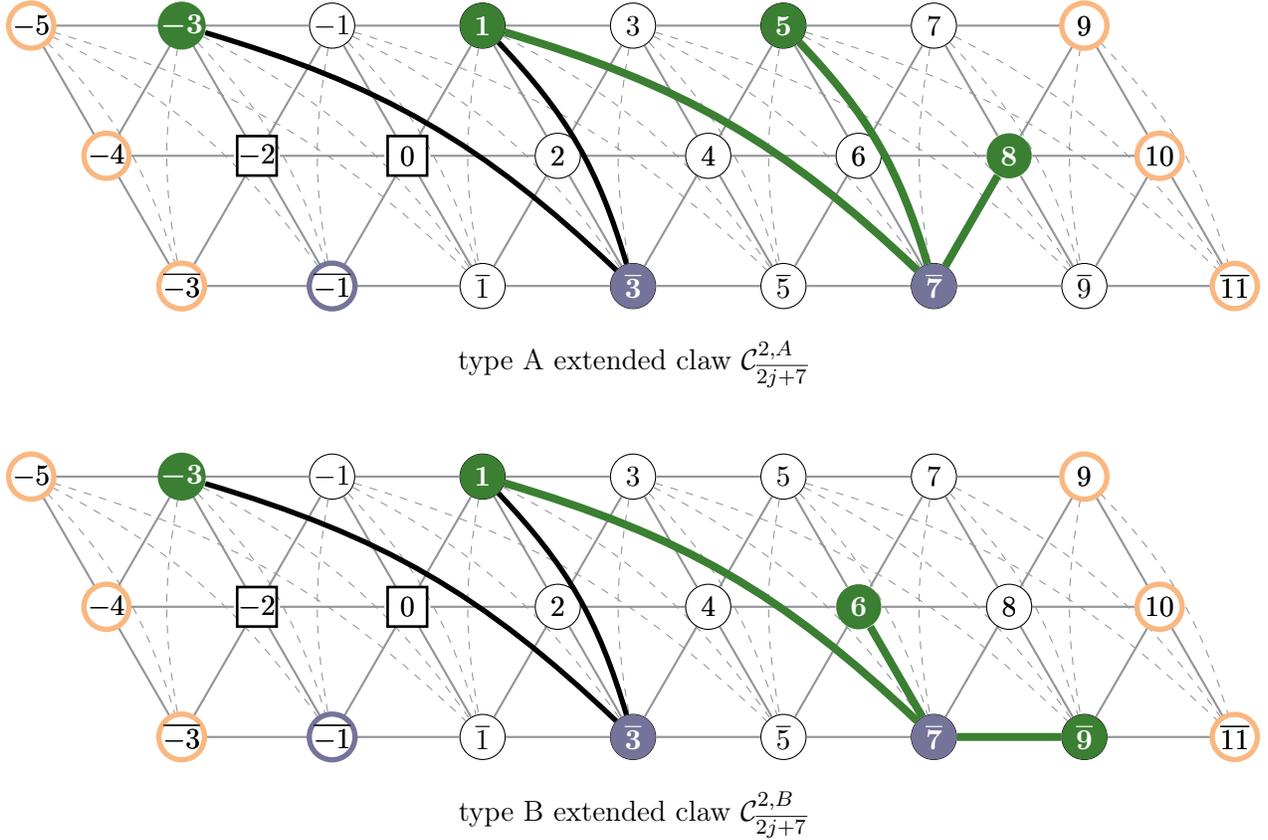
\begin{figure}[tbp]
	\centering
	\begin{tikzpicture}[triangular lattice small]

		\def\captionx{5}
		\def\captiony{-4.5}

		\begin{scope}[shift={(0,0)}]
			\pgfmathsetmacro{\offset}{-2}

			\node at (\captionx,\captiony) {{type A extended claw $\cc{2}{A}{\overline{2j+7}}$}};

			\frustrationgraph{-3}{11}{\offset}

			% rearrangeable clique
			% left
			\placevertex{-3}{\offset}{\yrowone}{rearrangeablevertex}
			\placevertex{-2}{\offset}{\yrowtwo}{rearrangeablevertex}
			\placevertexbar{-1}{\offset}{\yrowthree}{rearrangeablevertex}
			% right
			\placevertex{11}{\offset}{\yrowone}{rearrangeablevertex}
			\placevertex{12}{\offset}{\yrowtwo}{rearrangeablevertex}
			\placevertexbar{13}{\offset}{\yrowthree}{rearrangeablevertex}

			% frozen evens
			\placevertex{0}{\offset}{\yrowtwo}{frozeneven}
			\placevertex{2}{\offset}{\yrowtwo}{frozeneven}

			% other evens
			\placevertex{8}{\offset}{\yrowtwo}{vertex}
			\placevertex{4}{\offset}{\yrowtwo}{vertex}
			\placevertex{6}{\offset}{\yrowtwo}{vertex}

			% special odd
			\placevertexbar{1}{\offset}{\yrowthree}{specialodd}

			% claw
			\placevertexbar{9}{\offset}{\yrowthree}{clawcenter}
			\placevertexbar{5}{\offset}{\yrowthree}{clawcenter}

			\placevertex{3}{\offset}{\yrowone}{clawleaves}
			\placevertex{10}{\offset}{\yrowtwo}{clawleaves}
			\placevertex{7}{\offset}{\yrowone}{clawleaves}
			\placevertex{-1}{\offset}{\yrowone}{clawleaves}

			\draw[clawedge, curved] (v3) to (bar9);
			\draw[clawedge, curved] (v7) to (bar9);
			\draw[clawedge] (v10) to (bar9);
			\draw[highlightededge, curved] (v3) to (bar5);
			\draw[highlightededge, curved] (v-1) to (bar5);

		\end{scope}

		\begin{scope}[shift={(0,-6)}]
			\pgfmathsetmacro{\offset}{-2}

			\node at (\captionx,\captiony) {{type B extended claw $\cc{2}{B}{\overline{2j+7}}$}};

			\frustrationgraph{-3}{11}{\offset}

			% rearrangeable clique
			% left
			\placevertex{-3}{\offset}{\yrowone}{rearrangeablevertex}
			\placevertex{-2}{\offset}{\yrowtwo}{rearrangeablevertex}
			\placevertexbar{-1}{\offset}{\yrowthree}{rearrangeablevertex}
			% right
			\placevertex{11}{\offset}{\yrowone}{rearrangeablevertex}
			\placevertex{12}{\offset}{\yrowtwo}{rearrangeablevertex}
			\placevertexbar{13}{\offset}{\yrowthree}{rearrangeablevertex}

			% frozen evens
			\placevertex{0}{\offset}{\yrowtwo}{frozeneven}
			\placevertex{2}{\offset}{\yrowtwo}{frozeneven}

			% other evens
			\placevertex{10}{\offset}{\yrowtwo}{vertex}
			\placevertex{4}{\offset}{\yrowtwo}{vertex}
			\placevertex{6}{\offset}{\yrowtwo}{vertex}

			% special odd
			\placevertexbar{1}{\offset}{\yrowthree}{specialodd}

			% claw
			\placevertexbar{9}{\offset}{\yrowthree}{clawcenter}
			\placevertexbar{5}{\offset}{\yrowthree}{clawcenter}

			\placevertex{3}{\offset}{\yrowone}{clawleaves}
			\placevertexbar{11}{\offset}{\yrowthree}{clawleaves}
			\placevertex{8}{\offset}{\yrowtwo}{clawleaves}
			\placevertex{-1}{\offset}{\yrowone}{clawleaves}

			\draw[clawedge, curved] (v3) to (bar9);
			\draw[clawedge] (v8) to (bar9);
			\draw[clawedge] (bar11) to (bar9);
			\draw[highlightededge, curved] (v3) to (bar5);
			\draw[highlightededge, curved] (v-1) to (bar5);
		\end{scope}

	\end{tikzpicture}

	\caption{Structures of extended claws at the vertex $\hbar{2j+7}$: $\cc{n=2}{t}{\overline{2j+7}}$.
		The bold green edges connect the claw center to the claw leaves, and the bold black lines represent the other edges in the extended claws, while gray edges indicate all other edges not present in the extended claws.
		Each vertex labeled with number $i$ represents $h_{2j+i}$ and labeled with $\overline{i}$ represents $\hbar{2j+i}$.
		The blue-filled circle indicates the part $\ccc{2}{\overline{2j+7}}$, and the green-filled circles indicate the part $\ccl{2}{t}{\overline{2j+7}}$ in $\cc{n=2}{t}{\overline{2j+7}}$.
		Solid rectangles represent the frozen even vertices $\vfrz{1}_{\overline{2j+3}} = h_{2j}$ and $\vfrz{2}_{\overline{2j+3}} = h_{2j-2}$.
		The blue-bordered circle indicates the special odd vertex $\vspo_{\overline{2j+3}} = \hbar{2j-1}$.
		The left orange-bordered circles constitute the rearrangeable clique $\rearrangeable{\overline{2j+3}}{+}$, and the right ones constitute $\rearrangeable{\overline{2j+7}}{-}$.}
	\label{fig:extended-claw-configurations-bar}

\end{figure}

\begin{figure}[tbp]
	\centering
	\begin{tikzpicture}[
			triangular lattice small,
		]
		\pgfmathsetmacro{\offset}{0}

		\def\rightx{3.6}
		\def\righty{1.2}
		\def\rightyforleaves{0.9}
		\def\delta{1}
		\def\deltay{0.675}
		\def\deltax{0.4}
		\def\belowoffsety{0.3}

		\begin{scope}[shift={(0,0)}]
			\draw[dashed] (0.65*\rightx,-7)--(0.65*\rightx,4.75);

			\node at (0.1*\rightx,4.35) {{\Large Vertices in $L$}};
			\node at (1.2*\rightx,4.35) {{\Large Vertices in $R$}};
			%\draw[dashed] (0.65*\rightx,5)--(0.65*\rightx,6);

			% claw center
			\placevertexwithoutlabel{0}{0}{center}{clawcenter}

			% rearrengeble
			\placevertexwithoutlabel{\deltax}{\righty+\delta}{reup1}{rearrangeablevertex}
			\placevertexwithoutlabel{-\deltax}{\righty+\delta}{reup2}{rearrangeablevertex}
			\placevertexwithoutlabel{0}{\righty+\delta+\deltay}{reup3}{rearrangeablevertex}

			\node[left=3mm of reup3, font=\large] (upperlabel) {$\rearrangeable{c}{-}$};

			% claw leaves
			\placevertexwithoutlabel{\rightx}{\rightyforleaves}{leaf1}{clawleaves}
			\placevertexwithoutlabel{{\rightx}}{0}{leaf2}{clawleaves}
			\placevertexwithoutlabel{\rightx}{{-\rightyforleaves}}{leaf3}{clawleaves}

			\node[right=5.5mm of leaf3, font=\large] (leaf3label) {$\vleafindep_{c}$};
			\draw[->, thick] (leaf3label.west) -- (leaf3);

			% rearrengeble
			%\placevertexwithoutlabel{0}{2*\righty}{reup}{rearrangeablevertex}
			\placevertexwithoutlabel{\rightx}{-3*\righty}{vright}{vertexgreen}

			\placevertexwithoutlabel{0}{-2*\righty}{spodd}{specialodd}
			\node[left=5.5mm of spodd, font=\large] (spoddlabel) {$\vspo_{c} = h_{c^\prime}$};
			\draw[->, thick] (spoddlabel.east) -- (spodd);

			\node[left=5.5mm of center, font=\large] (centerlabel) {Claw center $c$};
			\draw[->, thick] (centerlabel.east) -- (center);

			% rearrengeble
			\placevertexwithoutlabel{\deltax}{-3.*\righty-\delta+\belowoffsety}{redown1}{rearrangeablevertex}
			\placevertexwithoutlabel{-\deltax}{-3.*\righty-\delta+\belowoffsety}{redown2}{rearrangeablevertex}
			\placevertexwithoutlabel{0}{-3*\righty-\delta-\deltay+\belowoffsety}{redown3}{rearrangeablevertex}

			\node[above left=2mm and -6mm of redown2, font=\large] (upperlabel) {$\rearrangeable{c^\prime}{+}$};

			% specials
			\placevertexwithoutlabel{-\deltax}{-3*\righty-\delta-2*\deltay+\belowoffsety}{spodd2}{specialodd}
			\placevertexwithoutlabel{+\deltax}{-3*\righty-\delta-2*\deltay+\belowoffsety}{speven1}{frozeneven}
			\placevertexwithoutlabel{0}{-5.25*\righty-\deltay+\belowoffsety}{speven2}{frozeneven}

			\node[left=10mm of spodd2, font=\large] (spoddlabel2) {$\vspo_{c^\prime}$};
			\draw[->, thick] (spoddlabel2.east) -- (spodd2);

			%\draw[->, thick] (spoddlabel2.west) -- (spodd2);

			\node[right=4mm of speven1, font=\large] (evenlavel) {$\vfrz{2}_{c^\prime}$};
			\draw[->, thick] (evenlavel.west) -- (speven1);

			\node[left=5.5mm of speven2, font=\large] (evenlavel2) {$\vfrz{1}_{c^\prime}$};
			\draw[->, thick] (evenlavel2.east) -- (speven2);

			%\node[above right=1.5mm and -17mm of vright, font=\Large, text width=6cm, align=center] {The only candidate};

			\draw[clawedge] (center) -- (leaf1);
			\draw[clawedge] (center) -- (leaf2);
			\draw[clawedge] (center) -- (leaf3);

			%\draw[potensialedge] (leaf1) -- (reup);
			\draw[edged] (spodd) -- (vright);

			% Define corner points for the leaves dotted rectangle
			\coordinate (leaftri1) at ($(leaf1)+(-0.5,0.5)$);
			\coordinate (leaftri2) at ($(leaf3)+(-0.5,-0.5)$);
			\coordinate (leaftri3) at ($(leaf3)+(0.5,-0.5)$);
			\coordinate (leaftri4) at ($(leaf1)+(0.5,0.5)$);

			% Dotted rectangle for claw leaves
			\draw[dotted, rounded corners=3pt, thick, draw=OliveGreen]
			(leaftri1) -- (leaftri2) -- (leaftri3) -- (leaftri4) -- cycle;

			% Define corner points for upper triangle
			\coordinate (uptri1) at ($(reup2)+(-0.8,-0.5)$);
			\coordinate (uptri2) at ($(reup1)+(0.8,-0.5)$);
			\coordinate (uptri3) at ($(reup3)+(0,1)$);

			% Dotted triangle for reup1, reup2, reup3
			\draw[dotted, rounded corners=20pt, thick]
			(uptri1) -- (uptri2) -- (uptri3) -- cycle;

			% Define corner points for lower region
			\coordinate (lowreg1) at ($(redown2)+(-0.5,0.5)$);
			\coordinate (lowreg2) at ($(redown1)+(0.5,0.5)$);
			\coordinate (lowreg3) at ($(speven1)+(0.5,-0.5)$);
			\coordinate (lowreg4) at ($(spodd2)+(-0.5,-0.5)$);

			% Dotted quadrilateral for lower region
			\draw[dotted, rounded corners=20pt, thick]
			(lowreg1) -- (lowreg2) -- (lowreg3) -- (lowreg4) -- cycle;

			\draw[edge] (reup1) -- (reup2) -- (reup3) -- (reup1);

			\draw[edge] (redown1) -- (redown2) -- (redown3) -- (redown1);

			\draw[edge] (spodd2) -- (speven1) -- (speven2) -- (spodd2);

			\draw[edge] (spodd2) -- (speven1) -- (speven2) -- (spodd2);

			\draw[edge] (spodd2) -- (redown3) -- (speven1);

			\draw[edge] (redown2) -- (spodd2);
			\draw[edge] (redown1) -- (speven1);

			% Potential edges from leaves to points on dotted boundaries
			% From leaf1 to upper triangle boundary (on the right edge)

			\coordinate (leafpoint1) at ($(leaftri1)!0.15!(leaftri4)$);
			\coordinate (uppertriangle) at ($(uptri2)!0.2!(uptri3)$);
			\draw[potensialedge] (leafpoint1) -- (uppertriangle);

			% From leaf3 to lower region boundary (on the top edge)
			\coordinate (lowerregion) at ($(lowreg2)!0.3!(lowreg3)$);
			\draw[potensialedge] (vright) -- (lowerregion);

			% \coordinate (leafpoint3) at ($(leaftri2)!0.8!(leaftri3)$);
			\draw[potensialedge] (leaf3) .. controls +(4,-3.5) and +(3,0) .. (speven2);

			\node[right=3mm of leaf2, font=\large] (leaflabel) {Claw leaves $\leaves{c}{t}$};

			\node[above right =4mm and 1mm of vright,font=\large] (cprimeleaf) {$\vleafindep_{c^\prime}$};
			\draw[->, thick] (cprimeleaf.south west) -- (vright);

			% \coordinate (leafpoint2) at ($(leaftri2)!0.1!(leaftri3)$);
			% \draw[edged] (leafpoint2) -- (spodd);
			\coordinate (leafpoint2) at ($(leaftri2)!0.1!(leaftri3)$);
			\draw[edged] (leaf3) -- (spodd);

		\end{scope}

	\end{tikzpicture}
	\caption{
		Potential neighbors of $\cc{n=2}{t}{c}$.
		The blue-filled vertices belonging to $L$ form $\ccc{2}{c} = \{h_c, \vspo_{c}\}$, where $\vspo_{c} = h_{c^\prime}$, while the green-filled vertices belonging to $R$ form $\ccl{2}{t}{c} = \leaves{c}{t} \sqcup \{\vleafindep_{c^\prime}\}$.
		For the definitions of the symbols (circles and dotted edges etc), see the captions of Figures~\ref{fig:claw-neigbors-AB} and \ref{fig:claw-neigbors-ABbar}.
		When an edge connects to a cluster enclosed by a dotted line, it connects to one specific vertex within that cluster; the exact connection is determined by specifying the claw center and claw type $t \in \{A, B\}$ for $\leaves{c}{t}$, and determined by specifying only the claw center for other dotted circles.
	}
	\label{fig:CC-part-one-spodd-plusone-neigbor}
\end{figure}

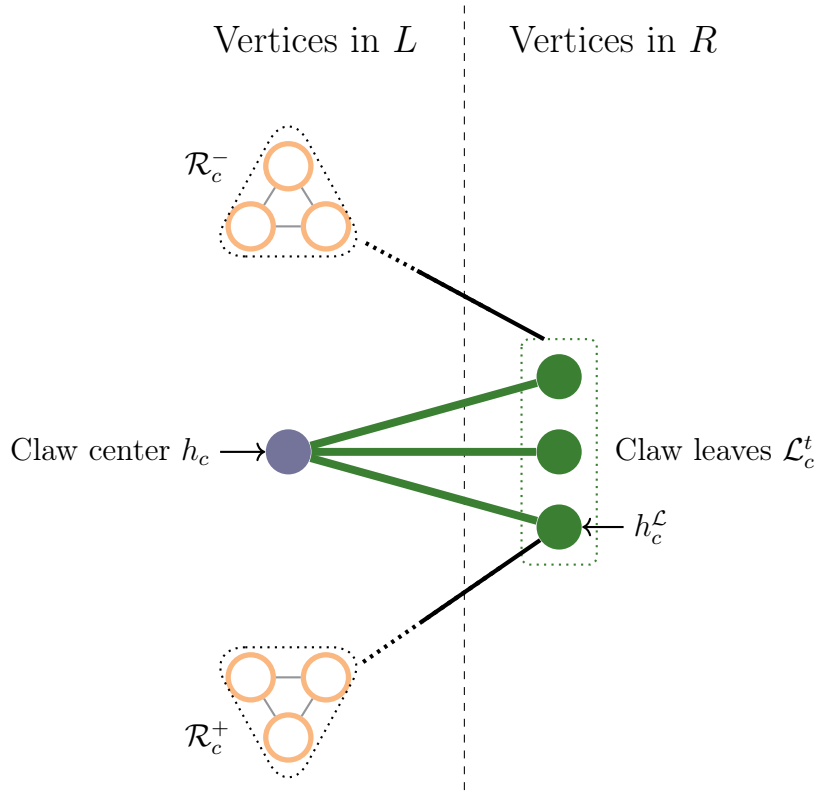
\begin{figure}[tbp]
	\centering
	\begin{tikzpicture}[
			triangular lattice small,
		]
		\pgfmathsetmacro{\offset}{0}

		\def\rightx{3.6}
		\def\righty{2.7}
		\def\rightyleaf{1}
		\def\delta{0.3}
		\def\deltay{0.8}
		\def\deltax{0.5}
		\def\belowoffsety{0.3}

		\begin{scope}[shift={(0,0)}]
			\draw[dashed] (0.65*\rightx,-4.5)--(0.65*\rightx,6);

			\node at (0.1*\rightx,5.5) {{\Large Vertices in $L$}};
			\node at (1.2*\rightx,5.5) {{\Large Vertices in $R$}};

			% claw center
			\placevertexwithoutlabel{0}{0}{center}{clawcenter}
			\node[left=6mm of center, font=\large] (centerlabel) {Claw center $h_c$};
			\draw[->, thick] (centerlabel.east) -- (center);

			% claw leaves
			\placevertexwithoutlabel{\rightx}{\rightyleaf}{leaf1}{clawleaves}
			\placevertexwithoutlabel{{\rightx}}{0}{leaf2}{clawleaves}
			\placevertexwithoutlabel{\rightx}{{-\rightyleaf}}{leaf3}{clawleaves}

			\node[right=5.5mm of leaf3, font=\large] (leaf3label) {$\vleafindep_{c}$};
			\draw[->, thick] (leaf3label.west) -- (leaf3);

			% rearrangeable clique - upper
			\placevertexwithoutlabel{\deltax}{\righty+\delta}{reup1}{rearrangeablevertex}
			\placevertexwithoutlabel{-\deltax}{\righty+\delta}{reup2}{rearrangeablevertex}
			\placevertexwithoutlabel{0}{\righty+\delta+\deltay}{reup3}{rearrangeablevertex}

			% rearrangeable clique - lower
			\placevertexwithoutlabel{\deltax}{-\righty-\delta}{redown1}{rearrangeablevertex}
			\placevertexwithoutlabel{-\deltax}{-\righty-\delta}{redown2}{rearrangeablevertex}
			\placevertexwithoutlabel{0}{-\righty-\delta-\deltay}{redown3}{rearrangeablevertex}

			\draw[clawedge] (center) -- (leaf1);
			\draw[clawedge] (center) -- (leaf2);
			\draw[clawedge] (center) -- (leaf3);

			% Draw triangular edges for rearrangeable clique
			\draw[edge] (reup1) -- (reup2) -- (reup3) -- (reup1);
			\draw[edge] (redown1) -- (redown2) -- (redown3) -- (redown1);

			% Define corner points for the leaves dotted rectangle
			\coordinate (leaftri1) at ($(leaf1)+(-0.5,0.5)$);
			\coordinate (leaftri2) at ($(leaf3)+(-0.5,-0.5)$);
			\coordinate (leaftri3) at ($(leaf3)+(0.5,-0.5)$);
			\coordinate (leaftri4) at ($(leaf1)+(0.5,0.5)$);

			% Dotted rectangle for claw leaves
			\draw[dotted, rounded corners=3pt, thick, draw=OliveGreen]
			(leaftri1) -- (leaftri2) -- (leaftri3) -- (leaftri4) -- cycle;

			% Label for claw leaves with arrow
			\node[right=3mm of leaf2, font=\large] (leaflabel) {Claw leaves $\leaves{c}{t}$};
			%\draw[->, thick] (leaflabel.west) -- ($(leaftri2)+(0.2,0)$);

			% Define corner points for upper rearrangeable triangle
			\coordinate (uptri1) at ($(reup2)+(-0.6,-0.4)$);
			\coordinate (uptri2) at ($(reup1)+(0.6,-0.4)$);
			\coordinate (uptri3) at ($(reup3)+(0,0.8)$);

			% Dotted triangle for upper rearrangeable clique
			\draw[dotted, rounded corners=15pt, thick]
			(uptri1) -- (uptri2) -- (uptri3) -- cycle;

			% Label for upper rearrangeable clique with arrow
			\node[left=3mm of reup3, font=\large] (upperlabel) {$\rearrangeable{c}{-}$};
			%\draw[->, thick] (upperlabel.east) -- ($(reup3)+(-0.3,0)$);

			% Define corner points for lower rearrangeable triangle
			\coordinate (downtri1) at ($(redown2)+(-0.6,0.4)$);
			\coordinate (downtri2) at ($(redown1)+(0.6,0.4)$);
			\coordinate (downtri3) at ($(redown3)+(0,-0.8)$);

			% Dotted triangle for lower rearrangeable clique
			\draw[dotted, rounded corners=15pt, thick]
			(downtri1) -- (downtri2) -- (downtri3) -- cycle;

			% Label for lower rearrangeable clique with arrow
			\node[left=3mm of redown3, font=\large] (lowerlabel) {$\rearrangeable{c}{+}$};
			%\draw[->, thick] (lowerlabel.east) -- ($(redown3)+(-0.3,0)$);

			% Potential edges from leaves to rearrangeable groups
			\coordinate (leafpoint1) at ($(leaftri1)!0.3!(leaftri4)$);
			\coordinate (uppertriangle) at ($(uptri2)!0.1!(uptri3)$);
			\drawpartdotted{leafpoint1}{uppertriangle}

			\coordinate (lowertriangle) at ($(downtri2)!0.1!(downtri3)$);
			% \coordinate (leafpoint3) at ($(leaftri2)!0.3!(leaftri3)$);
			% \drawpartdotted{leafpoint3}{lowertriangle}
			\drawpartdotted{leaf3}{lowertriangle}

		\end{scope}

	\end{tikzpicture}
	\caption{The simplest claw-cancellation (CC) part.
		The blue-filled circle in $L$ indicates the claw center $h_c$.
		The three green-filled circles represent the claw leaves from $h_c$, enclosed by a green dotted rectangle.
		The orange-bordered circles indicate the rearrangeable cliques.
		The solid-dotted lines from the leaves indicate two possibilities: either the leaf connects to one of the rearrangeable clique in the indicated clique, or the leaf has no connections except to its claw center.
		When an edge connects to a cluster enclosed by a dotted line, it connects to one specific vertex within that cluster; the exact connection is determined by specifying the claw center for $\rearrangeable{c}{\pm}$, and determined by specifying the claw center and claw type $t \in \{A, B\}$ for $\leaves{c}{t}$.
	}
	\label{fig:CC-part-simple}
\end{figure}

\begin{figure}[tbp]
	\centering
	\begin{tikzpicture}[
			triangular lattice small,
		]
		\pgfmathsetmacro{\offset}{0}

		\def\rightx{3.6}
		\def\righty{1.2}
		\def\rightyforleaves{0.9}
		\def\delta{1}
		\def\deltay{0.675}
		\def\deltax{0.4}
		\def\belowoffsety{0.3}

		\begin{scope}[shift={(0,0)}]
			\draw[dashed] (0.65*\rightx,-5.5)--(0.65*\rightx,4);

			% \node at (0.1*\rightx,3.5) {{\Large Vertices in $L$}};
			% \node at (1.2*\rightx,3.5) {{\Large Vertices in $R$}};

			% claw center
			\placevertexwithoutlabel{0}{0}{center}{clawcenter}
			\node[left=5.5mm of center, font=\large] (centerlabel) {Claw center $h_{c}$};
			\draw[->, thick] (centerlabel.east) -- (center);

			% rearrangeable clique - upper
			\placevertexwithoutlabel{\deltax}{\righty+\delta}{reup1}{rearrangeablevertex}
			\placevertexwithoutlabel{-\deltax}{\righty+\delta}{reup2}{rearrangeablevertex}
			\placevertexwithoutlabel{0}{\righty+\delta+\deltay}{reup3}{rearrangeablevertex}

			\node[left=3mm of reup3, font=\large] (upperlabel) {$\rearrangeable{c}{-}$};

			% claw leaves
			\placevertexwithoutlabel{\rightx}{\rightyforleaves}{leaf1}{clawleaves}
			\placevertexwithoutlabel{{\rightx}}{0}{leaf2}{clawleaves}
			\placevertexwithoutlabel{\rightx}{{-\rightyforleaves}}{leaf3}{clawleaves}

			\node[right=5.5mm of leaf3, font=\large] (leaf3label) {$\vleafindep_{c}$};
			\draw[->, thick] (leaf3label.west) -- (leaf3);

			% special odd
			\placevertexwithoutlabel{0}{-2*\righty}{spodd}{specialodd}
			\node[left=5.5mm of spodd, font=\large] (spoddlabel) {Special odd $\vspo_c = h_{c^\prime}$};
			\draw[->, thick] (spoddlabel.east) -- (spodd);

			% vertex at right
			\placevertexwithoutlabel{\rightx}{-3*\righty}{vright}{vertexgreen}
			\node[right=5.5mm of vright, font=\large] (cprimeleaf) {$\vleafindep_{c^\prime}$};
			\draw[->, thick] (cprimeleaf.west) -- (vright);

			% rearrangeable clique - lower
			\placevertexwithoutlabel{\deltax}{-3.*\righty-\delta+\belowoffsety}{redown1}{rearrangeablevertex}
			\placevertexwithoutlabel{-\deltax}{-3.*\righty-\delta+\belowoffsety}{redown2}{rearrangeablevertex}
			\placevertexwithoutlabel{0}{-3*\righty-\delta-\deltay+\belowoffsety}{redown3}{rearrangeablevertex}

			\node[left=3mm of redown3, font=\large] (lowerlabel) {$\rearrangeable{c^\prime}{+}$};

			% Draw edges
			\draw[clawedge] (center) -- (leaf1);
			\draw[clawedge] (center) -- (leaf2);
			\draw[clawedge] (center) -- (leaf3);

			\draw[edged] (spodd) -- (vright);

			% Define corner points for the leaves dotted rectangle
			\coordinate (leaftri1) at ($(leaf1)+(-0.5,0.5)$);
			\coordinate (leaftri2) at ($(leaf3)+(-0.5,-0.5)$);
			\coordinate (leaftri3) at ($(leaf3)+(0.5,-0.5)$);
			\coordinate (leaftri4) at ($(leaf1)+(0.5,0.5)$);

			% Dotted rectangle for claw leaves
			\draw[dotted, rounded corners=3pt, thick, draw=OliveGreen]
			(leaftri1) -- (leaftri2) -- (leaftri3) -- (leaftri4) -- cycle;

			% Label for claw leaves
			\node[right=3mm of leaf2, font=\large] (leaflabel) {Claw leaves $\leaves{c}{t}$};

			% Define corner points for upper triangle
			\coordinate (uptri1) at ($(reup2)+(-0.8,-0.5)$);
			\coordinate (uptri2) at ($(reup1)+(0.8,-0.5)$);
			\coordinate (uptri3) at ($(reup3)+(0,1)$);

			% Dotted triangle for upper rearrangeable clique
			\draw[dotted, rounded corners=20pt, thick]
			(uptri1) -- (uptri2) -- (uptri3) -- cycle;

			% Define corner points for lower triangle
			\coordinate (downtri1) at ($(redown2)+(-0.8, 0.5)$);
			\coordinate (downtri2) at ($(redown1)+(0.8,0.5)$);
			\coordinate (downtri3) at ($(redown3)+(0,-1)$);

			% Dotted triangle for lower rearrangeable clique
			\draw[dotted, rounded corners=20pt, thick]
			(downtri1) -- (downtri2) -- (downtri3) -- cycle;

			% Draw triangular edges
			\draw[edge] (reup1) -- (reup2) -- (reup3) -- (reup1);
			\draw[edge] (redown1) -- (redown2) -- (redown3) -- (redown1);

			% Potential edges from leaves to upper triangle
			\coordinate (leafpoint1) at ($(leaftri1)!0.3!(leaftri4)$);
			\coordinate (uppertriangle) at ($(uptri2)!0.2!(uptri3)$);
			\drawpartdotted{leafpoint1}{uppertriangle}

			% Potential edge from vright to lower triangle
			\coordinate (lowertriangle) at ($(downtri2)!0.3!(downtri3)$);
			\drawpartdotted{vright}{lowertriangle}

			% Edge from leaf to special odd
			% \coordinate (leafpoint2) at ($(leaftri2)!0.3!(leaftri3)$);
			%\draw[edged] (leafpoint2) -- (spodd);
			\draw[edged] (leaf3) -- (spodd);

		\end{scope}

	\end{tikzpicture}
	\caption{CC-part of extended claw $\cc{n=2}{t}{c}$.
		For the definitions of the symbols, please refer to Figures~\ref{fig:CC-part-one-spodd-plusone-neigbor} and~\ref{fig:CC-part-simple}.
		The CC-part $\cc{n}{t}{c}$ for general $n$ is obtained by starting from the usual claw $\cc{1}{t}{c}$ and successively adding the pairs $\{h_{c^{(l)}}, \vleafindep_{c^{(l)}}\}$ with $h_{c^{(l)}} = \vspo_{c^{(l-1)}}$ for $l=2,\ldots,n$.}
	\label{fig:CC-part-one-spodd-plusone}
\end{figure}

\subsubsection{Cancellation via claw-cancellation part}
\label{app:cc-part-cancellation}
We prove that if $\LopR$ includes a CC-part, then the term $h_{L} h_{R}$ on the right-hand side of~\eqref{eq:commutator-cancellation} is cancelled by other terms.
We prove the cancellation below for the case where the extended claw in the CC-part is $\cc{n}{t}{2j-1}$.
The case $\cc{n}{t}{\overline{2j-1}}$ follows similarly due to the symmetry of the frustration graph, so we leave this case to the reader.

Consider the case where $\cc{n}{t}{2j-1} \subset \LopR$ is a CC-part and $2j-1$ is an odd integer.
Then there exist four pairs $(L^{t, r}, R^{t, r}) \in \indepset{M}{k} \times \indepset{M}{l}$ for $t \in \{A, B\}$ and $r \in \{1, 2\}$, with $\LopR \in \{L^{t, r} \oplus R^{t, r}\}_{t \in \{A, B\}, r \in \{1, 2\}}$, defined by
\begin{align}
	L^{t,1} & \equiv (L \setminus \ccc{n}{2j-1}) \sqcup \ccc{n}{2j - 1},
	\\
	R^{t,1} & \equiv (R \setminus \ccl{n}{t}{2j-1}) \sqcup \ccl{n}{t}{2j - 1},
	\\
	L^{t,2} & \equiv (L \setminus \ccc{n}{2j-1}) \sqcup \ccl{n}{t}{\overline{2j - 1 + 4n}},
	\\
	R^{t,2} & \equiv (R \setminus \ccl{n}{t}{2j-1}) \sqcup \ccc{n}{\overline{2j - 1 + 4n}}.
\end{align}
The existence of these four extended claws is assured by the fact that extended claws are connected to rearrangeable cliques if any, and are not adjacent to their special odd and frozen even vertices.
For the case where the extended claw is connected to frozen even vertices or special odd vertices, see Appendix~\ref{app:exhausiveness}.

We next show the following identity:
\begin{align}
	\a{2j-1}
	\left(\prod_{l=1}^{n} \C{2j+4l-2}\right)
	=
	\left(\prod_{l=1}^{n} \C{2j+4l-4} \right)
	\a{2j-1+4n}
	\,,
	\label{eq:acc-cca}
\end{align}
which can be proved recursively using the relation $\a{2j-1}  \C{2j+2} =  \C{2j} \a{2j+3}$~\eqref{eq:ac-ca-relation}.

The contribution can be calculated as
\begin{align}
	  &
	\sum_{t \in \{A,B\}} \sum_{r = 1, 2} h_{L^{t, r}} h_{R^{t, r}}
	\nonumber \\
	= &
	\sum_{t \in \{A,B\}}
	h_{L \setminus \ccc{n}{2j-1}}
	\left(
	h_{\ccc{n}{2j-1}} h_{\ccl{n}{t}{2j-1}}
	+
	h_{\ccl{n}{t}{\overline{2j-1+4n}}} h_{\ccc{n}{\overline{2j-1+4n}}}
	\right)
	h_{R \setminus \ccl{n}{t}{2j-1}}
	\nonumber \\
	= &
	h_{L \setminus \ccc{n}{2j-1}}
	\left[
		\a{2j-1}
		h_{\ccc{n}{2j-1}}
		h_{\cclwtotype{n}{2j-1}}
		+
		h_{\cclwtotype{n}{\overline{2j-1+4n}}}
		h_{\ccc{n}{\overline{2j-1+4n}}}
		\a{2j-1+4n}
		% h_{\ccc{n}{2j-1}} h_{\cclwtotype{n}{2j-1}}
		% -
		% \a{2j-1+4n} h_{\ccc{n}{\overline{2j-1+4n}}} h_{\cclwtotype{n}{\overline{2j-1+4n}}}
		\right]
	h_{R \setminus \ccl{n}{t}{2j-1}}
	\nonumber \\
	= &
	(-1)^{n-1}
	h_{L \setminus \ccc{n}{2j-1}}
	\left[
		\a{2j-1}
		\left(\prod_{l=1}^{n} \C{2j+4l-2}\right)
		-
		\left(\prod_{l=1}^{n} \C{2j+4l-4} \right)
		\a{2j-1+4n}
		\right]
	h_{R \setminus \ccl{n}{t}{2j-1}}
	\nonumber \\
	= &
	0,
\end{align}
where in the second equality we have used the fact that $h_{\ccl{n}{t}{\overline{2j-1+4n}}}$ and $h_{\ccc{n}{\overline{2j-1+4n}}}$ anticommute because the number of edges between them is odd ($2n+1$), in the fourth equality, we used the relation $h_{\ccc{n}{2j-1}} h_{\cclwtotype{n}{2j-1}} = (-1)^{n-1}\prod_{l=1}^{n} \C{2j+4l-2}$ and $h_{\cclwtotype{n}{\overline{2j-1+4n}}} h_{\ccc{n}{\overline{2j-1+4n}}} = (-1)^{n} \prod_{l=1}^{n} \C{2j+4l-4}$,  and in the last equality we have used~\eqref{eq:acc-cca}.

\subsection{Exhaustiveness of cancellation mechanisms}
\label{app:exhausiveness}
We next prove that all cancellations in~\eqref{eq:commutator-cancellation} can be explained by BOE components or CC-parts.

What we should consider here is the case where an extended claw in an isolated connected component $\mathcal{O} \subset \LopR$ is connected to its frozen even vertices or special odd vertices.
Below we use the notation $\mathcal{O}_L \equiv \mathcal{O} \cap L$ and $\mathcal{O}_R \equiv \mathcal{O} \cap R$.
After all, such an extended claw does not have the ability to break the balance between $|\mathcal{O}_L|$ and $|\mathcal{O}_R|$.
Thus, we can conclude that if an isolated connected component in $\LopR$ has an imbalance between $|\mathcal{O}_L|$ and $|\mathcal{O}_R|$, then $\mathcal{O}$ must contain a CC-part, which leads to cancellation.
For the case where $\mathcal{O}$ has balanced $|\mathcal{O}_L|$ and $|\mathcal{O}_R|$: if $|{E}_{\mathcal{O}}|$ is odd, the cancellation occurs through BOE components; if $|{E}_{\mathcal{O}}|$ is even, we can attribute the cancellation to other isolated connected components.
Note that there exists at least one isolated connected component with an odd number of edges from the fact that $|{E}_{\LopR}|$ is odd.
This concludes the proof.

We next prove that an isolated connected component $\mathcal{O} \subset \LopR$ containing an extended claw connected to frozen even vertices or special odd vertices cannot break the balance between the number of vertices in $\mathcal{O}_L$ and $\mathcal{O}_R$.

We define an \emph{effective vertex} $\mathcal{E}$ as a subgraph in $\LopR$ satisfying: (i) $|\mathcal{E} \cap L| = |\mathcal{E} \cap R| - 1$, and (ii) the possible neighbors are at most two cliques in $L$.
The same definition applies with the roles of $L$ and $R$ interchanged.
An effective vertex behaves as if it were a single vertex in a path.
If an isolated connected component in $\LopR$ forms an odd path after treating all effective vertices as single vertices, the component is a BOE component and cancels.
An extended claw connected to its frozen even vertices or special odd vertex becomes such an effective vertex when it is not connected to any vertex other than a leaf of the extended claw.
Similarly, an extended claw connected to its frozen even vertices or special odd vertex becomes such an effective vertex when the special odd vertex is not connected to any vertices other than the vertex in the extended claw.
This concludes the proof of mutual commutativity.

%%%%%%%%%%%%%%%%%%
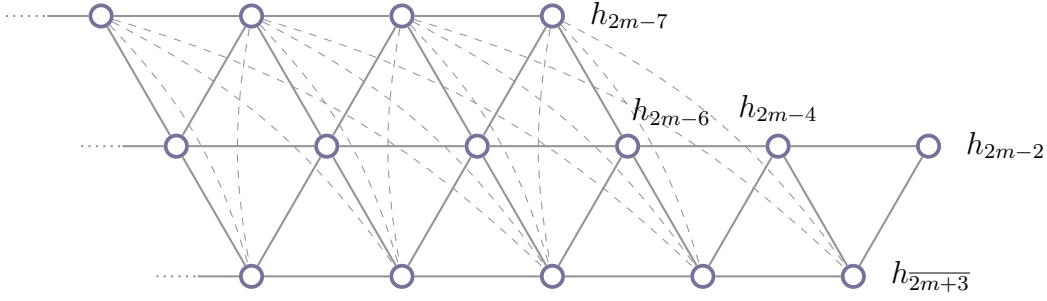
\begin{figure}[tbp]
	\centering
	\begin{tikzpicture}[triangular lattice small]
		% Define common y-coordinates
		\pgfmathsetmacro{\yrowone}{0}
		\pgfmathsetmacro{\yrowtwo}{-1.732}  % -sqrt(3)
		\pgfmathsetmacro{\yrowthree}{-3.464} % -2*sqrt(3)
		\pgfmathsetmacro{\offset}{-2}

		\def\captionx{6.5}
		\def\captiony{-4.425}
		\def\shifty{-6.5}
		\begin{scope}[shift={(0,-12)}]
			\foreach \i in {1,...,4} {
					\pgfmathsetmacro{\v}{int(2*\i-1)}
					\node[vs] (v\v) at (\v, \yrowone) {};
				}

			\foreach \i in {2,...,6} {
					\pgfmathsetmacro{\v}{int(2*\i-1)}
					\node[vs] (bar\v) at (\v, \yrowthree) {};
				}

			% Second row vertices (even numbers)
			\foreach \i in {1,...,6} {
					\pgfmathsetmacro{\v}{int(2*\i)}
					\node[vs] (v\v) at (\v, \yrowtwo) {};
				}

			% Horizontal edges
			\foreach \i in {1,...,6} {
					\pgfmathsetmacro{\inext}{int(\i+2)}
					\draw[edge] (v\i) -- (v\inext);
				}
			\draw[edge] (v8) -- (v10) -- (v12);

			\foreach \i in {1,...,7} {
					\pgfmathsetmacro{\inext}{int(\i+1)}
					\draw[edge] (v\i) -- (v\inext);
				}

			% Diagonal edges
			\foreach \i in {1,...,5} {
					\pgfmathsetmacro{\v}{int(2*\i-1)}
					\pgfmathsetmacro{\vnext}{int(2*\i+1)}
					\pgfmathsetmacro{\even}{int(2*\i)}
					\draw[edge] (v\even) -- (bar\vnext);
				}

			\foreach \i in {2,...,6} {
					\pgfmathsetmacro{\v}{int(2*\i-1)}
					\pgfmathsetmacro{\vnext}{int(2*\i+1)}
					\pgfmathsetmacro{\even}{int(2*\i)}
					\draw[edge] (v\even) -- (bar\v);
				}

			\foreach \i in {2,...,5} {
					\pgfmathsetmacro{\v}{int(2*\i-1)}
					\pgfmathsetmacro{\vnext}{int(2*\i+1)}
					\pgfmathsetmacro{\even}{int(2*\i)}
					\draw[edge] (bar\v) -- (bar\vnext);
				}

			% Curved edges
			\foreach \i in {2,...,4} {
					\pgfmathsetmacro{\v}{int(2*\i-1)}
					\draw[edgefar, curved_right] (v\v) to (bar\v);
				}

			\foreach \i in {1,...,4} {
					\pgfmathsetmacro{\v}{int(2*\i-1)}
					\pgfmathsetmacro{\vnext}{int(2*\i+1)}
					\draw[edgefar] (v\v) to (bar\vnext);
				}

			\foreach \i in {1,...,4} {
					\pgfmathsetmacro{\v}{int(2*\i-1)}
					\pgfmathsetmacro{\vnext}{int(2*\i+3)}
					\draw[edgefar] (v\v) to (bar\vnext);
				}

			\foreach \i in {1,...,3} {
					\pgfmathsetmacro{\v}{int(2*\i-1)}
					\pgfmathsetmacro{\vnext}{int(2*\i+5)}
					\draw[edgefar] (v\v) to (bar\vnext);
				}

			% Left extensions
			\foreach \v in {v1, v2, bar3}{
					\coordinate[left=5.5mm of \v] (l\v);
					\draw[edge] (l\v) -- (\v);
					\coordinate[left=5.5mm of l\v] (ll\v);
					\draw[edge, dotted] (ll\v) -- (l\v);
				}

			% Labels
			\node[right=2mm of v7, font=\large] {$h_{2m-7}$};
			\node[above right=0mm and -2mm of v8, font=\large] {$h_{2m-6}$};
			\node[above = 0mm of v10, font=\large] {$h_{2m-4}$};
			\node[right=2mm of bar11, font=\large] {$\hbar{2m+3}$};
			\node[right=2mm of v12, font=\large] {$h_{2m-2}$};

		\end{scope}
	\end{tikzpicture}
	\caption{
		Free fermionic frustration graph $G_{2m-4}''$.
	}
	\label{fig:frustrationgraph-for-even-Q-recursion}
\end{figure}
%%%%%%%%%%%%%%%%%%
\section{Proof of recursion for charges\label{app:charge-recursion}}
In this appendix, we prove the recursions for the charges~\eqref{eq:m2-charge-recursion-even} and~\eqref{eq:m2-charge-recursion-odd}.

We first prove~\eqref{eq:m2-charge-recursion-even}.
We consider $G = G_{2m}$ and the clique $K = \{h_{2m}, \hbar{2m+1}\}$.
Using~\eqref{eq:general-Q-recursion}, we have
\begin{align}
	\Q{G_{2m}}{k}
	 & =
	\Q{2m-1}{k}
	+
	h_{2m} \Q{2m-3}{k-1}
	+
	\hbar{2m+1} \Q{G_{2m-4}''}{k-1},
	\label{eq:even-Q-rec-app}
\end{align}
where in the third term on the RHS, we define $G_{2m-4}'' \equiv G_{2m} \setminus \Gamma[\hbar{2m+1}]$ and use the fact that
\begin{align}
	G_{2m} \setminus K              & = G_{2m-1},
	\\
	G_{2m} \setminus \Gamma[h_{2m}] & = G_{2m-3}.
\end{align}
Note that $G_{2m-4}''$ is a free fermionic frustration graph, where we can confirm that there is no contradiction with the relation~\eqref{eq:ac-ca-relation}.

The last term on the RHS of~\eqref{eq:even-Q-rec-app} requires more careful analysis.
We analyze the structure of $G_{2m-4}''$ in Figure~\ref{fig:frustrationgraph-for-even-Q-recursion}.
Applying~\eqref{eq:general-Q-recursion} with $G = G_{2m-4}''$ and $K = \{h_{2m-2}\}$, we have
\begin{align}
	\Q{G_{2m-4}''}{k-1}
	=
	\Q{G_{2m-4}'}{k-1}
	+
	h_{2m-2}
	\Q{2m-6}{k-2},
	\label{eq:recursion-from-q2mprime-4}
\end{align}
where we use the facts that $G_{2m-4}'' \setminus \{h_{2m-2}\} = G_{2m-4}'$ and $G_{2m-4}'' \setminus \Gamma[h_{2m-2}] = G_{2m-6}$.

Next, we consider the recursion~\eqref{eq:general-Q-recursion} for $G = G_{2m-4}$ and $K = \{h_{2m-5}\}$:
\begin{align}
	\Q{2m-4}{k-1}
	 & =
	\Q{G_{2m-4}'}{k-1}
	+
	h_{2m-5} \Q{2m-8}{k-2},
	\label{eq:recursion-from-q2m-4}
\end{align}
where we use the facts that $G_{2m-4} \setminus \Gamma[h_{2m-5}] = G_{2m-8}$ and $G_{2m-4} \setminus \{h_{2m-5}\} = G_{2m-4}'$.

Combining~\eqref{eq:recursion-from-q2m-4} and~\eqref{eq:recursion-from-q2mprime-4}, we obtain
\begin{align}
	\Q{G_{2m-4}''}{k-1}
	=
	\Q{2m-4}{k-1}
	+
	h_{2m-2} \Q{2m-6}{k-2}
	-
	h_{2m-5} \Q{2m-8}{k-2}.
	\label{eq:recursion-from-q2mprime-4-2}
\end{align}
Substituting~\eqref{eq:recursion-from-q2mprime-4-2} into~\eqref{eq:even-Q-rec-app}, we obtain the desired recursion~\eqref{eq:m2-charge-recursion-even}.

The proof of the charge recursion~\eqref{eq:m2-charge-recursion-odd} follows similarly but more simply.

\section{Proof of recursion for polynomial\label{app:polyproof}}
In this appendix, we prove Theorem~\ref{thm:poly-recursion}: $T_G(u) T_G(-u)$ is proportional to the identity operator multiplied by the polynomial $P_G(u^2)$~\eqref{eq:poly}, and derive the recursion relations for the polynomial~\eqref{eq:m2-poly-recursion-even} and~\eqref{eq:m2-poly-recursion-odd}.

We first introduce the notation of the transfer matrix for the pseudo-charges~\eqref{eq:pseudo-charges} of a general frustration graph $G$:
\begin{align}
	\label{eq:pseudo-transfer-mat}
	T_{G}(u) \equiv \sum_{k=0}^{\alpha_{G}} (-u)^k \Q{G}{k},
\end{align}
where $u$ is the spectral parameter and $\alpha_{G}$ is the \emph{independence number}, which is the order of the largest independent set in $G$.
Here, an independent set is a subset of vertices with no edges between them.
For $G = G_{M}$, the independence number is $\alpha_{G_{M}} = S_M$.
We then introduce the following quantity for free fermionic frustration graphs $G$ (Figure~\ref{fig:free-fermionic-frustration-graph}):
\begin{align}
	P_{G} (u^2) = T_{G}(u)T_{G}(-u)
	\,.
\end{align}

We first define some terminology regarding even holes in the frustration graph.
These definitions are referred to in~\cite{unified-graph-th}.
We denote by $\mathcal{C}^{(\text{even})}_{G}$ the set of all even holes in $G$.
Two even holes $C$ and $C'$ are said to be \emph{compatible} if they share no vertices and have no edges between them, i.e., if the induced subgraph $G[C \cup C']$ is disconnected with components $G[C]$ and $G[C']$.
Let $\mathscr{C}^{(\text{even})}_{G}$ denote the collection of all subsets of $\mathcal{C}^{(\text{even})}_{G}$ whose elements are pairwise compatible.
Each element $\mathcal{X} \in \mathscr{C}^{(\text{even})}_{G}$ represents a set of mutually compatible even holes.
For such a subset $\mathcal{X}$, we define $\partial \mathcal{X} = \bigcup_{C \in \mathcal{X}} C$ as the union of all vertices contained in the even holes of $\mathcal{X}$.
The cardinality $|\mathcal{X}|$ counts the number of even holes in $\mathcal{X}$, while $|\partial \mathcal{X}|$ denotes the total number of vertices across all holes in $\mathcal{X}$.

Consider $L, R \in \indepsetwhole{G}$, and even holes $C, C' \in L \oplus R$.
Here, $C$ and $C'$ must be compatible, because otherwise it would contradict the condition that $L$ and $R$ are independent sets, as is clear from the structure of the even hole in Figure~\ref{fig:even-hole}.

From the mutual commutativity of the charges in Theorem~\ref{thm:mutual-commutativity}, we have
\begin{align}
	P_{G}(u^2)
	 & =
	\sum_{\substack{s,t = 0                    \\ s+t = 0 (\bmod 2)}}^{\alpha(G)} (-1)^{s}u^{s+t} \Q{G}{s} \Q{G}{t}
	\nonumber                                  \\
	 & =
	\sum_{\substack{S, T \in \indepsetwhole{G} \\  \abs{S} + \abs{T} = 0 (\bmod 2)}} (-1)^{\abs{S}} u^{\abs{S} + \abs{T}} \qty(h_{S \cap T})^2 h_{S \setminus T} h_{T \setminus S}
	\,.
\end{align}
Note the following relation:
\begin{align}
	h_{S \setminus T} h_{T \setminus S}
	=
	(-1)^{\abs{E_{S \oplus T}}}
	h_{T \setminus S} h_{S \setminus T}
	\,.
\end{align}
Then we have
\begin{align}
	P_{G}(u^2)
	=
	\sum_{\substack{S, T \in \indepsetwhole{G} \\  \abs{S} + \abs{T} = 0 (\bmod 2) \\ E_{S\oplus T} = 0 (\bmod 2)}} (-1)^{\abs{S}} u^{\abs{S} + \abs{T}} \qty(h_{S \cap T})^2 h_{S \setminus T} h_{T \setminus S}
	\,.
\end{align}
Using the same argument as in Lemma 11 of~\cite{unified-graph-th} and the proof in Appendix~\ref{app:mutual-commutativity}, we can prove
\begin{align}
	P_{G}(u^2)
	 & = \sum_{\substack{S, T \in \indepsetwhole{G} \\ S \oplus T = \partial \mathcal{X} \\  \mathcal{X} \in \mathscr{C}_G^{\mathrm{even}}}} (-u^2)^{\abs{S}} h_{S} h_{T}
	\nonumber                                       \\
	 & =
	\sum_{\substack{S, T \in \indepsetwhole{G}      \\ S \oplus T = \partial \mathcal{X} \\  \mathcal{X} \in \mathscr{C}_G^{\mathrm{even}}}} (-u^2)^{\abs{S}} \qty(\prod_{\boldsymbol{j} \in S \cap T} b_{\boldsymbol{j}}^2)  \qty(\prod_{C \in \mathcal{X}} \mu_{C})
	\,,
\end{align}
where $\mu_{C} \equiv \mu_{2m-1}$ when the even hole is $C = C_{2m-1} \equiv \{h_{2m-3}, h_{2m}, h_{2m-2}, h_{\overline{2m+1}}\}$.
Note that $\mu_{\mathcal{X}}$, which is the product of the generators in the even hole, is a scalar as explained in~\eqref{eq:mu-scalar}.
From this, we see that $P_G(u^2)$ is a scalar polynomial in $u^2$.
This concludes the proof of~\eqref{eq:poly}.

Then, we obtain the following recursion relation:
\begin{align}
	P_G(u^2) & = P_{G \setminus K}(u^2) + \sum_{\boldsymbol{j} \in K } (-u^2) b_{\boldsymbol{j}}^2  P_{G \setminus \Gamma[\boldsymbol{j}]}(u^2)
	+
	2
	\sum_{\substack{C \in \mathcal{C}^{(\mathrm{even})}_{G}                                                                                     \\ C \cap K \neq \emptyset}} u^{4} \mu_{C} P_{G \setminus \Gamma[C]}(u^2)
	\,,
	\label{eq:general-poly-recursion}
\end{align}
where $G$ is a free fermionic frustration graph and $K$ is a clique in $G$.
The clique $K$ must be chosen such that $G \setminus \Gamma[\boldsymbol{j}]$ for all $\boldsymbol{j} \in K$ and $G \setminus \Gamma[C]$ for all $C \in \mathcal{C}^{(\mathrm{even})}_{G}$ with $C \cap K \neq \emptyset$ are free fermionic frustration graphs.
We define $\Gamma[C] \equiv \bigcup_{\boldsymbol{j} \in C} \Gamma[\boldsymbol{j}]$.
The factor of 2 in the third term on the r.h.s.\ arises from the combination of pairs $(S, T)$ that produce the even hole $C$: if $S, T \in \indepsetwhole{G}$ and $C \subset S \oplus T$ for some even hole $C$, then there exists another configuration $(S', T')$ with $S', T' \in \indepsetwhole{G}$, where $S' \equiv S \oplus C$ and $T' \equiv T \oplus C$, such that $S' \oplus T' = S \oplus T$.
The factor of 2 accounts for this pair.

For $G = G_{2m-1}$ and $K = \{h_{2m-1}\}$, we can apply~\eqref{eq:general-poly-recursion} and have
\begin{align}
	P_{2m-1}(u^2)
	 & =
	P_{2m-2}(u^2)
	-
	u^2 b_{2m-1}^2 P_{2m-4}(u^2)
	\,,
	\label{eq:app-poly-recursion-odd}
\end{align}
where we used the fact that $G_{2m-1} \setminus \Gamma(h_{2m-1}) = G_{2m-4}$.
There are no even hole related terms for this case.
This concludes the proof of~\eqref{eq:m2-poly-recursion-odd}.

For $G = G_{2m}$ and $K = \{h_{2m}, \hbar{2m+1}\}$, we can apply~\eqref{eq:general-poly-recursion} and have
\begin{align}
	P_{2m}(u^2)
	 & =
	P_{2m-1}(u^2)
	-
	u^2 b_{2m}^2 P_{2m-3}(u^2)
	-
	u^2 \bbar{2m+1}^2 P_{G_{2m-4}''}(u^2)
	+
	2 u^4 \mu_{2m-1} P_{2m-6}(u^2)
	\,,
	\label{eq:poly-even-expand}
\end{align}
where the first to third terms can be shown in the same way as in~\eqref{eq:even-Q-rec-app}, and for the last term we used $G_{2m} \setminus \Gamma[\mathcal{X}_{2m-1}] = G_{2m-6}$.
In the same way as~\eqref{eq:recursion-from-q2mprime-4-2}, we can prove the following:
\begin{align}
	P_{G_{2m-4}''}(u^2)
	=
	P_{2m-4}(u^2)
	-
	u^2	b_{2m-2}^2 P_{2m-6}(u^2)
	+
	u^2	b_{2m-5}^2 P_{2m-8}(u^2)
	\,.
	\label{eq:poly-recursion-from-q2mprime-4-2}
\end{align}
Here again, there are no even hole related terms.
Also note that the signs of the second and third terms are different (flipped) from those in~\eqref{eq:recursion-from-q2mprime-4-2}.

Applying~\eqref{eq:app-poly-recursion-odd} to the first and second terms of~\eqref{eq:poly-even-expand} and substituting~\eqref{eq:poly-recursion-from-q2mprime-4-2} into the third term of~\eqref{eq:poly-even-expand}, we have
\begin{align}
	P_{2m}(u^2)
	 & =
	\qty(P_{2m-2}(u^2)	- u^2 b_{2m-1}^2 P_{2m-4}(u^2))
	-
	u^2 b_{2m}^2 \qty(P_{2m-4}(u^2)	- u^2 b_{2m-3}^2 P_{2m-6}(u^2))
	\nonumber       \\
	 & \hspace{3em}
	-
	u^2 \bbar{2m+1}^2 \qty[
		P_{2m-4}(u^2)
		-
		u^2	b_{2m-2}^2 P_{2m-6}(u^2)
		+
		u^2	b_{2m-5}^2 P_{2m-8}(u^2)
	]
	\nonumber       \\
	 & \hspace{4em}
	\pm
	2 u^4 b_{2m-3} b_{2m} b_{2m-2} \bbar{2m+1} P_{2m-6}(u^2)
	\nonumber       \\
	 & =
	P_{2m-2}(u^2)
	-
	u^2 \qty(b_{2m-1}^2 + b_{2m}^2 + \bbar{2m+1}^2) P_{2m-4}(u^2)
	\nonumber       \\
	 & \hspace{4em}
	+
	u^4 \qty(b_{2m-3} b_{2m} \pm b_{2m-2} \bbar{2m+1})^2 P_{2m-6}(u^2)
	-
	u^4	\bbar{2m+1}^2 b_{2m-5}^2 P_{2m-8}(u^2)
	\nonumber       \\
	 & =
	P_{2m-2}(u^2)
	-
	u^2 \s{2m}^2 P_{2m-4}(u^2)
	+
	u^4 \a{2m-1}^2 P_{2m-6}(u^2)
	-
	u^4 \C{2m-2}^2 P_{2m-8}(u^2)
	\,,
	\nonumber
\end{align}
which concludes the proof of~\eqref{eq:m2-poly-recursion-even}.

\section{Proof of the independence number \label{app:proof-Sk}}
In this appendix, we prove that the independence number of the frustration graph is $S_M = \floor{(M+2)/3}$.
Note that $S_M$ is also the degree of the polynomial $P_M(x)$~\eqref{eq:poly}.
For this, it is sufficient to prove that the degree of the polynomial~\eqref{eq:poly} is $S_M = \floor{(M+2)/3}$.

From the recursions~\eqref{eq:m2-poly-recursion-even} and~\eqref{eq:m2-poly-recursion-odd}, we have the recursion for the degrees:
\begin{align}
	S_{2m}   & = \max(S_{2m-2}, S_{2m-4}+1, S_{2m-6}+2, S_{2m-8}+2)
	\label{eq:even-deg}
	\,,
	\\
	S_{2m-1} & = \max(S_{2m-2}, S_{2m-4}+1)
	\label{eq:odd-deg}
	\,.
\end{align}
For~\eqref{eq:even-deg}, we can further simplify as
\begin{align}
	S_{2m}
	 & = \max( \max(S_{2m-2}, S_{2m-4}+1), 1+\max(S_{2m-4}, S_{2m-6}+1))
	\nonumber                                                            \\
	 & =
	\max( S_{2m-1}, S_{2m-3} + 1)
	\,,
\end{align}
where we have used the trivial identity $S_{2m-6} \ge S_{2m-8}$ in the first equality.
Thus, we have the recursion for the degree which does not depend on whether $M$ is even or odd:
\begin{align}
	S_{M} = \max(S_{M-1}, S_{M-3}+1)
	\,,
\end{align}
where the initial conditions are $S_{1} = S_{2} = S_{3} = 1$.

This is the same recursion as in the original FFD case~\cite{fendley-fermions-in-disguise}.
Thus, we can conclude that the independence number here is the same as that of the original FFD case: $S_M = \floor{(M+2)/3}$.
This is the desired result.

\section{Proof of the free fermionic relations}
Here we provide the proofs that are essential for the construction of the fermionic modes in Section \ref{sec:fermion}.
We only show these statements for the right-end simplicial mode that was defined in Subsection \ref{subsec:rightedge}.
We omit the treatment of the left-end simplicial mode in Subsection \ref{subsec:leftedge}, as it is possible to arrive at the same relations after applying similar ideas akin to those presented here to the initial steps of the recursion for the transfer matrix instead.
In Subsection \ref{subsec:mainfermion} we prove the crucial formula \eqref{eq:mainfermion} and in Subsection \ref{subsec:CAR} the canonical anticommutation of the fermions.

\subsection{The commutation relation \label{subsec:mainfermion} }

We would like to prove the master relation
\begin{equation}\label{eq:master}
	T_{M}(u)\left(o_{0}+uo_{1}\right)=\left(o_{0}-uo_{1}\right)T_{M}(u)
\end{equation}
where $o_{0}=\chi_{M},\ o_{1}=\s{M}\chi_{M}$ according to \eqref{eq:krylov}.
This is equivalent to the modified inversion relation \eqref{eq:mainfermion} after multiplying it with $T_M(-u)$ and using the inversion relation \eqref{eq:poly}.
As explained in Section \ref{sec:fermion} this leads to \eqref{eq:eigen} in a straightforward way.
After reordering, we set out to prove
\begin{equation}
	[o_{0},T_{M}(u)]=u\{o_{1},T_{M}(u)\}.\label{eq:commacomm}
\end{equation}

We first consider the odd system size $M=2m-1$.
In this case $o_{0}=\chi_{2m-1},\ o_{1}=h_{2m-1}\chi_{2m-1}$, and using the recursion \eqref{eq:m2-transfermat-recursion-odd} and that $[\chi_{2m-1},T_{2m-k}(u)]=0$ for $k>1$ we have
\begin{equation}\label{eq:prop}
	[o_{0},T_{2m-1}(u)]=2uo_{1}T_{2m-4}(u)
\end{equation}
on the l.h.s. of \eqref{eq:commacomm}.
On the r.h.s. we have
\begin{equation}\label{eq:commTodd}
	\{o_{1},T_{2m-1}(u)\}=\{o_{1},T_{2m-2}(u)\}=2o_{1}T_{2m-4}(u)
\end{equation}
where we first used the recursion~\eqref{eq:m2-transfermat-recursion-odd} for odd system sizes and $\{h_{2m-1}\chi_{2m-1},h_{2m-1}\}=0$; then we used the recursion~\eqref{eq:m2-transfermat-recursion-even} for even system sizes to expand $T_{2m-2}(u)$ and checked that all terms except the first one vanish.

To see this latter fact let us write \eqref{eq:m2-transfermat-recursion-even} as
\begin{equation}
	T_{2m}(u)=T_{2m-2}(u)-u\s{2m}T_{2m-4}(u)+u^{2}\a{2m-1}T_{2m-6}(u)+u^{2}\C{2m-2}T_{2m-8}(u)\label{eq:evenTrec}
\end{equation}
where we substituted $T_{2m-1}$.
Then the statement boils down to showing
\begin{equation}
	\{o_{1},\s{2m-2}\}=\{o_{1},\a{2m-3}\}=\{o_{1},\C{2m-4}\}=0
\end{equation}
with $o_{1}=h_{2m-1}\chi_{2m-1}$.

For even system sizes $M=2m$ our task is a bit more involved.
Now
$o_{0}=\chi_{2m}$ while $o_{1}=\s{2m}\chi_{2m}=(h_{2m}+h_{2m-1}+h_{\overline{2m+1}})\chi_{2m}$.

First, we show that \eqref{eq:evenTrec} simplifies to \eqref{eq:shortrec}.
The prefactors in \eqref{eq:evenTrec} satisfy
\begin{align}
	[\s{2m}, \a{2 m-1}] = \lbrace \s{2 m} , \C{2 m-2} \rbrace = \lbrace \s{2 m}, \a{2 m - 3} \rbrace = \lbrace \s{2 m} , \C{2 m -4} \rbrace = [\s{2 m} , \a{2 m -5 }] = [\s{2 m} ,\C{2 m - 6} ] = 0\,.
\end{align}
Using these and their definitions in~\eqref{eq:def-a} and~\eqref{eq:def-c} the commutator and anticommutator of $T_{2 m -2}(u)$ and $T_{2 m-4}(u)$ with $\s{2 m}$ leads to
\begin{align}
	\lbrace \s{2 m} , T_{2 m -2 }(u) \rbrace & = \lbrace  \s{2 m} , T_{2 m - 4}(u) \rbrace - 2 u \a{2 m - 1} T_{2 m - 6}(u), \label{eq:STacomm} \\
	[\s{2 m}, T_{2 m -4}(u)]                 & = 2 u \C{2 m - 2} T_{2 m-8}(u). \label{eq:STcomm}
\end{align}
Replacing the last two terms on the r.h.s. of \eqref{eq:evenTrec} by \eqref{eq:STacomm} and \eqref{eq:STcomm} respectively, leads to
\begin{equation}
	T_{2 m}(u) = T_{2 m -2}(u) - \frac{u}{2} \lbrace \s{2 m}, T_{2 m -2}(u)\rbrace,
\end{equation}
after collecting the terms.
Then, proving~\eqref{eq:commacomm} for $M=2m$ requires taking the commutator of both sides of~\eqref{eq:shortrec} with $\chi_{2m}$.
Since $[\chi_{2 m}, T_{2 m - k}] = 0$ for $k>1$ one can bring this commutator inside the anticommutator on the r.h.s. of \eqref{eq:shortrec}, giving $[ \chi_{2 m} , \s{2 m}] = - 2 \s{2 m} \chi_{2 m} = -2 o_1$.
That is, we arrive at
\begin{equation}
	[o_0, T_{2 m}(u)] = u \{ o_1, T_{2 m -2}(u) \} = u \{o_1,T_{2m}(u)\},
\end{equation}
where it is a remaining task to prove the last equality.
After applying $\{o_1, \cdot \}$ on \eqref{eq:shortrec} one can indeed see the cancellation of the last term on the r.h.s. as
\begin{align}
	\{o_1, \{\s{2 m} , T_{2m-2}(u)\} \} = -\s{2 m}  \{\s{2 m}, T_{2m-2}(u)\} \chi_M +  \{\s{2 m}, T_{2m-2}(u)\} \s{2 m} \chi_M \nonumber \\
	=   -  \left(\mathbb{S}^2_{2 m} T_{2m-2}(u) +   \s{2 m} T_{2m-2}(u) \s{2 m} \right) \chi_M +  \left( \s{2 m} T_{2m-2}(u) \s{2 m} +  T_{2m-2}(u) \mathbb{S}^2_{2 m}  \right) \chi_M
\end{align}
where we anticommuted and commuted $\chi_M$ through $\s{2m}$ and $T_{2m-2}(u)$, respectively, and used $\mathbb{S}^2_{2m} \propto \mathbb{1}$.

\subsection{\label{subsec:CAR}The canonical anticommutation relations}

For later convenience we introduce $\varphi_{M}(u)\equiv T_{M}(-u)\chi_{M}T_{M}(u)$.
We will study the anticommutator
\begin{equation}
	\{\varphi_{M}(u),\chi_{M}\}=\bar{T}_{M}(-u)T_{M}(u)+T_{M}(-u)\bar{T}_{M}(u)\label{eq:phichiacomm}
\end{equation}
with $\bar{T}_{M}(u)\equiv\chi_{M}T_{M}(u)\chi_{M}$, and show that it is proportional to identity.
We consider the odd and even cases together.
It is possible to rephrase each of the corresponding recursions \eqref{eq:m2-transfermat-recursion-odd} and \eqref{eq:shortrec} as
\begin{align}
	T_{M}(u)       & =T_{M-r_{M}}(u)-\frac{1}{2}[\chi_{M},T_{M}(u)]\chi_{M}\label{eq:Tplus}   \\
	\bar{T}_{M}(u) & =T_{M-r_{M}}(u)+\frac{1}{2}[\chi_{M},T_{M}(u)]\chi_{M}.\label{eq:Tminus}
\end{align}
where $r_{M}=1$ for $M$ odd and $r_{M}=2$ for $M$ even.
Then, after substituting \eqref{eq:Tplus} and \eqref{eq:Tminus} into \eqref{eq:phichiacomm} our anticommutator looks like
\begin{align*}
	\{\varphi_{M}(u),\chi_{M}\} & =2P_{M-r_{M}}(u^{2})\mathbb{1}+\frac{1}{2}[\chi_{M},T_{M}(-u)][\chi_{M},T_{M}(u)]=                                             \\
	                            & =2P_{M-r_{M}}(u^{2})\mathbb{1}+\frac{1}{2}\left(\bar{T}_{M}(-u)T_{M}(u)+T_{M}(-u)\bar{T}_{M}(u)-2P_{M}(u^{2})\mathbb{1}\right)
\end{align*}
where the cross-terms dropped out from the products in \eqref{eq:phichiacomm} explicitly.
Recognizing \eqref{eq:phichiacomm} and reordering leads to
\begin{equation}\label{eq:numeratorpoly}
	\{\varphi_{M}(u),\chi_{M}\}=2\left(2P_{M-r_{M}}(u^{2})-P_{M}(u^{2})\right)\mathbb{1}.
\end{equation}

One can then derive the exchange relation of the fermions and the transfer matrix that reads
\begin{equation}
	\left(1-v/u_{k}\right)\varphi_{M}(u_{k})T_{M}(v)=\left(1+v/u_{k}\right)T_{M}(v)\varphi_{M}(u_{k}).\label{eq:tmpsiflip}
\end{equation}
For this we have to start out as
\begin{align}
	\varphi_{M}(u_{k})T_{M}(v) & =T_{M}(-u_{k})o_{0}T_{M}(v)T_{M}(u_{k})=                                                                         \\
	                           & =T_{M}(-u_{k})\left(T_{M}(v)o_{0}+vT_{M}(v)o_{1}+vo_{1}T_{M}(v)\right)T_{M}(u_{k})=\nonumber                     \\
	                           & =T_{M}(v)\varphi_{M}(u_{k})+vu_{k}^{-1}T_{M}(v)\varphi_{M}(u_{k})+vu_{k}^{-1}\varphi_{M}(u_{k})T_{M}(v)\nonumber
\end{align}
where we used \eqref{eq:master}, the commutation of transfer matrices $ [T_M(u),T_M(v)]=0$ and that
\begin{equation}
	T_{M}(-u_{k})o_{1}T_{M}(u_{k})=\frac{1}{2}[H,\varphi_{M}(u_{k})]=u_{k}^{-1}\varphi_{M}(u_{k}).
\end{equation}
After combining \eqref{eq:numeratorpoly} and \eqref{eq:tmpsiflip} we arrive at
\begin{align}
	\{\varphi_{M}(u_{k}),\varphi_{M}(v)\} & =\varphi_{M}(u_{k})T_{M}(-v)\chi_{M}T_{M}(v)+T_{M}(-v)\chi_{M}T_{M}(v)\varphi_{M}(u_{k})\nonumber \\
	                                      & =\frac{1-v/u_{k}}{1+v/u_{k}}T_{M}(-v)\{\varphi_{M}(u_{k}),\chi_{M}\}T_{M}(v)=\nonumber            \\
	                                      & =4P_{M-r_{M}}(u_{k}^{2})\frac{1-v/u_{k}}{1+v/u_{k}}P_{M}(v^{2})\mathbb{1}.\label{eq:fermionlim}
\end{align}
Due to the $P_{M}(v^{2})$ factor this expression vanishes for any
$v=u_{j},\quad j\ne-k$, while for $v=u_{-k}=-u_{k}$ one has to take
the limit of the r.h.s. in \eqref{eq:fermionlim}.
\begin{equation}
	\lim_{v\to-u_{k}}\{\varphi_{M}(u_{k}),\varphi_{M}(v)\}=16P_{M-r_{M}}(u_{k}^{2})\prod_{j=1:j\ne k}^{S_M}\left(1-u_{k}^{2}/u_{j}^{2}\right)\mathbb{1}.\label{eq:normlim}
\end{equation}
In summary, we have
\begin{align}
	\Psi_{k} & \equiv\frac{T_{M}(-u_{k})\chi_{M}T_{M}(u_{k})}{\mathcal{N}_{k}}, & \{\Psi_{k},\Psi_{k'}\} & =\delta_{k+k'}\mathbb{1},\quad k,k'\neq0
\end{align}
with normalization factors given as the square roots of the r.h.s.
of \eqref{eq:normlim}:
\begin{equation}
	\mathcal{N}_{k}=4\sqrt{P_{M-r_{M}}(u_{k}^{2})\left(-u_{k}^{2}P_{M}'(u_{k}^{2})\right)},
\end{equation}
where we have rewritten the product in \eqref{eq:normlim} using the derivative of the polynomial w.r.t. its argument, denoted by the prime.

When considering the anticommutator of a $\Psi_{k},\ k\neq0$ fermionic mode and the zero mode we may introduce $\mathcal{Q}_{u}\equiv\frac{1}{2}\left(\chi_{M}+\varphi_{M}(u)/P_{M}(u^{2})\right)$ such that $C_{0}\Psi_{0}=\lim_{u\to\infty}\mathcal{Q}_{u}$ and using \eqref{eq:fermionlim} and \eqref{eq:numeratorpoly} it is straightforward to show that this anticommutator vanishes
\begin{equation}
	C_{0}\mathcal{N}_{k}\{\Psi_{k},\Psi_{0}\}=\lim_{u\to\infty}\{\varphi_{M}(u_{k}),\mathcal{Q}_{u}\}=\lim_{u\to\infty}2P_{M-r_{M}}(u_{k}^{2})\left(\frac{1}{1+u/u_{k}}\right)\mathbb{1}=0.
\end{equation}
One can then extract the coefficients $C_{k}$ from \eqref{eq:chidecomp} for $k\neq0$ using \eqref{eq:numeratorpoly} as
\begin{equation}
	k\neq0:\quad\{\varphi_{M}(u_{k}),\chi_{M}\}=\mathcal{N}_{k}C_{k}\mathbb{1}=4P_{M-r_{M}}(u_{k}^{2})\mathbb{1},
\end{equation}
leading to formula \eqref{eq:Ccoeffs}.
We define the normalization of the zero mode through \eqref{eq:zerodef} such that $\Psi_{0}^{2}=\mathbb{1}$.
Since $\varphi_{M}^{2}(u)=P_{M}^{2}(u^{2})\mathbb{1}$ due to the inversion formula $T_M(u)T_M(-u)=P_M(u^2)\mathbb{1}$ and the property $\chi_{M}^{2}=\mathbb{1}$, one can also easily extract $C_{0}$ using
\begin{equation}
	C_{0}^{2}\mathbb{1}=\lim_{u\to\infty}\mathcal{Q}_{u}^{2}=\lim_{u\to\infty}P_{M-r_{M}}(u^{2})/P_{M}(u^{2})\mathbb{1}.
\end{equation}
Note that the limit is zero whenever the polynomial in the numerator has lower order than that in the denominator, i.e. if $S_{M-r_M} < S_M$.
The completeness of the relation \eqref{eq:chidecomp} is guaranteed by the argument using the contour integral in Appendix B of \cite{sajat-ffd-corr}.

\section{Proof of recursion equation for the circuit}
\label{app:proof-of-circuit-recursion}

In this appendix, we give the proof of the recursion relation for the circuit~\eqref{eq:m2-transfermat-recursion-circuit-even}.
For this, we utilize the unitary local gate rather than the non-unitary local gate~\eqref{eq:local-gate}:
\begin{align}
	u_{j}(\phi_{j})=\cos\frac{\phi_{j}}{2}+i\sin\frac{\phi_{j}}{2} b_j^{-1}h_{j},
\end{align}
since the calculation is slightly simpler in this notation.
The translation between these two notations is given by~\cite{sajat-floquet}:
\begin{align}
	A_j g_j(\theta_j)= & u_j(\phi_j)
	\label{eq:gate-translation-with-unitary}
	\,,
\end{align}
where the relation between the two angles $\theta_j$ and $\phi_j$ is
\begin{align}
	\tan(\theta_j/2)= & i\tan(\phi_j/2)
	\label{eq:tan-relation-with-unitary}
	\,.
\end{align}
From~\eqref{eq:tan-relation-with-unitary} we have
\begin{align}
	\cos\theta_j=\frac{1}{\cos\phi_j}
	\,,\quad
	\sin\theta_j=i\frac{\sin\phi_j}{\cos\phi_j}
	\,.
\end{align}
The scalar factor $A_j$ is
\begin{align}
	A_j= & \sqrt{\cos\phi_j}=\frac{1}{\sqrt{\cos\theta_j}}\,.
\end{align}

In the following, we use the abbreviated notations:
\begin{align}
	u_j & \equiv u_j(\phi_j)\,,\quad u_j^{-} \equiv u_j(-\phi_j)
	\,.
\end{align}
We show the useful relations for the unitary local gate below:
\begin{align}
	% u_j u_j^{-}                                                             & = u_j^{-} u_j = 1\,,
	% \\
	% {h}_{j+\delta} u_j                                                & = u_j^{-} {h}_{j+\delta}
	% \quad (\abs{\delta} = 1, 2),
	% \\
	u_j  \hat{h}_{j+\delta} u_j & = \hat{h}_{j+\delta}
	\quad (\abs{\delta} = 1, 2),
	\\
	u_{j}^{2}=                  & \hat{c}_{j} + \hat{h}_{j}
	\,,
\end{align}
where $\hat{h}_j \equiv i\sin\phi_j b_j^{-1} h_j$ and $\hat{c}_{j} \equiv \cos\phi_{j}$.

We define the transfer matrix for the unitary notation:
\begin{align}
	U_{M} \equiv F_{M} \cdot F_{M}^\top\,,
\end{align}
where $F_{M}$ is defined by
\begin{align}
	F_{2k}
	 & =
	(u_2 u_1) (u_4 u_3) \cdots (u_{2k} u_{2k-1})
	\,,
\end{align}
and $F_{2k-1} \equiv \eval{F_{2k}}_{\phi_{2k} = 0}$.
The translation of $\mathcal{V}_{M}$ into $U_M$ is given by
\begin{align}
	\mathcal{V}_{M}
	=
	\qty(\prod_{j=1}^{M} \cos\theta_j)
	U_M\,.
	\label{eq:translate-U-V}
\end{align}
We also introduce another transfer matrix as $\hat{U}_{2k} \equiv \hat{F}_{2k} \cdot \hat{F}_{2k}^\top$ where $\hat{F}_{2k} \equiv \eval{F_{2k}}_{\phi_{2k-1} = 0}$.

In the following, we will derive the recursion relation by expanding the products in the transfer matrix $U_M$.
\begin{align}
	U_{2k}= & (u_{2}u_{1})(u_{4}u_{3})\cdots(u_{2k-2}u_{2k-3})(u_{2k}u_{2k-1})\cdot(u_{2k-1}u_{2k})(u_{2k-3}u_{2k-2})\cdots(u_{3}u_{4})(u_{1}u_{2})
	\nonumber
	\\=&
	\hat{c}_{2k-1}\hat{U}_{2k}+(u_{2}u_{1})(u_{4}u_{3})\cdots(u_{2k-2}u_{2k-3})\hat{h}_{2k-1}(u_{2k-3}u_{2k-2})\cdots(u_{3}u_{4})(u_{1}u_{2})
	\nonumber
	\\=&
	\hat{c}_{2k-1}\hat{U}_{2k}+\hat{h}_{2k-1}U_{2k-4}
	\,.
	\label{eq:U-Uhat-recursion-pre}
\end{align}
$\hat{U}_{2k}$ can also be expanded as
\begin{align}
	\hat{U}_{2k}
	= &
	(u_{2}u_{1})(u_{4}u_{3})\cdots(u_{2k-2}u_{2k-3})u_{2k}^{2}(u_{2k-3}u_{2k-2})\cdots(u_{3}u_{4})(u_{1}u_{2})
	\nonumber                                                                                                                              \\
	= & \hat{c}_{2k}U_{2k-2}+(u_{2}u_{1})(u_{4}u_{3})\cdots(u_{2k-2}u_{2k-3})\hat{h}_{2k} (u_{2k-3}u_{2k-2})\cdots(u_{3}u_{4})(u_{1}u_{2})
	\nonumber                                                                                                                              \\
	= & \hat{c}_{2k}U_{2k-2}+(u_{2}u_{1})(u_{4}u_{3})\cdots(u_{2k-4}u_{2k-5})u_{2k-2}\hat{h}_{2k} (\hat{c}_{2k-3}+\hat{h}_{2k-3} )
	\nonumber                                                                                                                              \\
	  & \hspace{17em} \times u_{2k-2}(u_{2k-5}u_{2k-4})\cdots(u_{3}u_{4})(u_{1}u_{2})
	\nonumber                                                                                                                              \\
	= & \hat{c}_{2k}U_{2k-2}+(u_{2}u_{1})(u_{4}u_{3})\cdots(u_{2k-4}u_{2k-5})\hat{h}_{2k} (\hat{c}_{2k-3}+\hat{h}_{2k-3} u_{2k-2}^{2})
	\nonumber                                                                                                                              \\
	  & \hspace{17em} \times
	(u_{2k-5}u_{2k-4})\cdots(u_{3}u_{4})(u_{1}u_{2})
	\nonumber                                                                                                                              \\
	= & \hat{c}_{2k}U_{2k-2}+\hat{c}_{2k-3}\hat{h}_{2k} U_{2k-4}
	\nonumber                                                                                                                              \\
	  & \qquad
	+\underbrace{(u_{2}u_{1})(u_{4}u_{3})\cdots(u_{2k-4}u_{2k-5})\hat{h}_{2k} \hat{h}_{2k-3} u_{2k-2}^{2}(u_{2k-5}u_{2k-4})\cdots(u_{3}u_{4})(u_{1}u_{2})}_{\qty(\#1)}
	\,,
\end{align}
and the last term can be further expanded as
\begin{align}
	\qty(\#1) = & (u_{2}u_{1})(u_{4}u_{3})\cdots(u_{2k-4}u_{2k-5})\hat{h}_{2k} \hat{h}_{2k-3} (\hat{c}_{2k-2}+\hat{h}_{2k-2} )(u_{2k-5}u_{2k-4})\cdots(u_{1}u_{2})
	\nonumber                                                                                                                                                      \\
	=           & U_{2k-6}\hat{c}_{2k-2}\hat{h}_{2k} \hat{h}_{2k-3}
	+
	\underbrace{(u_{2}u_{1})(u_{4}u_{3})\cdots(u_{2k-4}u_{2k-5})\hat{h}_{2k} \hat{h}_{2k-3} \hat{h}_{2k-2} (u_{2k-5}u_{2k-4})\cdots(u_{1}u_{2})}_{\qty(\#2)}
	\,,
\end{align}
and again, the last term can be expanded as
\begin{align}
	\qty(\#2)
	= & (u_{2}u_{1})(u_{4}u_{3})\cdots(u_{2k-6}u_{2k-7})(u_{2k-4}u_{2k-5})\hat{h}_{2k} \hat{h}_{2k-3} \hat{h}_{2k-2} (u_{2k-5}u_{2k-4})(u_{2k-7}u_{2k-6})\cdots(u_{1}u_{2})
	\nonumber                                                                                                                                                               \\
	= & (u_{2}u_{1})(u_{4}u_{3})\cdots(u_{2k-6}u_{2k-7})u_{2k-4}\hat{h}_{2k} \hat{h}_{2k-3} \hat{h}_{2k-2} u_{2k-4}(u_{2k-7}u_{2k-6})\cdots(u_{1}u_{2})
	\nonumber                                                                                                                                                               \\
	= & \hat{h}_{2k} \hat{h}_{2k-3} \hat{h}_{2k-2} (u_{2}u_{1})(u_{4}u_{3})\cdots(u_{2k-6}u_{2k-7})u_{2k-4}^{2}(u_{2k-7}u_{2k-6})\cdots(u_{1}u_{2})
	\nonumber                                                                                                                                                               \\
	= & \hat{h}_{2k} \hat{h}_{2k-3} \hat{h}_{2k-2} \hat{U}_{2k-4}
	\,.
\end{align}
Then we have
\begin{align}
	\hat{U}_{2k}= & \hat{c}_{2k}U_{2k-2}+\hat{c}_{2k-3}\hat{h}_{2k} U_{2k-4}+\hat{h}_{2k-2} \hat{h}_{2k-3} \hat{h}_{2k} \hat{U}_{2k-4}+\hat{c}_{2k-2}\hat{h}_{2k-3} \hat{h}_{2k} U_{2k-6}
	\,.
	\label{eq:Uhat-U-recursion-pre}
\end{align}
Eliminating $\hat{U}_{2k}$ and $\hat{U}_{2k-4}$ from~\eqref{eq:Uhat-U-recursion-pre} using~\eqref{eq:U-Uhat-recursion-pre}, we have
\begin{align}
	U_{2k}= & \hat{c}_{2k}\hat{c}_{2k-1}U_{2k-2}+(\hat{c}_{2k-1}\hat{c}_{2k-3}\hat{h}_{2k}+\hat{h}_{2k-1}  + \frac{\hat{c}_{2k-1}}{\hat{c}_{2k-5}} \hat{h}_{2k-2}  \hat{h}_{2k-3} \hat{h}_{2k} )U_{2k-4}
	\nonumber                                                                                                                                                                                            \\
	        & \hspace{7em}
	+\hat{c}_{2k-1}\hat{c}_{2k-2}\hat{h}_{2k-3} \hat{h}_{2k} U_{2k-6} + \frac{\hat{c}_{2k-1}}{\hat{c}_{2k-5}}\hat{h}_{2k-5}  \hat{h}_{2k-2} \hat{h}_{2k-3} \hat{h}_{2k} U_{2k-8}
	\,.
\end{align}
Using~\eqref{eq:translate-U-V}, we can derive the recursion for $\mathcal{V}_{2k}$:
\begin{align}
	\mathcal{V}_{2k}
	 & =
	\qty(\prod_{j=1}^{2k} c_j)
	\big[
	{c}^{-1}_{2k}{c}^{-1}_{2k-1}U_{2k-2}
	+
	(
	{c}^{-1}_{2k-1}{c}^{-1}_{2k-3} c^{-1}_{2k} {h'}_{2k}
	+
	c^{-1}_{2k-1} {h'}_{2k-1}
	\nonumber       \\
	 & \hspace{0em}
	+
	{{c}_{2k-5} {c}^{-1}_{2k-1} {c}^{-1}_{2k-2} {c}^{-1}_{2k-3} {c}^{-1}_{2k}} {h'}_{2k-2} {h'}_{2k-3} {h'}_{2k}
	) U_{2k-4}
	+
	{c}^{-1}_{2k-1}{c}^{-1}_{2k-2} {c}^{-1}_{2k-3} {c}^{-1}_{2k} {h'}_{2k-3} {h'}_{2k} U_{2k-6}
	\nonumber       \\
	 & \hspace{1em}
	+
	{c}_{2k-5}{{c}^{-1}_{2k-1}} {c}^{-1}_{2k-5} {c}^{-1}_{2k-2} {c}^{-1}_{2k-3} {c}^{-1}_{2k} {h'}_{2k-5} {h'}_{2k-2} {h'}_{2k-3} {h'}_{2k} U_{2k-8}
	\big]
	\nonumber       \\
	 & =
	\mathcal{V}_{2k-2}
	+
	(
	c_{2k-2} {h'}_{2k}
	+
	c_{2k-3} c_{2k-2} c_{2k} {h'}_{2k-1}
	+
	c_{2k-5} {h'}_{2k-2} {h'}_{2k-3} {h'}_{2k}
	) \mathcal{V}_{2k-4}
	\nonumber       \\
	 & \qquad
	+
	c_{2k-4} c_{2k-5} {h'}_{2k-3} {h'}_{2k} \mathcal{V}_{2k-6}
	+
	c_{2k-4} c_{2k-5} c_{2k-6} c_{2k-7} {h'}_{2k-5} {h'}_{2k-2} {h'}_{2k-3}  {h'}_{2k} \mathcal{V}_{2k-8}
	\,,
	\label{eq:recursion-proof}
\end{align}
where we used the following relation:
\begin{align}
	\hat{h}_{j}
	=
	\frac{1}{\cos\theta_j}
	h_{j}^{\prime}
	\,.
\end{align}
The last line in~\eqref{eq:recursion-proof} is the right-hand side of~\eqref{eq:m2-transfermat-recursion-circuit-even}, which proves the recursion relation~\eqref{eq:m2-transfermat-recursion-circuit-even}.

By substituting $\theta_{2k}=0$ into~\eqref{eq:recursion-proof}, we can also prove~\eqref{eq:m2-transfermat-recursion-circuit-odd}.

\nolinenumbers

\end{document}